\documentclass[11pt,oneside,reqno]{amsart}

\usepackage{ifthen}
\newboolean{showfigures}
\setboolean{showfigures}{true} 

\usepackage[left]{lineno}
\newboolean{linenumbers}
\setboolean{linenumbers}{false}

\usepackage{adjustbox, algorithm, algpseudocode, amsfonts, amssymb, amsthm, amsmath, array, authblk, bm, booktabs, braket, caption, changepage, comment, dsfont, enumitem, etoolbox, fullpage, graphicx, mathrsfs, mathtools, microtype, multicol, multirow, pifont, pythonhighlight, relsize, subfloat, textcomp, url, xcolor, xpatch}

\usepackage{tikz, tikz-cd}
\usetikzlibrary{arrows.meta, positioning, shapes.geometric, decorations.pathreplacing}

\usepackage{subcaption}

\usepackage[superscript,biblabel,nomove]{cite}
\usepackage{bibunits}
\defaultbibliography{refs}
\defaultbibliographystyle{naturemag}

\setlist[enumerate]{itemsep=1ex, topsep=-1ex}
\setlist[itemize]{itemsep=0pt, topsep=0pt}

\usepackage[foot]{amsaddr}

\usepackage{titlesec}

\titleformat{\section}
	{\titlerule\vspace{-2ex}}
	{\thesection.}{1ex}{\centering\scshape}
	[\vspace{.5ex}\titlerule]
	
\titleformat{\subsection}[runin]
  {\normalfont\normalsize\bfseries}{\thesubsection.}{1ex}{}[.]	
\titlespacing*{\subsection}{0pt}{0.0\baselineskip}{0.5ex}

\titleformat{\subsubsection}[runin]
  {\normalfont\normalsize\itshape}{\thesubsubsection.}{1ex}{}	
\titlespacing*{\subsection}{0pt}{0.0\baselineskip}{0.5ex}

\newtheorem{theorem}{Theorem}[section]

\newtheoremstyle{style}
  {\baselineskip} 
  {0em} 
  {\itshape} 
  {} 
  {\bfseries} 
  {.} 
  {.5em} 
  {} 
\theoremstyle{style}

\newtheorem*{theorem*}{Theorem}

\newtheorem*{definition*}{Definition}

\newtheorem*{lemma*}{Lemma}

\newtheorem*{prop*}{Proposition}

\usepackage{hyperref}
\hypersetup{
	colorlinks=true,
	linkcolor=black,
	citecolor=black,
	urlcolor=blue}
	
\usepackage[noabbrev,capitalise]{cleveref}
\creflabelformat{equation}{#2\textup{#1}#3}

\algrenewcommand\algorithmicrequire{\textbf{Input:}}
\algrenewcommand\algorithmicensure{\textbf{Output:}}

\DeclareMathOperator{\1}{\mathds{1}}

\DeclareMathOperator*{\argmin}{arg\,min}

\newcommand{\boldj}{\bm{j}}

\newcommand{\ci}{\perp\!\!\!\!\perp}
\newcommand{\tr}{\mathtt{tr}}

\title{{\large L}{\small ocal graph estimation with pathwise false discovery control}}

\author{Omar Melikechi\textsuperscript{1,$\ast$}}
\author{David B. Dunson\textsuperscript{1}}
\author{Noureddine Melikechi\textsuperscript{2}}
\author{Jeffrey W. Miller\textsuperscript{3}}

\address{\textsuperscript{1}Department of Statistical Science, Duke University, Durham, NC}
\address{\textsuperscript{2}Kennedy College of Sciences, University of Massachusetts Lowell, Lowell, MA}
\address{\textsuperscript{3}Department of Biostatistics, Harvard T.H. Chan School of Public Health, Boston, MA}
\address{\textsuperscript{$\ast$}Corresponding author: omar.melikechi@duke.edu}

\begin{document}

\frenchspacing

\ifthenelse{\boolean{linenumbers}}{\linenumbers}{}

\maketitle

\begin{bibunit}

\textbf{Publication notice.} This version of the paper has been peer-reviewed and published in \textit{Nature Communications}. See \url{https://doi.org/10.1038/s41467-026-72796-9}.

\begin{abstract}
Many datasets include a small set of variables, such as biomarkers or clinical outcomes, whose relationships to the broader system are of primary scientific interest. Estimating the full network of inter-variable relationships in such settings often obscures local structures around these targets, limiting interpretability. To address this fundamental problem, we introduce local graph estimation, a statistical framework for inferring substructures around target variables. We show that traditional graph estimation methods often fail to recover local structure, and present pathwise feature selection (PFS) as an effective alternative. PFS estimates local subgraphs by iteratively applying feature selection and propagating uncertainty along network paths, providing rigorous finite-sample false discovery control even in settings with mixed variable types and nonlinear dependencies. In four distinct applications spanning environmental and public health, multiomics, brain connectomics, and single-nucleus RNA sequencing, PFS recovers interpretable networks consistent with domain knowledge, highlighting its ability to uncover established mechanisms and generate novel hypotheses.
\end{abstract}


\section{Introduction}\label{sec:intro}


Datasets increasingly contain hundreds or thousands of variables, often of different types or \textit{modalities}. In many applications, however, the objective is not to study the full system of inter-variable relationships but to understand a small number of \textit{target variables}, such as biomarkers or clinical outcomes. In these settings, investigators aim to uncover local substructures within the system, such as paths or clusters of covariates, that are meaningfully related to the targets of interest.

Despite substantial advances in high-dimensional and multimodal data analysis in recent years, many approaches are misaligned with local structure learning. Popular dimension reduction algorithms such as iCluster and Multi-Omics Factor Analysis (MOFA) focus on low-dimensional representations that explain sources of variation across the full dataset, but are likely to miss local structures if they account for only a small fraction of overall variation in the data~\cite{icluster,mofa+,multiomics_cantini}. Predictive modeling is another important component of data analysis, but predictive accuracy alone does not provide a principled basis for structure learning~\cite{shmueli,lipton}.

A popular framework that focuses directly on structure learning is \textit{graph estimation}, which represents variables as nodes in a graph and inter-variable relationships as edges between them~\cite{drton,koller}. In this work, we focus on \textit{conditional independence graphs} (CIGs), where edges encode conditional dependencies between variables. Graphical representations are particularly appealing for interpretation, as they formalize direct associations and organize indirect relationships in a way that can be visualized and used to reason about both individual variables and the system as a whole.

Many CIG estimation methods fall into two broad classes. The first comprises regularized estimation approaches, most notably the graphical and nodewise lasso and their variants~\cite{glasso,mb_graph}. These methods scale to thousands of variables, but are prone to estimating dense graphs with many false positives (\cref{fig:limitations}), often assume Gaussian data, and typically perform poorly when inter-variable relationships are nonlinear. The second class consists of methods based on explicit conditional independence testing, such as constraint-based approaches~\cite{bnlearn}. Previous work and results in this manuscript show that popular implementations of many methods in this class scale poorly with dimension~\cite{expensive_wang,expensive_yu,ling2020}. As we will see, these approaches also tend to have relatively low statistical power, especially when the underlying graph contains nonlinear relationships or is moderately dense.

Crucially, while many methods are designed to estimate either the entire graph of inter-variable relationships (global estimation) or immediate neighborhoods of variables (Markov blanket estimation), far less work focuses on recovering extended local structures, such as paths or clusters, around target variables. Markov blanket estimates can in principle be combined to estimate such structures, and prior work has explored extending local procedures to produce global estimates~\cite{mb_graph,gao2017}. However, theoretical results for these approaches often rely on idealized assumptions such as infinite samples~\cite{gao2017} or the ability to perfectly identify conditional independence relations from data~\cite{wang2014}. In particular, they do not provide finite-sample error control at the level of extended local structures, which can lead to unreliable estimates in practice. More broadly, as we will show, global false discovery guarantees do not yield valid control of local errors.

To address these gaps, we propose \textit{local graph estimation}, a statistical framework that formalizes the problem of recovering extended local structures around target variables, and introduce a method, \textit{pathwise feature selection} (PFS), to solve it. Unlike global estimation methods, PFS avoids estimating irrelevant parts of the graph, enabling more accurate local inference around target variables while reducing computational burden in high dimensions. We also prove that PFS provides finite-sample false discovery control at the path level, establishing a highly interpretable framework for pathway discovery. In \cref{fig:limitations}b, for example, our theory states that there is at most a $26\%$ chance that the estimated path $(X_1, X_2, X_{78}, X_{90})$ is not in the true graph, which is the sum of the $q$-values---shown as edge weights---along that path.

We demonstrate local graph estimation and PFS on simulated data and in four diverse applications: an environmental and public health study of cancer; a multiomic breast cancer study; an analysis of brain networks and cognition; and a single-nucleus RNA sequencing study of Alzheimer's disease. In simulations, PFS achieves a favorable balance between true positive and false discovery rates relative to existing approaches across a range of linear, nonlinear, sparse, and dense regimes. In all four applications, PFS yields interpretable local structure that is consistent with established domain knowledge. Together, these studies illustrate that local graph estimation and PFS are promising tools for uncovering meaningful, target-specific structure in complex high-dimensional data.

\ifthenelse{\boolean{showfigures}}{
\begin{figure}[ht!]
\centering
\captionsetup{width=0.95\textwidth}

\begin{subfigure}[t]{\textwidth}
    \centering
    \includegraphics[height=0.4\textheight, width=0.9\textwidth]{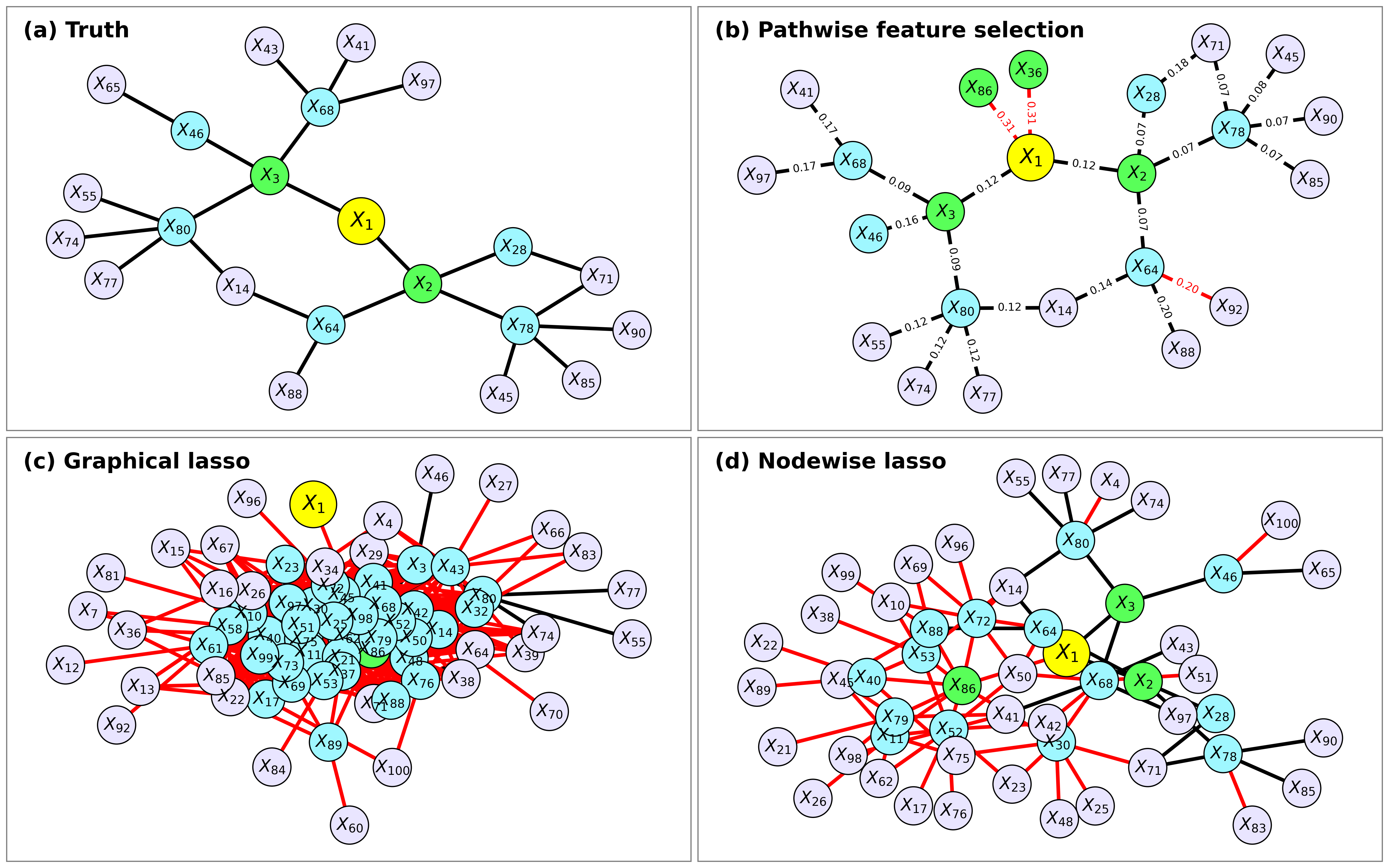}
\end{subfigure}

\begin{subfigure}[t]{\textwidth}
	\centering
    \includegraphics[
    		height=0.31\textheight, 
    		width=0.80\textwidth,
    		trim={0pt 0pt 0pt 10pt}, clip
    	]{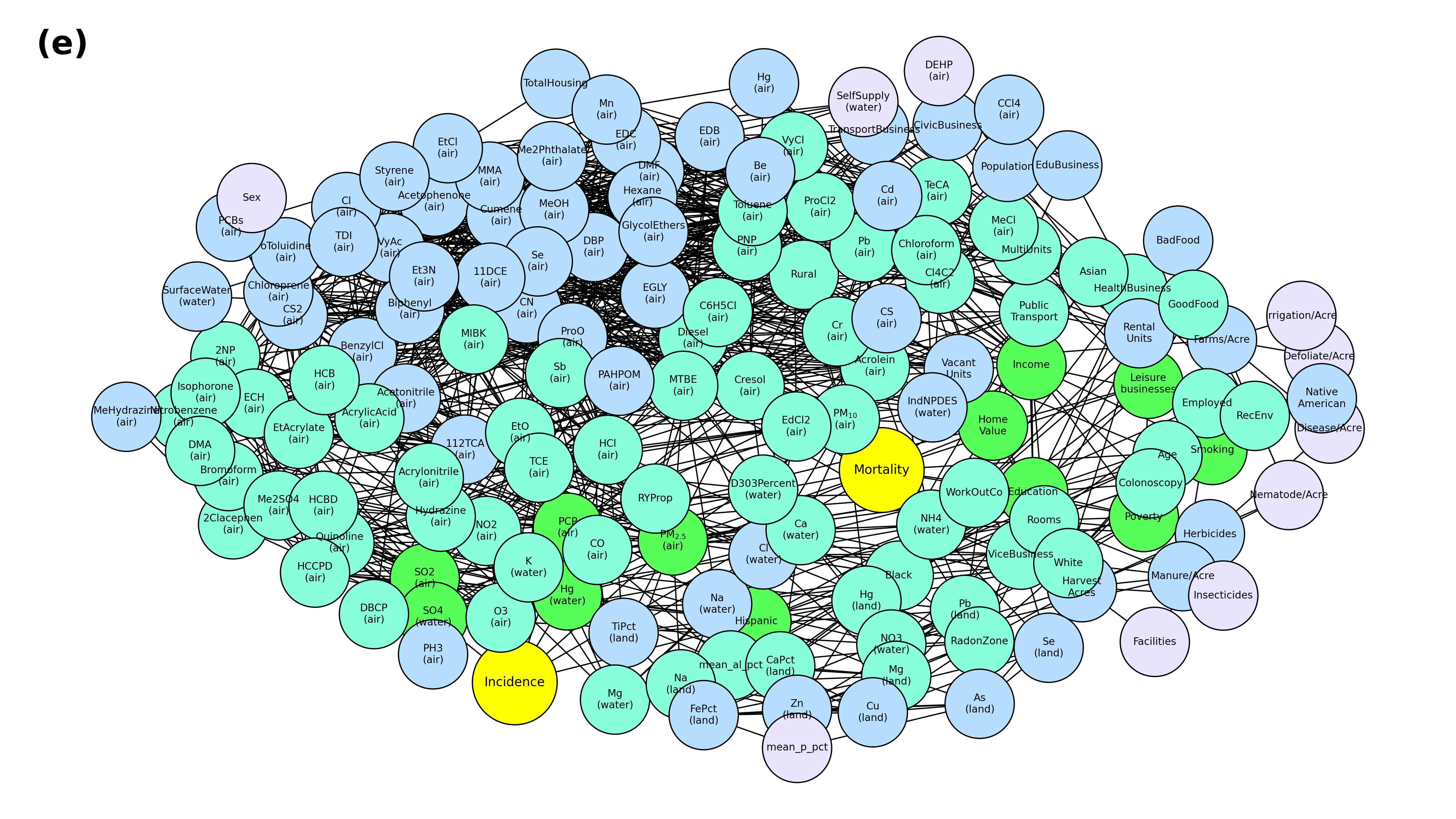}
\end{subfigure}

\vspace*{-1em}

\caption{\textit{Benefits of pathwise false discovery control and limitations of existing methods on simulated and real data.} 
(\textbf{a}) The true local graph of radius $3$ around the target variable $X_1$, shown in yellow. 
(\textbf{b}) Pathwise feature selection identifies $21$ true edges and $3$ false edges. It also provides edge-specific uncertainty quantification in the form of $q$-values---shown as edge weights---with lower $q$-values indicating greater confidence. 
(\textbf{c}) The graphical lasso \cite{glasso} identifies $12$ true edges and $699$ false edges despite the fact that the regularization parameter was tuned to yield the sparsest graph in which $X_1$ has at least one neighbor(\cref{sup_sec:limits_global}). (\textbf{d}) The nodewise lasso \cite{mb_graph} identifies 22 true edges and 57 false edges. All estimates are based on the same $200$ samples drawn from a $p=100$-dimensional Gaussian graphical model. True and false edges are black and red, respectively; green, blue, and purple nodes are distances 1, 2, and 3 from $X_1$ in graphs (a)--(d). (\textbf{e}) Applying the graphical lasso to the environmental health dataset described in \cref{sec:eqi} yields a dense graph that obscures meaningful interactions.}
\label{fig:limitations}
\end{figure}
}{}


\section{Results}\label{sec:results}



\subsection{Local graph estimation}\label{sec:lge}


A graph is a pair $G = (V,E)$, where the set of nodes $V = \{1,\dots,p\}$ indexes variables $X_1,\dots,X_p$, and $E\subseteq V\times V$ denotes the set of edges between pairs of nodes. The terms \textit{variable} and \textit{node} will be used interchangeably, since each variable maps to a unique node. In this work, we focus on \textit{undirected graphs}, where edges $(j,k)$ and $(k,j)$ are identical, but local graph estimation can be considered for directed graphs as well. We assume $G$ is a \textit{conditional independence graph} for $X$, meaning that $(j,k)$ is in $E$ if and only if $j\neq k$ and $X_j$ and $X_k$ are conditionally dependent given all other variables; that is, there is an edge between $j$ and $k$ if and only if $X_j$ and $X_k$ contain information about each other that is not captured by the rest of the system. A unique such graph exists whenever the joint distribution of the variables is strictly positive \cite{koller}. Nodes $j$ and $k$ are called \textit{neighbors} in $G$ if $(j,k)$ is in $E$; the \textit{neighborhood} of a node is the set of all its neighbors. For example, the neighbors of $X_3$ in \cref{fig:limitations}a are $X_1$, $X_{46}$, $X_{68}$, and $X_{80}$. 

Traditionally, graph estimation aims to infer all of $G$ from $n$ independent samples of a random vector $X = (X_1,\dots, X_p)$, but this is often unnecessary when only a subset of \textit{target variables} $V_0\subseteq V$ is of particular interest. In this case, let $B_r(V_0)$ be the set of nodes that are reachable from any node in $V_0$ by a path of length $r$ or less (that is, for every node $j$ in $B_r(V_0)$, there is a sequence of at most $r$ adjacent edges in $G$ connecting $j$ to a node in $V_0$), and let $E_r(V_0)$ be the set of edges $(j,k)$ such that at least one of $j$ or $k$ is in $B_{r-1}(V_0)$. The \textit{local graph} of radius $r$ around $V_0$ is the subgraph $G_r(V_0)=(B_r(V_0),E_r(V_0))$ of $G$ (see \cref{fig:limitations}a for an illustration). Given a set of target variables $V_0$ and a radius $r$, the goal of \textit{local graph estimation} is to estimate $G_r(V_0)$ from samples of $X$. Furthermore, we aim to do so in a way that controls false discoveries, applies to mixed variable types such as discrete and continuous data, and handles complex relationships such as nonlinear interactions. Additional details about the local graph estimation problem are in \cref{sup_sec:lge}.

A natural approach to local graph estimation is to obtain an estimate $\widehat{G}$ of $G$ using one of many existing full graph estimation methods, then estimate the local graph $G_r(V_0)$ by the local graph of radius $r$ around $V_0$ in $\widehat{G}$. This is often implicitly done in practice; for example, whenever one draws conclusions about the neighborhood of a variable based on the full estimated graph $\widehat{G}$. However---in addition to being computationally expensive when $p$ is large (\cref{sup_sec:computation})---results reported in \cref{sup_sec:simulations,sup_sec:other_methods} show that this approach exhibits poor local graph estimation performance across a wide range of methods, including those with global false discovery rate control. In \cref{sup_sec:limitations}, we describe fundamental mathematical limitations of this approach that help explain the poor performance of full graph estimation methods in the context of local graph estimation.


\subsection{Pathwise feature selection}\label{sec:pfs}


To address these limitations, we introduce \textit{pathwise feature selection} (PFS), a method designed for local graph estimation. PFS constructs local graphs around target variables by iteratively applying feature selection and tracking uncertainty along paths in the estimated graph using $q$-values, which are formally defined in \cref{sup_sec:qvalues}. Informally, each $q$-value $q_j(k)$ for $k \neq j$ is a number between $0$ and $1$ that quantifies uncertainty about whether the edge $(j,k)$ is in $G$, with larger values indicating greater uncertainty. In the first iterative step of PFS, the neighborhood of each node $j$ in $V_0$ is estimated by including all nodes $k$ such that $q_j(k)$ falls below a user-specified threshold. Let $S_1(V_0)$ denote the set of newly selected nodes not already in $V_0$. In the next step, each $j$ in $S_1(V_0)$ is treated as a response, and the same procedure is applied to estimate its neighbors, yielding a new layer $S_2(V_0)$ of previously unidentified nodes, and so on. A full description of PFS is given in \cref{alg:pfs} and discussed in the Methods section, with implementation details provided in \cref{sup_sec:implementation}.

In \cref{sup_sec:paths} we state and prove our main theoretical result, \cref{thrm:main}, which says that the sum of $q$-values along an estimated path upper bounds the probability that the path is not in the true graph $G$. This gives a principled stopping rule for PFS: the iterative procedure halts either when the maximum radius $r_{\text{max}}$ is reached or when the cumulative uncertainty along a path---quantified by the sum of its $q$-values---exceeds a threshold $q_{\text{path}}^*$. For example, in \cref{fig:limitations}b, every node in the estimated local graph lies on a path starting at $X_1$ whose $q$-values sum to at most $q_{\text{path}}^*=0.4$. In particular, there is greater uncertainty in the false edges $(1,36)$ and $(1,86)$ (both with $q$-values of $0.31$) than in the true edges $(1,2)$ and $(1,3)$. PFS therefore extends paths beyond $X_1$ and $X_2$, but not beyond $X_{36}$ and $X_{86}$ under this pathwise threshold. Further discussion of \cref{thrm:main} is provided in the Methods section.

PFS offers several practical advantages beyond principled pathwise error control. Users can specify custom $q$-value thresholds for different nodes or node types in order to prioritize features of known importance or emphasize cross-modal connections. Adding to this flexibility, our implementation of PFS uses integrated path stability selection~\cite{ipss,ipss_nonparametric} (IPSS) to estimate $q$-values nonparametrically, enabling detection of both linear and nonlinear relationships without requiring distributional assumptions such as Gaussianity (Methods). IPSS also employs a repeated subsampling procedure that has been shown to improve the robustness and reproducibility of feature selection~\cite{nogueira}. Finally, in \cref{tab:runtimes} we find that PFS with IPSS is significantly faster than other graph estimation methods in high dimensions, including local methods such as those based on Markov blanket estimation. A comparative study of PFS with IPSS and alternative $q$-value constructions is provided in \cref{sup_sec:qvalue_methods}.

\begin{algorithm}
\caption{(Pathwise feature selection)}\label{alg:pfs}
\begin{algorithmic}[1]
\Require Data matrix $\bm{X}\in\mathbb{R}^{n\times p}$, target features $V_0\subseteq V$, maximum radius $r_{\max}$, neighborhood FDR thresholds $\{q_r^* : 1 \leq r \leq r_{\max}\}$, path threshold $q_{\mathrm{path}}^*$, an algorithm for computing $q$-values
\State Initialize current features $S \gets V_0$, visited features $B \gets V_0$, and $q$-value matrix $Q \gets \mathbf{1}\in\mathbb{R}^{p\times p}$
\For{$r = 1$ to $r_{\max}$}
	\For{$j \in S$}
    		\State Compute $q$-values $q_j(k)$ for all $k \in V\setminus\{j\}$
    		\If{$q_j(k) \leq q_r^*$}
    			\State $Q_{jk} \gets Q_{kj} \gets \min\{q_j(k), Q_{kj}\}$
    		\EndIf
    	\EndFor
    	\State Compute $d_Q(V_0,j,r)$ for all $j \in V\setminus B$ (see \cref{eq:lightest_path})
   	\State $S \gets$ $\{j \in V\setminus B : d_Q(V_0,j,r)\leq q_{\mathrm{path}}^*\}$
    \State $B \gets B \cup S$
\EndFor
\Ensure Weighted adjacency matrix $Q$ (see \cref{sup_sec:implementation} for details)
\end{algorithmic}
\end{algorithm}


\subsection{Simulation studies}\label{sec:simulations}


We conducted simulation studies to evaluate local graph estimation performance when the true graph is known, varying sparsity, sample size, and whether relationships between the target variable and its neighbors were linear or nonlinear. We compared PFS to the graphical and nodewise lasso, five regularization-based methods with global false discovery rate (FDR) control, and seven constraint- or mutual information-based methods, six of which have both local and global variants. Descriptions and implementation details for each of these methods are in \cref{sup_sec:other_method_details}. Local graph recovery at varying radii around the target variable was assessed using true positive rate (TPR) and FDR (\cref{eq:tpr,eq:fdp}). Full simulation settings and results are reported in \cref{sup_sec:simulations}.

Simulation results in Tables \ref{tab:linear_sparse_n100}--\ref{tab:nonlinear_dense_n500} show that PFS achieves a favorable balance between TPR and FDR across all experiments. One reason for this is that PFS is the only method with rigorous pathwise false discovery control, which enables precise inference at target variables while preventing errors from propagating away from them. In linear settings, graphical and nodewise lasso recover many true edges but suffer from high FDRs, while several global FDR and constraint-based methods are overly conservative, yielding low FDRs at the cost of substantially reduced TPRs. PFS performs especially well in nonlinear settings, achieving the highest TPR among all methods at radii 1, 2, and 3 for both sparse and dense graphs. Furthermore, its consistent balance between discovery and error control across the dense settings indicates that PFS neither incurs excessive false negatives nor depends on sparsity for good performance. Finally, \cref{fig:varying_n} shows the performance of PFS at different radii as $n$ increases in the sparse, linear setting. With $p$ fixed at $200$, PFS maintains FDR control while TPR steadily increases, approaching one once $n>p$.


\subsection{Environmental and social drivers of cancer}\label{sec:eqi}


Mounting evidence suggests that human health is shaped by interacting variables across multiple environmental, socioeconomic, and demographic domains \cite{exposome_maitre_early_life,exposome_wishart,exposome_young}. With this in mind, we compiled county-level data from the Environmental Protection Agency (EPA), U.S. Census Bureau, National Cancer Institute (NCI), and the Centers for Disease Control and Prevention (CDC) to investigate cancer outcomes across the contiguous United States. After data cleaning (\cref{sup_sec:eqi_cleaning}), the final dataset includes $p = 165$ variables measured across $n = 2857$ counties. Our two target variables are age-adjusted incidence and mortality rates for all cancers, obtained from the NCI and CDC as aggregate values spanning 2017-2022.

\ifthenelse{\boolean{showfigures}}{
\begin{figure*}[ht!]
\begin{adjustbox}{center}
\includegraphics[
	height=0.4\textheight, 
	width=\textwidth,
	trim={60pt 27.5pt 60pt 27.5pt}, clip  
	]
	{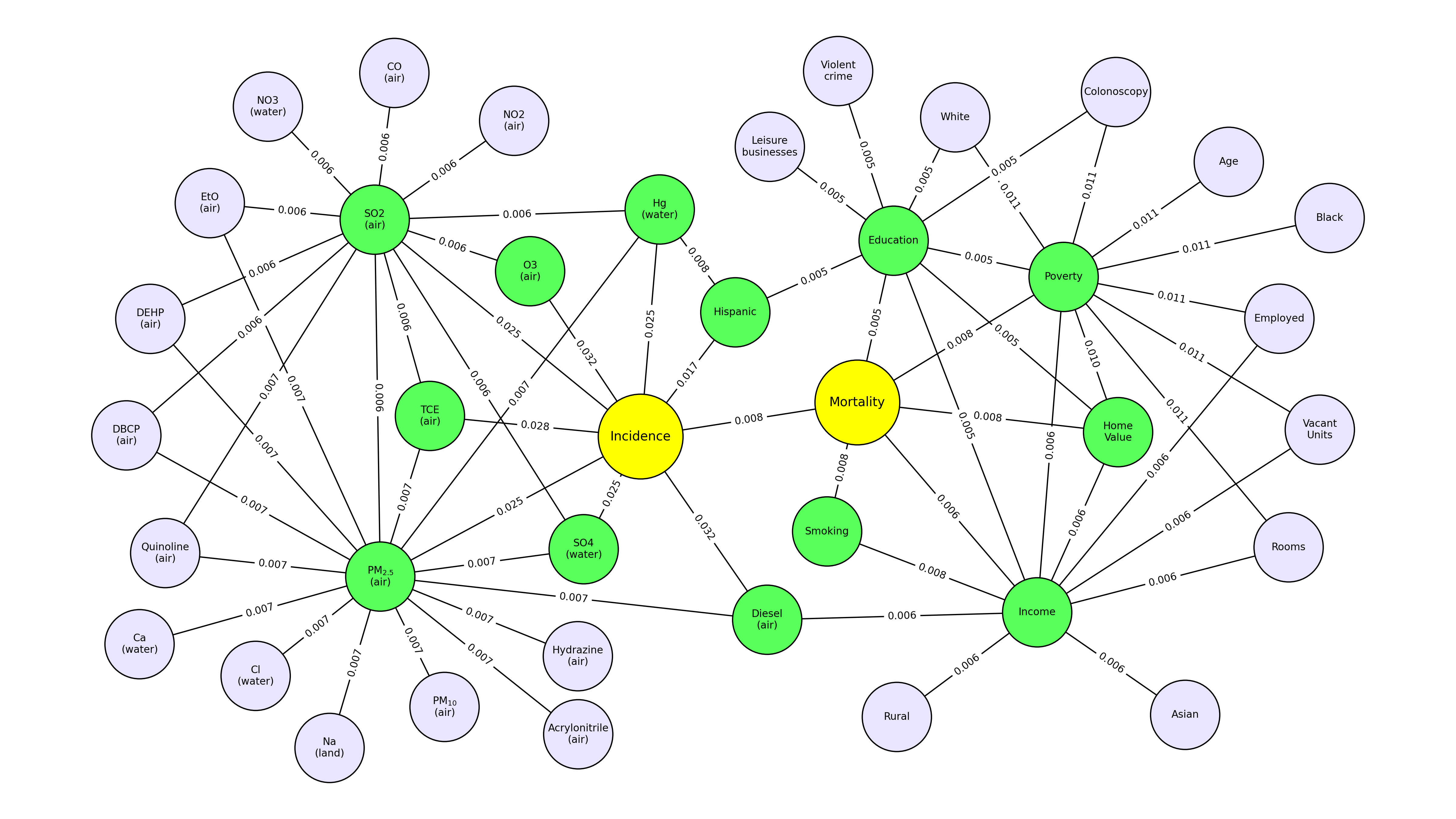}
\end{adjustbox}
\caption{\textit{Radius $2$ graph around county-level cancer incidence and mortality, estimated by PFS}. The target variables, cancer incidence and mortality, are shown in yellow. Nodes directly connected to a target are shown in green, and nodes connected to a green node but not to either target are shown in purple. Edges are annotated with their $q$-values, with smaller values indicating stronger evidence of conditional dependence. Exposures are labeled by environmental medium (air, land, or water).}
\label{fig:exposome_pfs}
\end{figure*}
}{}

%
%

\cref{fig:exposome_pfs} shows the local graph of radius 2 around cancer incidence and mortality estimated by PFS. Unlike other methods (\cref{sup_sec:other_methods_eqi}), PFS reveals two distinct and largely non-overlapping sets of interacting variables. The first set, located in the left half of the graph, links incidence to environmental exposures. Several of these---notably fine particulate matter (PM$_{2.5}$) and sulfur dioxide (SO$_2$)---directly connect to compounds classified by the International Agency for Research on Cancer (IARC) as known or probable human carcinogens (\cref{tab:exposures}). For example, PM$_{2.5}$ links to SO$_2$, which connects to other exposures such as ozone (O$_3$), carbon monoxide (CO), and nitrogen dioxide (NO$_2$), reflecting known atmospheric processes by which secondary pollutants are formed~\cite{exposome_who}.

The right half of the local graph, on the other hand, shows that cancer mortality connects most strongly to socioeconomic variables such as income and education. These conditional dependencies suggest a division of influence: even after adjusting for all other variables, environmental exposures remain strongly associated with cancer incidence, while survivorship is associated with social factors. This division is reflected in the heatmaps in \cref{fig:exposome_results}. High-incidence counties cluster in regions with elevated environmental exposures, while high-mortality counties overlap with high poverty and low education areas, most notably in the Mississippi Delta and Appalachian regions of Kentucky and West Virginia. Despite stark differences in racial composition---the Mississippi Delta has among the highest proportions of black residents in the country, while the Appalachian counties are almost entirely white---both areas experience similarly high cancer mortality. This pattern aligns with the estimated local graph, where black and white are conditionally independent of mortality given poverty. \cref{tab:correlations} further supports this hypothesis, showing that the correlation between black and mortality drops from 0.228 ($p$-value $< 10^{-10}$) to 0.053 ($p$-value $= 0.005$) after adjusting for poverty, while the correlation between poverty and mortality remains strong after adjusting for black (0.38, $p$-value $< 10^{-10}$). A similar asymmetry holds for white, reinforcing the role of poverty in shaping observed racial associations with mortality.

Counties in the Southwest with large Hispanic populations present a striking exception. These areas---spanning parts of Texas, New Mexico, Arizona, and Colorado---have poverty and education levels similar to those in the Mississippi Delta and Appalachia, yet exhibit significantly lower cancer incidence and mortality rates. One possible explanation supported by the local graph is that the Hispanic variable is associated with some of the lowest mercury levels in the country, which in turn links to low levels of other exposures. This raises the possibility that reduced exposure may lower cancer risk and contribute to lower mortality in these counties despite their socioeconomic disadvantage. Consistent with this interpretation, \cref{tab:correlations} shows that the correlation between Hispanic and incidence changes from $-0.34$ to $-0.174$ after adjusting for county-specific mercury levels. In contrast, adjusting for the four cancer screening variables only slightly reduces correlation from $-0.34$ to $-0.324$, indicating that low cancer incidence in counties with large Hispanic populations is unlikely to be explained by differences in screening rates.

\ifthenelse{\boolean{showfigures}}{
\begin{figure*}[ht!]
\centering
\captionsetup{width=0.95\textwidth}

\begin{subfigure}[t]{\textwidth}
\centering
\includegraphics[
    height=0.225\textheight,
    width=\textwidth,
    trim={0pt 0pt 0pt 0pt}, clip 
    ]{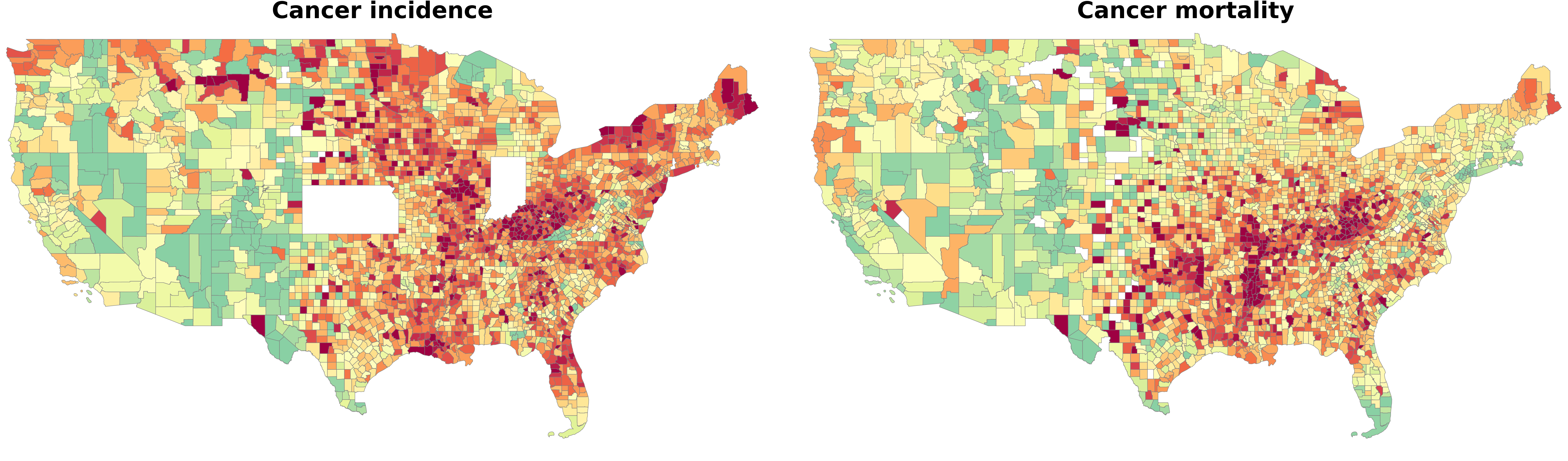}
\end{subfigure}

\vspace*{0em}

\begin{subfigure}[t]{0.495\textwidth}
\centering
\includegraphics[
    width=\textwidth,
    trim={0pt 0pt 0pt 0pt}, clip
    ]{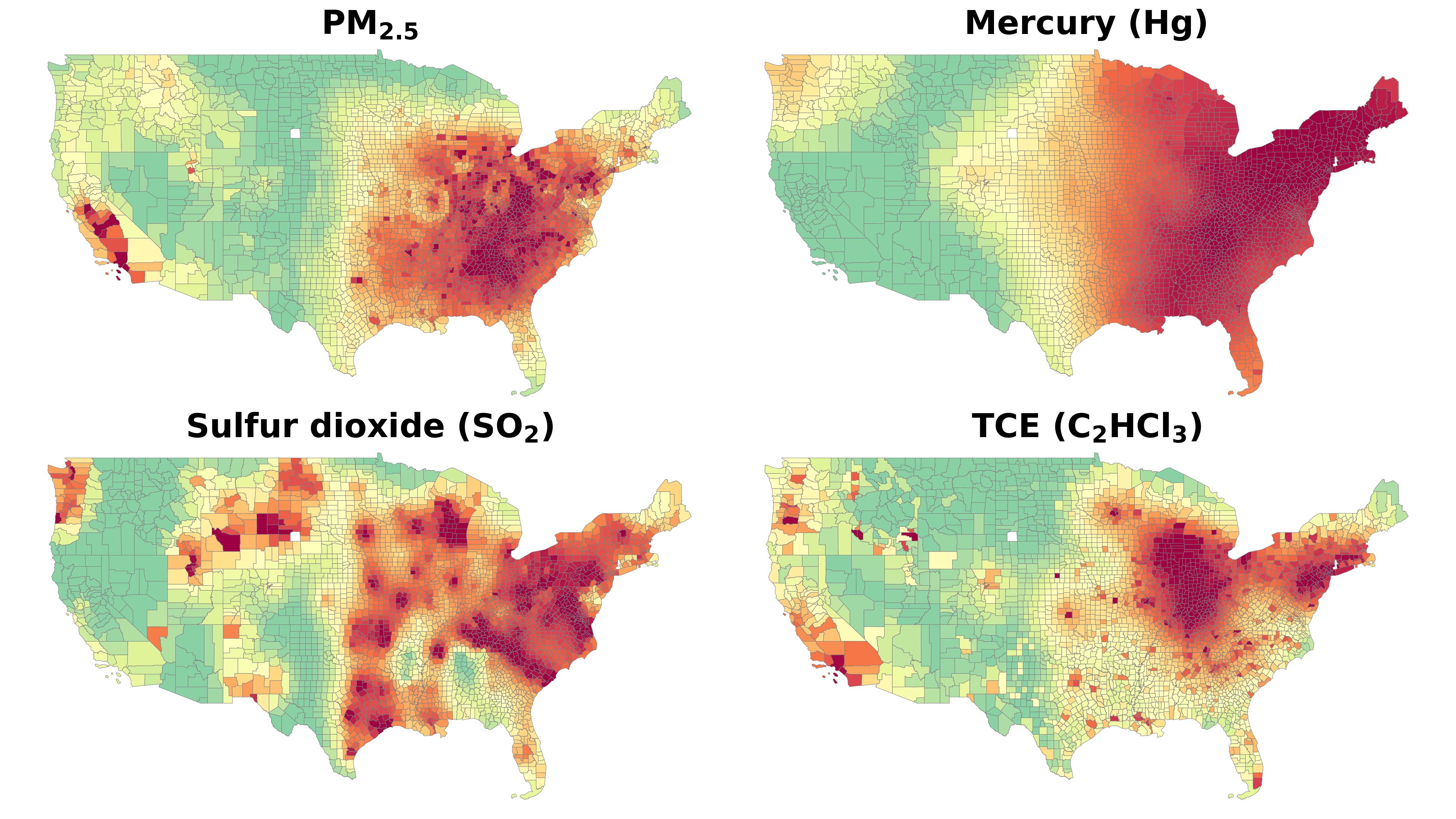}
\end{subfigure}
\hfill
\begin{subfigure}[t]{0.495\textwidth}
\centering
\includegraphics[
    width=\textwidth,
    trim={0pt 0pt 0pt 0pt}, clip
    ]{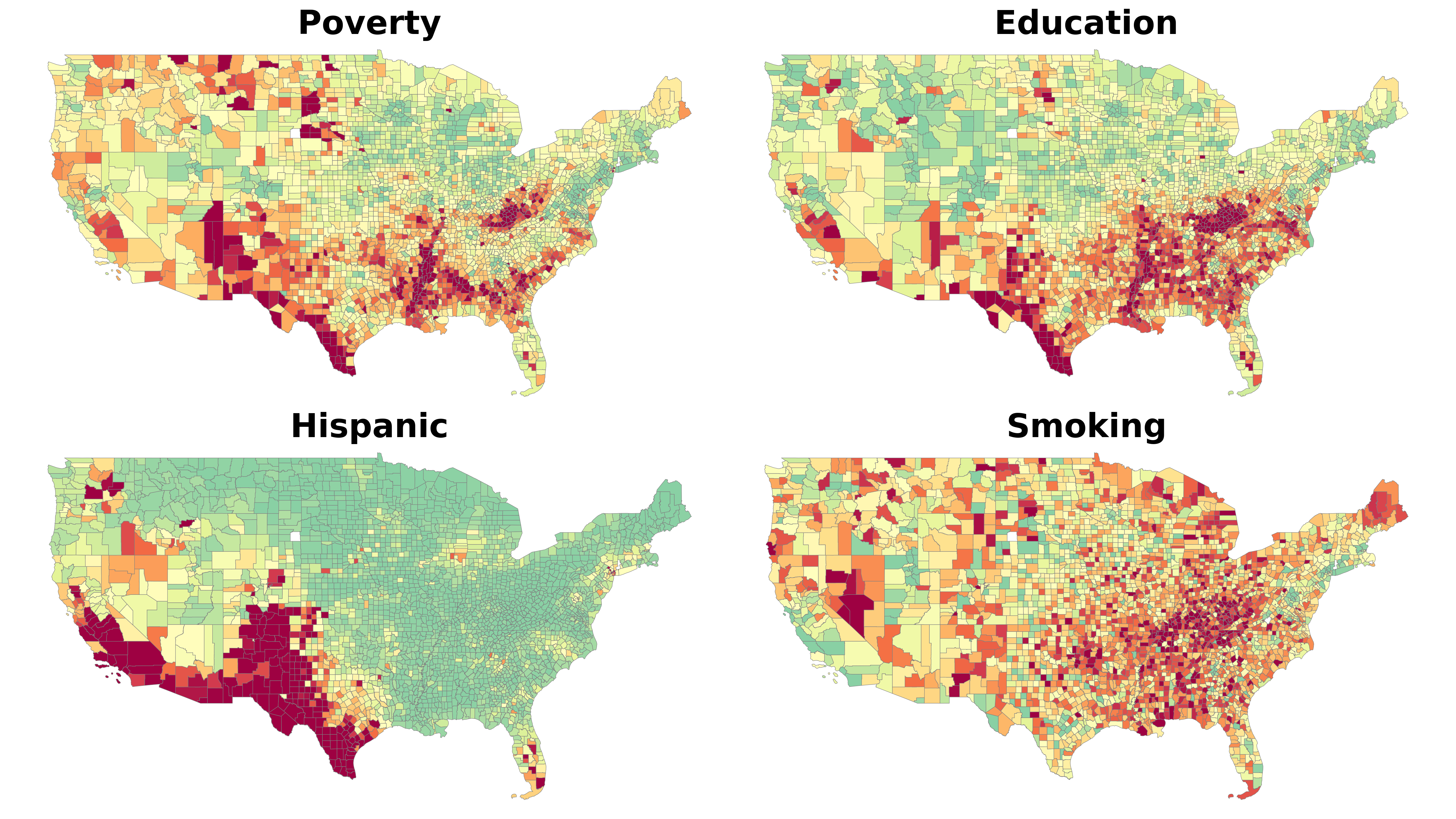}
\end{subfigure}

\caption{\textit{Geographic distributions of cancer burden and selected environmental and socioeconomic variables across the contiguous United States}.
Heatmaps show county-level values for cancer incidence and mortality (top row) and selected environmental exposures (bottom left) and social and demographic factors (bottom right). Cancer incidence tends to align with elevated levels of exposures---including particulate matter PM$_{2.5}$, mercury, sulfur dioxide, and trichloroethylene (TCE)---while cancer mortality is more closely linked to socioeconomic conditions. Counties with large Hispanic populations, concentrated in the Southwest, show high poverty but low cancer burden, possibly due to reduced environmental exposures. In the Education map, red corresponds to lower education levels and green to higher levels; in all other maps, red and green indicate higher and lower levels of the variable, respectively. Indiana and Kansas did not report cancer incidence and are therefore omitted from the study.}
\label{fig:exposome_results}
\end{figure*}
}{}


\subsection{Cross-modal pathways in breast cancer}\label{sec:breast_cancer}


Multiomic studies integrate different types of data to better understand structure within biological systems and identify disease-relevant pathways~\cite{multiomics_luo}. In cancer research, this strategy has enabled the discovery of molecular signatures with diagnostic and prognostic value~\cite{multiomics_wekesa, multiomics_himel2024}. Here, we apply PFS to multiomic breast cancer data from The Cancer Genome Atlas\cite{tcga} (TCGA) to investigate three clinical variables: histological type (invasive ductal carcinoma, IDC, or invasive lobular carcinoma, ILC), pathologic stage (stages I--IV), and survival status at last follow-up. In addition to these three targets, the cleaned dataset (\cref{sup_sec:bc}) includes $p = 10{,}741$ variables measured across $n = 547$ breast cancer patients, namely: $9785$ genes profiled by bulk RNA sequencing (RNA-seq), $819$ microRNAs (miRNAs), and $137$ proteins quantified by reverse phase protein arrays (RPPA). Because no ground-truth graph is available, we validate our findings through comparison with existing literature and additional statistical analyses, including over-representation analysis (ORA) and protein-protein edge validation (\cref{sup_sec:bc_additional}).

%
%

The local graph estimated by PFS (\cref{fig:breast_cancer}) reveals a biologically coherent network linking molecular and clinical features across modalities. The three clinical variables form a path---histological type (henceforth, \textit{subtype}) connects to pathologic stage, which connects to survival status---aligning with known clinical relationships. For example, ILC is often diagnosed at more advanced stages than IDC due to reduced detectability on mammography~\cite{breast_cancer_christgen}, consistent with the edge between subtype and stage. Stage is a well-established predictor of survival, and the absence of an edge between subtype and survival aligns with clinical uncertainty about their prognostic relationship~\cite{breast_cancer_zhao}.

CDH1 and FN1 have the strongest associations with subtype, forming a multimodal triangle with particularly small $q$-values. CDH1 encodes E-cadherin, whose inactivation is a defining feature of ILC, while FN1 has been linked to immune infiltration and poor prognosis in breast cancer~\cite{breast_cancer_barroso,breast_cancer_christgen,breast_cancer_zhang}. ORA results reported in \cref{tab:bc_ora} indicate that the CDH1-associated gene cluster is significantly enriched for metastasis-related signatures. miR-210, which plays an established role in hypoxic response and metabolic adaptation, forms a compact module with the protein PTEN and known gene targets NDRG1 and ISCU~\cite{breast_cancer_noman,breast_cancer_mccormick}. ORA of this cluster shows enrichment for hypoxia-related gene sets, while ORA of the miR-133a-1 module implicates genes involved in muscle contraction and cytoskeletal organization, aligning with prior evidence that miRNAs in this group co-regulate and are frequently downregulated in breast cancer~\cite{breast_cancer_kojima,breast_cancer_sawant}.

Many edges emanating from pathologic stage are between proteins. Independent validation against the STRING database~\cite{string} (\cref{sup_sec:bc_additional}) shows that PFS is the only method whose estimated graph exhibits significant enrichment for known protein-protein interactions (\cref{tab:bc_string}), supporting the biological plausibility of the protein subnetwork branching from stage. Stage is also part of a particularly notable pathway, connecting to ABL1, then to G6PD, and finally to the canonical oncogene ERBB2 (HER2). ABL1 is a non-receptor tyrosine kinase involved in oncogenic signaling, while G6PD plays a key role in oxidative stress regulation in HER2-positive breast cancer~\cite{breast_cancer_stern}.

ORAs of two gene clusters linked to survival, namely those emanating from ANKLE2 and HOOK3, show strong enrichment for breast cancer-specific copy number and mutation signatures (\cref{tab:bc_ora}). In both cases, the enriched gene sets are derived from recurrently amplified or mutated genomic regions in breast cancer, providing strong evidence that the survival-associated structure recovered by PFS reflects established disease biology rather than isolated gene-level associations.

Despite the validation results presented above, not all edges or clusters of edges in the estimated local graph admit a straightforward functional interpretation, and we do not attempt to assign biological roles to every substructure within it. Instead, we emphasize that PFS offers a structured and statistically grounded view of potential dependencies---both within and between modalities---that may warrant further investigation. In this way, local graph estimation can serve not only to recover known biology but to generate new hypotheses for future experimental validation.

\ifthenelse{\boolean{showfigures}}{
\begin{figure*}[ht!]
\begin{adjustbox}{center}
\includegraphics[
	height=0.4\textheight, 
	width=\textwidth,
	trim={80pt 40pt 80pt 40pt}, clip]
	{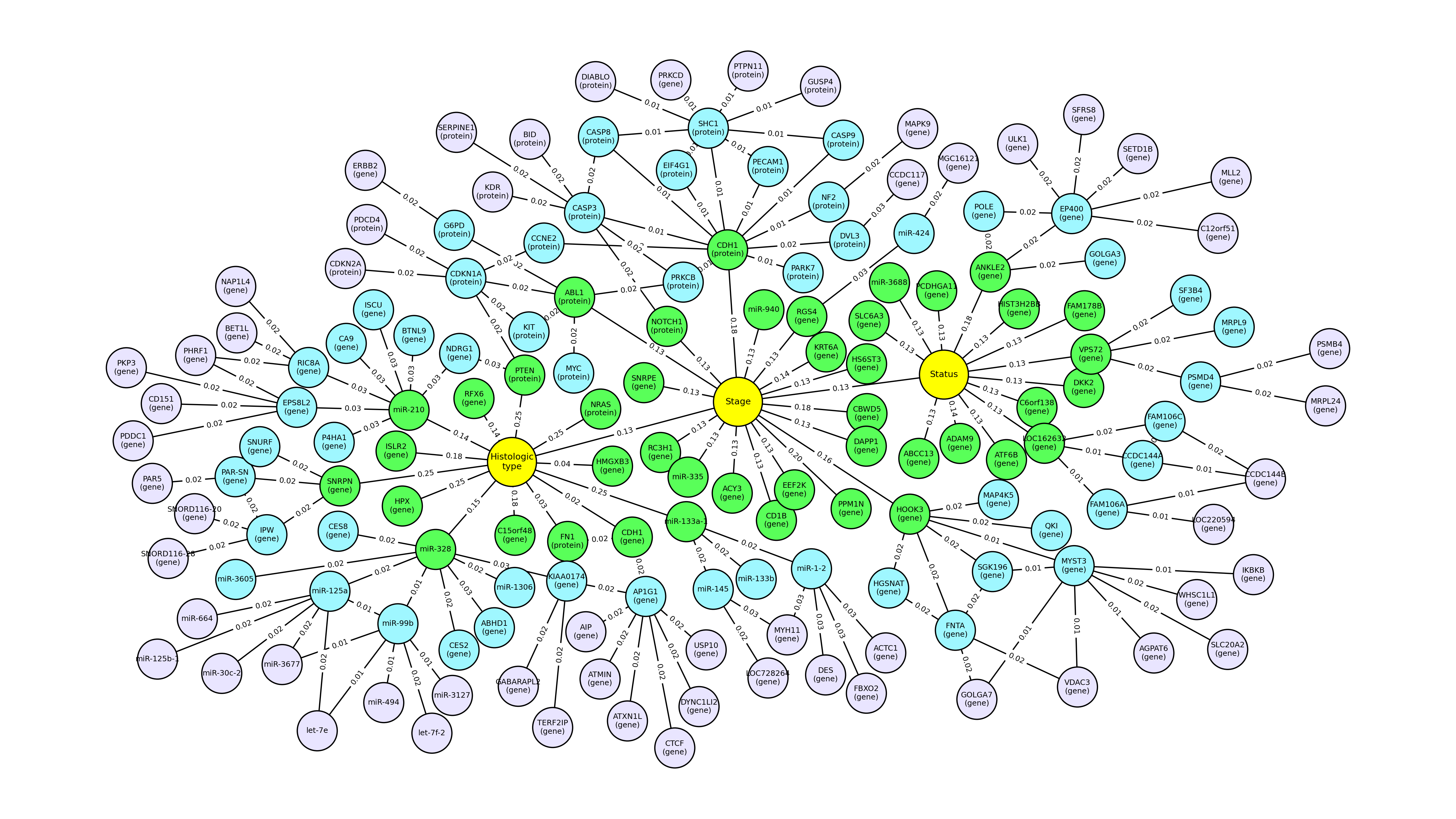}
\end{adjustbox}
\caption{\textit{Radius 3 graph around clinical variables in TCGA breast cancer data, estimated by PFS}. The target variables---histological type, pathologic stage, and survival status---are shown in yellow. Nodes directly connected to a target are shown in green, nodes at distance two are shown in blue, and nodes at distance three are shown in purple. Edges are annotated with $q$-values, with smaller values indicating stronger evidence of conditional dependence. Genes and proteins are labeled with``(gene)" and ``(protein)" respectively, indicating their molecular modality.}
\label{fig:breast_cancer}
\end{figure*}
}{}


\subsection{Brain networks and cognition}\label{sec:hcp}


The Human Connectome Project (HCP) is a large-scale neuroimaging effort to map macroscopic brain circuits and study their relationship to human behavior~\cite{hcp}. We applied PFS to data from the HCP Young Adult cohort to investigate brain network organization and cognitive ability. After preprocessing (\cref{sup_sec:hcp_cleaning}), the dataset includes $n = 1188$ individuals and $p = 213$ variables comprising various structural brain measurements---including cortical thickness, surface area, and regional volumetric measures across Desikan-Killiany regions~\cite{hcp_dk}---as well as behavioral and personality phenotypes. The target variable is the age-adjusted Fluid Cognition Composite score from the NIH Toolbox (henceforth, \textit{fluid cognition}), a standardized measure summarizing executive function, working memory, and processing speed.

%
%

\cref{fig:hcp_pfs} shows the local graph of radius 3 around fluid cognition estimated by PFS. Estimates from other methods are shown in \cref{fig:hcp_other_methods}. ORA results in \cref{tab:hcp} indicate that at radius 2, both PFS and a local version of the hybrid parents and children algorithm~\cite{bnlearn_hpc}, HPC(L), are enriched for regions in the frontoparietal control network (FPCN). At radius 3, both methods remain enriched for FPCN, while only PFS---whose graph contains 23 nodes at this radius---shows additional enrichment for the ventral attention network (VAN), suggesting greater specificity than HPC(L) (36 nodes).

The prominence of FPCN in the local graphs estimated by PFS and HPC(L) is consistent with prior work showing that this network plays a central role in executive control and flexibly couples with other large-scale systems depending on task demands~\cite{hcp_spreng,hcp_dixon}. In contrast to other methods (\cref{tab:hcp}), neither PFS nor HPC(L) exhibits enrichment for visual (VIS) or somatomotor (SM) networks at comparable radii. VIS and SM are \textit{unimodal} regions that primarily support sensory processing and motor function rather than integration across cortical systems~\cite{hcp_smallwood}. The absence of VIS/SM enrichment, together with FPCN (and, for PFS, VAN) involvement, aligns with fluid cognition, which emphasizes working memory, processing speed, and cognitive control.

\begin{table*}[ht!]
\centering
\begin{tabular}{lccccccccc}
\toprule
Method & Radius & Nodes & VIS & SM & DAN & VAN & LIM & FPCN & DMN \\
\bottomrule
\multirow{2}{*}{\textbf{PFS}} & 2 & 7 & -- & -- & -- & -- & -- & 0.0027 & -- \\
 & 3 & 23 & -- & -- & -- & 0.0655 & -- & 0.0109 & -- \\
\hline
\multirow{2}{*}{HPC(L)} & 2 & 8 & -- & -- & -- & -- & -- & 0.0042 & -- \\
 & 3 & 36 & -- & -- & -- & -- & -- & 0.0569 & -- \\
\hline
\multirow{2}{*}{IAMB(L)} & 2 & 14 & $<10^{-4}$ & -- & -- & -- & -- & -- & -- \\
 & 3 & 100 & 0.0218 & -- & -- & -- & -- & -- & -- \\
\hline
\multirow{2}{*}{MMPC(L)} & 2 & 6 & 0.0027 & -- & -- & -- & -- & -- & -- \\
 & 3 & 24 & 0.0009 & -- & -- & -- & -- & -- & -- \\
\hline
\multirow{2}{*}{StablePC} & 2 & 5 & $<10^{-4}$ & -- & -- & -- & -- & -- & -- \\
 & 3 & 12 & $<10^{-4}$ & -- & -- & -- & -- & -- & -- \\
\hline
\multirow{2}{*}{GFCSL} & 2 & 18 & -- & -- & 0.0342 & -- & -- & -- & -- \\
 & 3 & 148 & -- & -- & -- & -- & -- & -- & -- \\
\hline
\multirow{2}{*}{Glasso} & 2 & 10 & $<10^{-4}$ & -- & -- & -- & -- & -- & -- \\
 & 3 & 106 & 0.0971 & -- & -- & -- & -- & -- & -- \\
\bottomrule
\end{tabular}
\caption{\textit{Enrichment of local brain networks associated with fluid cognition}. For each method, we report $p$-values from over-representation analyses for Yeo-7 functional networks~\cite{hcp_yeo} among cortical regions in the estimated local graph. Local clusters are anchored at a single region of interest that connects directly to the target; radii correspond to cumulative neighborhoods containing all nodes within graph distances 2 and 3 of the target that extend beyond this anchor. The ``Nodes" column reports the total number of nodes in each cluster. The Yeo-7 networks are: VIS (visual), SM (somatomotor), DAN (dorsal attention), VAN (ventral attention), LIM (limbic), FPCN (frontoparietal control), and DMN (default mode). Only $p$-values less than $0.1$ are shown; entries marked ``--" indicate no enrichment at this threshold.}
\label{tab:hcp}
\end{table*}


\subsection{Cell-type-specific gene networks in Alzheimer's disease}\label{sec:ad}


Single-nucleus RNA sequencing (snRNA-seq) enables profiling of gene expression in individual cells. Grubman et al.~\cite{ad_grubman} generated snRNA-seq data from the entorhinal cortex of six individuals with Alzheimer's disease (AD) and six controls to characterize transcriptional changes associated with the disease. We applied PFS to these data to identify AD-associated gene networks in three cell types: astrocytes ($n=2171$ cells, $p=10{,}788$ genes), microglia ($n=449$, $p=10{,}514$), and oligodendrocyte progenitor cells (OPCs; $n=1078$, $p=10{,}773$). The target variable is AD status, defined as whether a cell is from an AD or control individual. Additional dataset details are in \cref{sup_sec:ad}.

%
%

\cref{fig:ad_pfs_astro,fig:ad_pfs_mg,fig:ad_pfs_opc} show the local graphs of radius 3 around AD status estimated by PFS. ORAs of these graphs (\cref{tab:ad}) indicate distinct biological mechanisms for each cell type. Specifically, the microglial graph is enriched for immune activation and inflammatory signaling, consistent with the central role of neuroinflammation in AD. The astrocyte graph shows enrichment for mitochondrial respiration and ion transport, reflecting altered energy metabolism and oxidative stress in diseased tissue. The OPC graph highlights central nervous system development, cell adhesion, and synapse-related processes. Several genes with well-established roles in AD, such as APOE and GFAP, appear in the estimated graphs. Importantly, the three analyses yield different local structures despite sharing the same target, suggesting that PFS captures cell-type-specific disease organization rather than a generic AD signature.

\begin{table}[ht!]
\centering
\small
\begin{tabular}{@{}llc@{}}
\toprule
\textbf{Cell Type} & \textbf{Enriched biological process} & \textbf{Adjusted $p$-value} \\
\midrule
\multicolumn{3}{@{}l}{\textbf{Astrocytes}} \\
& GO:BP: Monoatomic ion transport & 0.001 \\
& GO:BP: Oxidative phosphorylation & 0.001 \\
& GO:BP: ATP synthesis coupled electron transport & 0.002 \\
& Reactome: Aerobic respiration and respiratory electron transport & 0.012 \\
& Reactome: Respiratory electron transport & 0.012 \\
& Reactome: Extracellular matrix organization & 0.031 \\
\addlinespace
\multicolumn{3}{@{}l}{\textbf{Microglia}} \\
& GO:BP: Regulation of cell activation & $<10^{-3}$ \\
& GO:BP: Positive regulation of leukocyte proliferation & $<10^{-3}$ \\
& GO:BP: Regulation of lymphocyte activation & $<10^{-3}$ \\
& GO:BP: Adaptive immune response & $<10^{-3}$ \\
& Reactome: Cytokine signaling in immune system & 0.024 \\
& Reactome: Hemostasis & 0.029 \\
\addlinespace
\multicolumn{3}{@{}l}{\textbf{Oligodendrocyte progenitor cells (OPCs)}} \\
& GO:BP: Central nervous system development & $<10^{-3}$ \\
& GO:BP: Cell-cell adhesion & $<10^{-3}$ \\
& GO:BP: Presynapse assembly & $<10^{-3}$ \\
& GO:BP: Synapse organization & $<10^{-3}$ \\
& Reactome: Protein--protein interactions at synapses & 0.012 \\
& Reactome: Neuronal system & 0.013 \\
\bottomrule
\end{tabular}
\caption{\textit{Enrichment of PFS local graphs around AD status by cell type.} Over-representation analyses were performed using Reactome and Gene Ontology Biological Process (GO:BP) gene sets (\cref{sup_sec:ad_additional}). Background sets consist of all genes analyzed by PFS. All reported terms have an overlap of at least 5 genes.}
\label{tab:ad}
\end{table}


\section{Discussion}\label{sec:discussion}


We introduced local graph estimation, a statistical objective focused on identifying substructures around target variables in complex data, and proved that our proposed solution, PFS, provides both edgewise and pathwise uncertainty quantification. Together, local graph estimation and PFS enable statistically principled discovery of paths and clusters of variables, opening new directions for inference on localized network structure. Local graph estimation also supports downstream analyses that rely on graph structure, including graph-based clustering approaches~\cite{mappr}.

PFS is broadly applicable. When implemented with nonparametric selection methods such as IPSS, it accommodates mixed data types, avoids strong modeling assumptions, and readily scales to tens of thousands of variables. Across four distinct applications, we showed that PFS recovers established domain structure while revealing potentially novel relationships. More broadly, local graph estimation offers a flexible and interpretable alternative to global graph estimation for identifying candidate mechanistic relationships centered on variables of primary scientific interest.


\section*{Methods}\label{sec:methods}



\subsection*{Description of \cref{alg:pfs}}


\Cref{alg:pfs} provides a step-by-step description of PFS. The algorithm takes as input an $n \times p$ data matrix $\bm{X}$ comprising $n$ observations of $p$ variables, a set of target features $V_0$, and user-specified $q$-value thresholds. The output is a weighted $p \times p$ adjacency matrix, $Q$, whose entries are the pairwise $q$-values computed during the iterative selection process. Neighborhood thresholds $q_r^*$ control the false discovery rate (FDR) within each estimated neighborhood at iteration $r$, while the path threshold $q_{\mathrm{path}}^*$ constrains the maximum sum of $q$-values along any path, thereby providing an upper bound on the probability that a recovered path does not belong to the true graph under the assumptions of \cref{thrm:main}.

The quantity $d_Q(V_0, j, r)$ in Line 9 is the minimum sum of $q$-values along any path of length $r$ from $V_0$ to node $j$ in the estimated graph, namely
\begin{align}\label{eq:lightest_path}
d_Q(V_0,j,r) = \min \left\{ \sum_{s=0}^{r-1} Q_{j_s, j_{s+1}} : (j_0,\dots,j_r) \in \widehat{\mathcal{J}}_r(V_0),\, j_r = j \right\},
\end{align}
where $\widehat{\mathcal{J}}_r(V_0)$ is the set of length-$r$ paths in the current estimated local graph that start in $V_0$ (see \cref{sup_sec:lge,sup_sec:pfs} for a complete discussion of definitions and notation). Line 10 updates the new layer $S$ of previously unidentified nodes, and Line 11 updates the set $B$ of nodes whose neighborhoods have already been estimated. Together, these updates guarantee that each neighborhood is estimated at most once during the recursive process.


\subsection*{Discussion of \cref{thrm:main}}


\cref{thrm:main} upper bounds the probability that a given path does not belong to the true graph by the sum of the edge-level $q$-values along that path. This result applies to individual paths, but does not provide joint or simultaneous control over collections of paths; in particular, paths that share edges are dependent, complicating extensions of \cref{thrm:main} to such settings. The assumptions of \cref{thrm:main} are a direct extension of classical assumptions used to justify $q$-values in multiple testing (\cref{sup_sec:qvalues}). Specifically, \cref{thrm:main} supposes that, along a given path, test statistics used to determine whether each edge is selected are independent. Informally, this means that whether one edge appears in an estimated path depends only on the evidence for that edge, and not on the evidence for other edges along the path. If these assumptions are violated---for example, if strong dependence between edge-level test statistics causes the selection of one edge to influence the apparent evidence for another---then error control may fail. When these assumptions do hold, edge-level $q$-values can be interpreted as posterior probabilities that a selected edge is not in the true graph~\cite{storey}. Extending this to paths gives \cref{thrm:main} a natural Bayesian interpretation: the sum of $q$-values along a path upper bounds the posterior probability that an estimated path is not in the true graph (\cref{sup_sec:paths}), motivating our use of $q$-values to limit uncertainty propagation as paths extend away from targets.


\subsection*{Integrated path stability selection}\label{sec:ipss}


At each iteration of PFS, we estimate edges using integrated path stability selection (IPSS)~\cite{ipss}, a feature selection method that provides finite-sample false discovery control. IPSS is a refinement of classical stability selection~\cite{mb}, in which an arbitrary base feature selection algorithm is repeatedly applied to random half-samples of the data. The proportion of times each feature is selected across all subsamples is computed and then aggregated across regularization levels (for example, penalty parameters or importance score thresholds) to produce feature-specific \textit{stability paths}. Rather than applying a hard threshold to stability paths as in classical stability selection, IPSS integrates information along stability paths to yield more precise $q$-values and tighter false discovery control than previous approaches. IPSS also retains the favorable robustness properties of stability selection, which arise from the repeated subsampling scheme. In this work, we use the nonparametric variant of IPSS, which is compatible with arbitrary variable importance measures and captures nonlinear associations without requiring distributional or parametric modeling assumptions~\cite{ipss_nonparametric}. Specifically, we implement IPSS using mean decrease impurity from random forests and gradient boosting, which is computationally efficient and the default choice in many popular machine learning packages.


\subsection*{PFS parameters}


The principal PFS parameters are the maximum radius $r_{\max}$ of the estimated local graph, the pathwise $q$-value threshold, and, optionally, edge-level $q$-value thresholds. There is no single ``correct'' choice of these parameters in general: $q$-value thresholds encode how much uncertainty a user is willing to tolerate in the estimated graph, which is inherently problem-specific. As with significance levels in hypothesis testing more broadly, there is no universally optimal choice of edgewise or pathwise $q$-value thresholds, nor an oracle rule for selecting them. The choice of $r_{\max}$ reflects the scope of the scientific question rather than uncertainty tolerance. In some applications, domain knowledge may suggest that relationships beyond a certain distance from the target variables are not of interest, in which case $r_{\max}$ can be fixed accordingly. Alternatively, $r_{\max}$ can be made arbitrarily large, leaving the breadth of the local graph to be determined entirely by the pathwise $q$-value threshold. Larger radii and more liberal thresholds will generally yield broader, denser local graphs with more discoveries but greater uncertainty, while smaller radii and stricter thresholds will produce sparser, potentially more interpretable, structures with less uncertainty. PFS is designed to make these tradeoffs explicit and user-controlled, rather than implicit or fixed by the method.


\subsection*{Simulation design for \cref{fig:limitations}}


The data used in \cref{fig:limitations}, panels (a) through (d), are generated from a $p=100$ dimensional multivariate Gaussian distribution with mean zero and precision matrix $\Theta$, denoted $\mathcal{N}(0,\Theta^{-1})$. The precision matrix is composed of three blocks of sizes $1\times 1$, $2\times 2$, and $97\times 97$, corresponding to the target variable, its neighbors, and all remaining variables. Within-block and between-block connections are generated at random such that the average degree of the graph is approximately $3$, and nonzero entries of $\Theta$ are chosen uniformly at random from $\{\pm 1\}$. The resulting matrix is symmetrized, a positive constant is added to the diagonal to ensure positive definiteness, and the matrix is rescaled so that its eigenvalues lie between $0.01$ and $10$, allowing for strong correlations between variables. Given $\Theta$, we draw $n=200$ samples independently from $\mathcal{N}(0,\Theta^{-1})$, and standardize the resulting data to have zero mean and unit variance.


\subsection*{Data availability}


All data used in this work are publicly available. For the environmental and sociodemographic cancer study, county-level data on cancer incidence, mortality, screening, and smoking prevalence are available from the State Cancer Profiles project at \url{https://statecancerprofiles.cancer.gov/}. Environmental and socioeconomic variables are available from the EPA Environmental Quality Index (EQI) at \url{https://cfpub.epa.gov/ncea/risk/recordisplay.cfm?deid=316550}, and demographic data are available from the U.S. Census Bureau at \url{https://data.census.gov/}. Data from the multiomic breast cancer study can be downloaded from LinkedOmics at \url{https://www.linkedomics.org/data_download/TCGA-BRCA/}. Data from the brain network and cognition analysis were obtained from the Human Connectome Project (HCP Young Adult cohort) via ConnectomeDB (\url{https://www.humanconnectome.org}), accessed through the BALSA data portal at \url{https://balsa.wustl.edu}. Data from the Alzheimer's disease study can be downloaded from \url{https://adsn.ddnetbio.com} and are also available from the Gene Expression Omnibus (GEO) under the accession number GSE138852.


\subsection*{Code availability}\label{sec:code}


Code and processed data files required to reproduce all results in this paper are available at \url{https://github.com/omelikechi/localgraph-paper}. A Python package implementing local graph estimation and PFS is available at \url{https://github.com/omelikechi/localgraph} and can be installed via PyPI at \url{https://pypi.org/project/localgraph/}.


\putbib

\end{bibunit}


\subsection*{Acknowledgments}\label{sec:acks}


D.B.D. and O.M. were supported in part by funding from Merck \& Co. and the National Institutes of Health (NIH) grant R01ES035625. D.B.D. was supported in part by the Office of Naval Research grant N000142412626. J.W.M. and O.M. were supported in part by the Collaborative Center for X-linked Dystonia Parkinsonism (CCXDP). J.W.M. was supported in part by NIH grant R01CA240299. All authors thank the anonymous reviewers for their feedback.


\subsection*{Author contributions}\label{sec:contributions}


O.M. devised the method and developed the theory, wrote software, conducted simulations and applied analyses, and wrote the manuscript. D.B.D. and J.W.M. provided funding and supervision. D.B.D., J.W.M., and N.M. helped with conceptualization and editing.


\subsection*{Competing interests}\label{sec:competing}


The authors declare no competing interests.


\clearpage

\title{{\LARGE S}{\Large upplementary information}}

%

\maketitle

\begin{bibunit}

\setcounter{page}{1}
\setcounter{section}{0}
\setcounter{table}{0}
\setcounter{figure}{0}
\renewcommand{\theHsection}{SIsection.\arabic{section}}
\renewcommand{\theHtable}{SItable.\arabic{table}}
\renewcommand{\theHfigure}{SIfigure.\arabic{figure}}
\renewcommand{\thepage}{S\arabic{page}}  
\renewcommand{\thesection}{S\arabic{section}}   
\renewcommand{\thetable}{S\arabic{table}}   
\renewcommand{\thefigure}{S\arabic{figure}}
\renewcommand{\thealgorithm}{S\arabic{algorithm}}


\section{Local graph estimation}\label{sup_sec:lge}


In this section, we formally define the local graph estimation problem. Let $X=(X_1,\dots,X_p)$ be a random vector with conditional independence graph $G = (V,E)$; as noted in \cref{sec:lge}, a unique such $G$ exists whenever the distribution of $X$ is positive \cite{koller}. A \textit{path} of length $r$ is a sequence of distinct nodes $\boldj = (j_0,\dots,j_r)\in V^{r+1}$. The path $\boldj$ lies in $G$ if $(j_s,j_{s+1})\in E$ for all $0\leq s < r$, and it starts in $V_0\subseteq V$ if $j_0\in V_0$. The \textit{distance} $d(j,k)$ in $G$ between $j,k\in V$ is the length of the shortest path in $G$ from $j$ to $k$. If no path exists, then $d(j,k)=\infty$. For $V_0\subseteq V$, define $B_0(V_0) = V_0$ and, for $r\in\mathbb{N}$,
\begin{align}\label{eq:ball}
	B_r(V_0) &= \big\{k\in V : \min_{j\in V_0} d(j,k)\leq r\big\}.
\end{align}
$B_r(V_0)$ is the ball of radius $r$ around $V_0$. Its boundary, $S_r(V_0) = B_r(V_0)\setminus B_{r-1}(V_0)$, is the set of nodes that are distance $r$ from $V_0$ (the notation $S$ stands for \textit{sphere}, complementing the ball notation, $B$). For $r\in\mathbb{N}$, define the edge set
\begin{align}\label{eq:edge}
	E_r(V_0) &= \big\{(j,k)\in E : j\in B_{r-1}(V_0)\ \text{or}\ k\in B_{r-1}(V_0)\big\}.
\end{align}
Edges between nodes in the outer layer $S_r(V_0)$ are omitted in our definition of $E_r(V_0)$ to avoid unnecessary estimation when performing pathwise feature selection (PFS). This exclusion does not compromise the conditional independence structure under the assumption that the distribution of $X$ is positive; in particular, if $X$ has a positive distribution, then $B_{r-1}(V_0)$ is conditionally independent of $V\setminus B_r(V_0)$ given $S_r(V_0)$; see, for example, Corollary 4.1 in Koller and Friedman~\cite{koller}.

The \textit{local graph} of radius $r$ around $V_0$ is the subgraph $G_r(V_0) = (B_r(V_0), E_r(V_0))$ of $G$. When $V_0=\{j\}$ consists of a single node $j$, we write $G_r(j)$ as a shorthand for $G_r(\{j\})$, and $S_1(j)$ for the neighborhood of $j$. In the above notation, the local graph estimation problem is: \textit{Given $V_0\subseteq V$ and a radius $r\in\mathbb{N}$, estimate $G_r(V_0)$ from samples of $X$.} The two extremes, namely estimating $G_1(V_0)$ when $V_0$ contains a single variable (often referred to as supervised feature selection or Markov blanket estimation), and estimating the full graph, $G_\infty(V_0)=G$ (traditional graph estimation) are widely studied; the intermediate cases are not.


\section{Limitations of full graph estimation}\label{sup_sec:limitations}


In \cref{sec:lge}, we remarked that one way to perform local graph estimation is to obtain an estimate $\widehat{G}$ of the full graph $G$, then estimate $G_r(V_0)$ by $\widehat{G}_r(V_0)=(\widehat{B}_r(V_0), \widehat{E}_r(V_0))$, where $\widehat{B}_r(V_0)$ and $\widehat{E}_r(V_0)$ are defined as in \cref{eq:ball,eq:edge} but with $\widehat{G}$ in place of $G$. In this section, we describe some key limitations of this general approach, highlighting the need for methods like PFS that are specifically tailored to the local estimation problem.


\subsection{Limitations of global methods}\label{sup_sec:limits_global}


Most full graph estimation methods fall into one of two categories: \textit{Global} methods estimate the full graph simultaneously, while \textit{nodewise} methods estimate the neighborhood of each node, then combine these estimates to infer the full graph. Perhaps the most popular global method is the graphical lasso \cite{glasso}, which assumes that $X\sim\mathcal{N}(0,\Theta^{-1})$ is multivariate Gaussian with precision matrix $\Theta$. In this setting, $X_j$ and $X_k$ are conditionally dependent given all other variables, denoted $X_j \ci X_k \mid X_{V\setminus\{j,k\}}$, if and only if $\Theta_{jk}=0$, so estimating $G$ is equivalent to estimating the set of zero entries of $\Theta$. The graphical lasso estimates $\Theta$ via
\begin{align*}
\widehat{\Theta}(\lambda) &= \argmin_{\Theta\succ 0} \tr(\widehat{\Sigma}\Theta) - \log\det{\Theta} + \lambda\lVert \Theta\rVert_1,
\end{align*}
where $\widehat{\Sigma}$ denotes the empirical covariance matrix, $\lVert\Theta\rVert_1=\sum_{i,j=1}^p \lvert\Theta_{ij}\rvert$ is the $\ell_1$-norm of $\Theta$, and the minimization is over all symmetric positive-definite matrices $\Theta\succ 0$. The regularization parameter $\lambda\geq 0$ encourages sparse solutions, with larger values of $\lambda$ typically yielding more zeros in $\widehat{\Theta}(\lambda)$. 

A key limitation of the graphical lasso with respect to local graph estimation is that it is often impossible to choose a $\lambda$ that accurately estimates subgraphs of the true graph. Specifically, under the Gaussian assumption, a necessary condition\cite{isolated} for an edge to exist between $j$ and any other node is $\lambda < \max\{\lvert\widehat{\Sigma}_{jk}\rvert : k\in V\setminus\{j\}\}$. \cref{fig:limitations}c shows the estimate $\widehat{G}_3(1)$ of the local graph of radius $3$ around $X_1$ obtained using the graphical lasso with $\lambda = 0.9999\max\{\lvert\widehat{\Sigma}_{1k}\rvert : k\neq 1\}$ applied to data simulated from a $100$-dimensional Gaussian distribution, as described in the Methods section. Despite being essentially the sparsest estimate that includes the edge $(1,2)$, this choice of $\lambda$ yields $698$ false edges within radius $3$ of $X_1$. As shown in \cref{sup_sec:simulations}, such spurious edges are typical of the graphical lasso, particularly in scenarios where the connections between a target variable and its neighbors are weaker than the connections between those neighbors and other variables.

Global inference is possible for Gaussian graphical models due to the exceptional correspondence between the graph $G$ and the sparsity pattern of $\Theta$ in this setting. To relax the Gaussian assumption, many other methods for undirected graph estimation transform the data and then apply the graphical lasso \cite{nonparanormal, nonparanormal2}. However, since these methods ultimately rely on the graphical lasso in the final step, they are subject to the same limitation described above, regardless of the initial transformation. 


\subsection{Limitations of global error control}\label{sup_sec:limits_error}


Much recent work in graphical modeling focuses on developing methods for full graph estimation with \textit{false discovery rate} (FDR) control \cite{fdr_liu, silggm, fdr_li}. The FDR of an estimator $\widehat{E}$ of the edge set $E$ is the expected value of the \textit{false discovery proportion} (FDP),
\begin{align*}
\mathrm{FDP}(\widehat{E}) &= \frac{\lvert \widehat{E} \setminus E\rvert}{\max\{\lvert\widehat{E}\rvert, 1\}},
\end{align*}
where $\lvert A\rvert$ is the size of a set $A$ and the maximum in the denominator is taken to avoid division by zero. For $V_0\subseteq V$ and a radius $r\in\mathbb{N}$, the analogue for the local graph of radius $r$ around $V_0$ is 
\begin{align}\label{eq:fdp}
\mathrm{FDP}(\widehat{E}_r(V_0)) &= \frac{\lvert \widehat{E}_r(V_0) \setminus E_r(V_0)\rvert}{\max\{\lvert\widehat{E}_r(V_0)\rvert, 1\}}.
\end{align} 
\cref{fig:limits_error} shows that the global FDP of the full estimated graph can be very different from the FDPs within local graphs around specific nodes. For example, in the left panel of \cref{fig:limits_error} we have $\mathrm{FDP}(\widehat{E})=0.1$ and $\mathrm{FDP}(\widehat{E}_1(3))=1$, while in the right panel we have $\mathrm{FDP}(\widehat{E})=0.9$ and $\mathrm{FDP}(\widehat{E}_1(1))=0$. Thus, methods with global FDR control do not guarantee error control within local graphs. This is critical from the point of view of interpretability: We cannot conclude that the local graph around a variable of interest is accurate just because the global error rate is low, nor can we conclude that a local graph is inaccurate just because the global error rate is high.

\ifthenelse{\boolean{showfigures}}{
\begin{figure*}[ht]
\begin{adjustbox}{center}
\includegraphics[
	height=0.2\textheight, 
	width=\textwidth,
	]
	{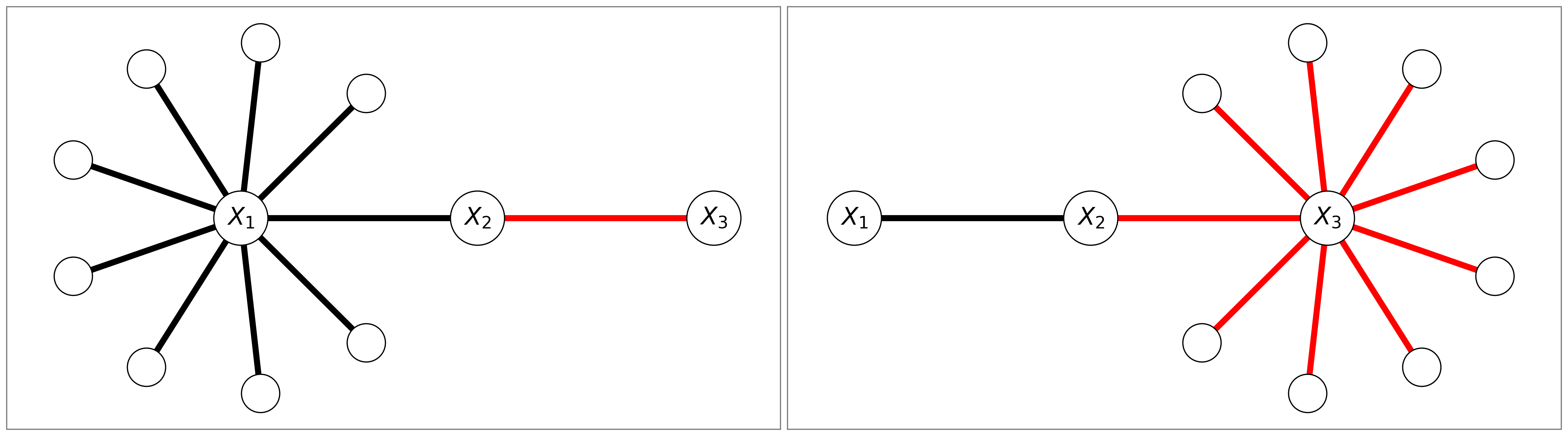}
\end{adjustbox}
\caption{\textit{Global error $\neq$ local error}. Hypothetical estimated graphs, with black and red edges representing true and false positives, respectively. In both cases, we have $\mathrm{FDP}(\widehat{E}_1(1))=0$, $\mathrm{FDP}(\widehat{E}_1(2))=0.5$, and $\mathrm{FDP}(\widehat{E}_1(3))=1$. (\textit{Left}) The false discovery proportion of the full estimated graph is $\mathrm{FDP}(\widehat{E})=0.1$. (\textit{Right}) The false discovery proportion of the full estimated graph is $\mathrm{FDP}(\widehat{E})=0.9$.}
\label{fig:limits_error}
\end{figure*}
}{}

The practical implications of the difference between global and local error control are evident in \cref{sup_sec:other_methods}, where multiple global FDR methods are applied to the county-level cancer and breast cancer datasets. In \cref{fig:bc_gfcsl}, for example, Gaussian graphical modeling with FDR control using scaled lasso~\cite{fdr_liu} (GFCSL), implemented in the R package \texttt{SILGGM}~\cite{silggm}, yields an extremely dense local subgraph of radius 2 around the three clinical targets (\cref{fig:bc_gfcsl}c). This occurs even though the radius-1 subgraph is reasonably sparse and the target global FDR was set to 0.1, which is the default value in \texttt{SILGGM} and the smallest threshold at which an edge is placed between pathologic stage and status (\cref{fig:bc_gfcsl}a). Taken together, these results illustrate that global FDR guarantees often fail to both recover the immediate neighborhoods of target variables \textit{and} constrain the number of edges at larger radii, significantly limiting their ability to yield interpretable local structure.


\section{Pathwise feature selection: Theory and implementation}\label{sup_sec:pfs}


Recall the description of pathwise feature selection (PFS) in \cref{sec:lge}: for each variable of interest $j_0\in V_0$, perform feature selection to obtain an estimate $\widehat{S}_1(j_0)$ of the neighborhood $S_1(j_0)$ of $j_0$. Then, for each $j_0\in V_0$ and each $j_1\in \widehat{B}_1(j_0)\setminus V_0$, perform feature selection again to obtain an estimate $\widehat{S}_1(j_1)$ of $S_1(j_1)$, and so on, until some stopping criterion is met.

As described, the above procedure is a generic nodewise selection algorithm that begins in $V_0$ and terminates according to some criterion. Similar methods for full graph estimation already exist \cite{drton,silggm}. The critical difference is that PFS quantifies the accumulation of uncertainty along estimated paths that start in $V_0$, enabling more precise inference near the target variables. To do this, we extend the concept of $q$-values to paths. First, we recall the usual definition of a $q$-value.


\subsection{Storey's \texorpdfstring{$\bm{q}$}{q}-values}\label{sup_sec:qvalues}


Graph estimation can be framed as a multiple hypothesis testing problem in which hypotheses $H_j(k) = \1(k \in S_1(j))$ are tested for all pairs $j,k\in V$. That is, $H_j(k)=0$ if $(j,k)\notin E$ and $H_j(k)=1$ otherwise. For each $j\in V$ and $k\in V\setminus\{j\}$, suppose the test $H_j(k)$ is based on a statistic $T_j(k)$. We assume for simplicity that the null hypothesis $H_j(k)=0$ is rejected at level $\alpha$ if and only if $T_j(k)\leq\alpha$, though the following results hold for more general significance regions \cite{storey}.

Fix $j\in V$ and let $\widehat{S}_1(j,\alpha) = \{k : T_j(k) \leq \alpha\}$ be the set of nodes whose null hypotheses $H_j(k)$ are rejected at level $\alpha$. As the notation suggests, $\widehat{S}_1(j,\alpha)$ is an estimator of $S_1(j)$. For $k\in V$ and an observed value $t_j(k)$ of the random test statistic $T_j(k)$, the associated \textit{$q$-value} is
\begin{align}\label{eq:qvalue}
	q_j(k) &= \inf_{\{\alpha : t_j(k) \leq \alpha\}} \mathrm{pFDR}(\alpha)
		= \inf_{\{\alpha : t_j(k) \leq \alpha\}} \mathbb{E}\bigg(\frac{\mathrm{FP}_j(\alpha)}{\lvert \widehat{S}_1(j,\alpha) \rvert}\ \bigg\vert\ \lvert \widehat{S}_1(j,\alpha) \rvert > 0\bigg),
\end{align}
where $\mathrm{FP}_j(\alpha) = \lvert \widehat{S}_1(j,\alpha) \setminus S_1(j) \rvert$ is the number of false positives in $\widehat{S}_1(j,\alpha)$, the second equality is the definition of the \textit{positive false discovery rate}, pFDR, and $\mathbb{E}(A\mid B)$ is the conditional expectation of $A$ given $B$. Intuitively, $q_j(k)$ is the smallest expected proportion of incorrectly rejected hypotheses among all rejected hypotheses, taken over all thresholds $\alpha$ for which the null hypothesis $H_j(k)$ is rejected. The following theorem shows that $q$-values admit a Bayesian interpretation.

\begin{theorem}\label{thrm:storey}
Fix $j\in V$ and suppose that hypothesis tests $H_j(k)$ are performed with the test statistics $T_j(k)$ for all $k\in V\setminus\{j\}$. Assume that the $(T_j(k), H_j(k))$ are independent and identically distributed (iid) with $T_j(k) \mid H_j(k) \sim (1 - H_j(k))F_{0j} + H_j(k) F_{1j}$ for some null distribution $F_{0j}$ and alternative distribution $F_{1j}$, and that $H_j(k)\sim\mathrm{Bernoulli}(\pi_1)$ are iid Bernoulli random variables for all $k\in V\setminus\{j\}$. Then for an observation $t_j(k)$ of $T_j(k)$,
\begin{align}\label{eq:posterior}
	q_j(k) &= \inf_{\{\alpha : t_j(k) \leq \alpha\}} \mathbb{P}\big(H_j(k) = 0 \mid T_j(k) \leq \alpha\big).
\end{align} 
\end{theorem}

\cref{thrm:storey} combines Theorem 1 and Corollary 2 from John Storey's seminal work on $q$-values \cite{storey}. This can be given a Bayesian interpretation in which the hypotheses $H_j(k)$ are given iid Bernoulli$(\pi_1)$ priors, and the statistics $T_j(k) \mid H_j(k)$ are modeled as iid draws from the mixture distribution $(1 - H_j(k))F_{0j} + H_j(k)F_{1j}$. Under this setup, $\mathbb{P}(H_j(k) = 0 \mid T_j(k) \leq \alpha)$ represents the posterior probability that a Type I error was made, given that the null was rejected when using a threshold of $\alpha$. Thus, the $q$-value $q_j(k)$ is the smallest posterior probability of a Type I error over all thresholds $\alpha$ greater than the observed test statistic $t_j(k)$. Equivalently, 
\begin{align}\label{eq:posterior_ball}
q_j(k) &= \inf_{\{\alpha : t_j(k) \leq \alpha\}}\mathbb{P}\big(k\notin S_1(j) \mid k \in \widehat{S}_1(j,\alpha)\big)
\end{align}
is the smallest posterior probability that $(j,k)$ is a false positive, given that $k$ is in the estimated neighborhood $\widehat{S}_1(j,\alpha)$ of node $j$.


\subsection{Generalization to paths}\label{sup_sec:paths}


Storey's $q$-values quantify uncertainty in individual edge estimates. In this section, we extend the concepts from \cref{sup_sec:qvalues} to quantify uncertainty in estimated paths. Let $V_0\subseteq V$ be a subset of variables of interest in the true graph $G=(V,E)$, and let $H_j(k)$ and $T_j(k)$ be hypotheses and test statistics, respectively, as defined above.

For $\alpha\in\mathbb{R}^p$, recursively define $\widehat{B}_0(V_0,\alpha) = V_0$ and, for each integer $s\geq 0$,
\begin{align*}
	\widehat{B}_{s+1}(V_0,\alpha) &= \bigcup_{j \in \widehat{B}_s(V_0,\alpha)} \widehat{B}_1(j,\alpha_j),
\end{align*}
where $\widehat{B}_1(j,\alpha_j) = \{j\}\cup\{k : T_j(k) \leq \alpha_j\}$ is the estimated ball of radius 1 around $j$. $\widehat{B}_r(V_0,\alpha)$ consists of all nodes that are estimated to be within distance $r$ of $V_0$ under the significance thresholds $\alpha=(\alpha_1,\dots,\alpha_p)$. Note that some $\alpha_j$ might be unused; however, taking $\alpha\in\mathbb{R}^p$ avoids the need to explicitly track which nodes appear in the recursive process. The corresponding edge set is
\begin{align*}
	\widehat{E}_r(V_0,\alpha) &= \bigcup_{s=0}^{r-1}\bigcup_{j\in\widehat{S}_s(V_0,\alpha)} \big\{(j,k) : k \in \widehat S_1(j,\alpha_j)\big\},
\end{align*}
where $\widehat{S}_0(V_0,\alpha)=V_0$ and, for $s\geq 0$, $\widehat{S}_{s+1}(V_0,\alpha) = \widehat{B}_{s+1}(V_0,\alpha)\setminus\widehat{B}_s(V_0,\alpha)$ is the set of nodes that are exactly distance $s+1$ from $V_0$ in the estimated graph. The graph $\widehat{G}_r(V_0,\alpha) = (\widehat{B}_r(V_0,\alpha),\widehat{E}_r(V_0,\alpha))$ is then an estimator of the true local graph $G_r(V_0)$.

Let $\mathcal{J}_r(V_0)$ and $\widehat{\mathcal{J}}_r(V_0,\alpha)$ be the set of paths starting in $V_0$ of length $r$ that lie in $G$ and $\widehat{G}_r(V_0,\alpha)$, respectively. A natural extension of \cref{eq:qvalue} to a path $\boldj$ of length $r$ starting in $V_0$ is
\begin{align}\label{eq:qvalue_path}
	q(\boldj) &= \inf_{\{\alpha : t(\boldj)\leq \alpha\}} \mathbb{E}\left(\frac{\lvert \widehat{\mathcal{J}}_r(V_0,\alpha) \setminus \mathcal{J}(V_0)\rvert}{\lvert \widehat{\mathcal{J}}_r(V_0,\alpha)\rvert}\ \bigg\vert\ \lvert \widehat{\mathcal{J}}_r(V_0,\alpha)\rvert > 0\right),
\end{align}
which is precisely \cref{eq:qvalue} when $\boldj=(j,k)$ consists of a single edge. A direct analogue of \cref{thrm:storey} seems unlikely to hold when $r>1$ unless no two paths in $\widehat{J}_r(V_0,\alpha)$ contain the same edge because extensions of the independence assumption in \cref{thrm:storey} to paths will be violated whenever two distinct paths in the estimated graph share an edge. For example, the test statistic and hypothesis pairs for distinct paths $(j,k,\ell)$ and $(j,k,\ell')$, namely $((T_j(k),H_j(k)),(T_k(\ell),H_k(\ell))$ and $((T_j(k),H_j(k)),(T_k(\ell'),H_k(\ell'))$, both contain $(T_j(k),H_j(k))$ and are therefore dependent.

While \cref{eq:qvalue_path} provides a natural extension of \cref{eq:qvalue} to paths, it depends on the full sets of estimated and true paths, $\widehat{\mathcal{J}}_r(V_0,\alpha)$ and $\mathcal{J}_r(V_0)$, making it difficult to analyze and interpret. Motivated by the Bayesian interpretation of $q$-values provided by \cref{thrm:storey}, we instead propose a more path-specific and interpretable measure of uncertainty that assesses the probability that a given path is not in the true graph. Specifically, for any path $\boldj = (j_0,\dots,j_r)$ starting in $V_0$ and observations $t_{j_s}(j_{s+1})$ of $T_{j_s}(j_{s+1})$, we define
\begin{align*}
	\widetilde{q}(\boldj) &= \inf_{\{\alpha : t(\boldj)\leq \alpha\}} \mathbb{P}\big(\boldj \notin \mathcal{J}(V_0) \mid \boldj \in \widehat{\mathcal{J}}_r(V_0,\alpha)\big),
\end{align*}
where $\{\alpha : t(\boldj)\leq \alpha\}$ is shorthand for $\{\alpha\in\mathbb{R}^p : t_{j_s}(j_{s+1}) \leq \alpha_{j_s}\ \text{for all}\ 0\leq s < r\}$ and $\mathcal{J}(V_0)$ is the set of all paths in $G$ starting in $V_0$. This is the smallest probability that $\boldj$ is not in the true graph, given that it is in $\widehat{G}_r(V_0,\alpha)$. \cref{thrm:main} states that the uncertainty in the estimated path $\boldj$ is at most the sum of the uncertainties $q_{j_s}(j_{s+1})$ in the neighborhood estimates along it. Its proof, provided below, essentially amounts to a union bound over the individual edges in the path.

\begin{theorem}\label{thrm:main}
Let $\boldj=(j_0,\dots,j_r)$ be a path starting in $V_0$, that is, $j_0\in V_0$ and $j_s\neq j_t$ for all $s\neq t$. Suppose that the assumptions of \cref{thrm:storey} hold for every $j_s$, $0\leq s < r$, and that $(T_{j_s}(j_{s+1}),H_{j_s}(j_{s+1}))$ are independent for all $0\leq s < r$. Then $\widetilde{q}(\boldj) \leq \sum_{s=0}^{r-1} q_{j_s}(j_{s+1})$.
\end{theorem}

\begin{proof}
Observe that $\boldj\notin\mathcal{J}(V_0)$ if and only if $H_{j_s}(j_{s+1}) = 0$ for some $0\leq s < r$. Also, $\boldj \in \widehat{\mathcal{J}}_r(V_0,\alpha)$ if and only if $T_{j_s}(j_{s+1})\leq\alpha_{j_s}$ for all $0\leq s < r$. Therefore,
\begin{align*}
	\widetilde{q}(\boldj) &= \inf_{\{\alpha : t(\boldj)\leq \alpha\}} \mathbb{P}\big(\boldj \notin \mathcal{J}(V_0) \mid \boldj \in \widehat{\mathcal{J}}_r(V_0,\alpha)\big) \\
		&= \inf_{\{\alpha : t(\boldj)\leq \alpha\}} \mathbb{P}\big(H_{j_s}(j_{s+1})=0\ \text{for some}\ 0\leq s<r \mid \boldj \in \widehat{\mathcal{J}}_r(V_0,\alpha)\big) \\
		&\leq \inf_{\{\alpha : t(\boldj)\leq \alpha\}} \sum_{s=0}^{r-1} \mathbb{P}\big(H_{j_s}(j_{s+1})=0 \mid \boldj \in \widehat{\mathcal{J}}_r(V_0,\alpha)\big) \\
		&= \inf_{\{\alpha : t(\boldj)\leq \alpha\}} \sum_{s=0}^{r-1} \mathbb{P}\big(H_{j_s}(j_{s+1})=0 \mid T_{j_s}(j_{s+1}) \leq \alpha_{j_s}\big) \\
		&= \sum_{s=0}^{r-1} \inf_{\{\alpha_{j_s} : t_{j_s}(j_{s+1})\leq\alpha_{j_s}\}}\mathbb{P}\big(H_{j_s}(j_{s+1})=0 \mid T_{j_s}(j_{s+1}) \leq \alpha_{j_s}\big) \\
		&= \sum_{s=0}^{r-1} q_{j_s}(j_{s+1}).
\end{align*}
The inequality is a union bound. The third equality holds because the $(T_{j_s}(j_{s+1}),H_{j_s}(j_{s+1}))$ are independent by assumption; specifically, for independent random variables $(X_1,Y_1),\dots,(X_n,Y_n)$, it is straightforward to verify that $Y_1$ is conditionally independent of $X_2,\dots,X_n$ given $X_1$. The infimum and sum are exchangeable in the fourth equality because each summand only depends on $\alpha_{j_s}$ (not the other $\alpha_{j_t}$), and the last equality holds by \cref{thrm:storey}.
\end{proof}


\subsection{Additional algorithmic details}\label{sup_sec:implementation}


In Line 6 of \cref{alg:pfs}, the entries $Q_{jk}$ and $Q_{kj}$ are set to the minimum of $q_j(k)$ and $q_k(j)$. This \textit{minimum rule} preserves symmetry and is analogous to the ``OR'' rule in nodewise graphical models \cite{mb_graph}, where an edge is included in the estimated graph if either node selects the other as a neighbor. As an alternative, a more conservative \textit{forward rule} may be applied, assigning $Q_{jk} = Q_{kj}$ according to the $q$-value of the node closer to the initial target set $V_0$ (still applying the minimum if both nodes lie in the same neighborhood layer). For example, if $j$ and $k$ are distances $r$ and $r+1$ from $V_0$, respectively, then $Q_{jk}$ and $Q_{kj}$ are set to $q_j(k)$.


\subsection{Computation}\label{sup_sec:computation}


The total runtime of PFS is roughly $p_{\mathrm{sel}}T_{\text{fs}}$, where $p_{\mathrm{sel}}$ is the number of nodes whose neighborhoods are estimated, and $T_{\text{fs}}$ is the time required for one run of the underlying feature selection method. The quantity $p_{\mathrm{sel}}$ depends on the data and the user-specified false discovery thresholds. Liberal thresholds lead to denser neighborhoods and more follow-up selections, while stricter thresholds yield sparser graphs and reduced computation. A practical strategy is to begin with $r_{\max} = 1$ and a liberal threshold to assess neighborhood sparsity around the target variables. Once an appropriate $q_1^*$ is chosen, increase to $r_{\max} = 2$ and rerun PFS using the updated $q_1^*$ along with a liberal value of $q_2^*$ to assess sparsity at radius 2. This process can be repeated iteratively to tune each $q_r^*$ and control the size of the estimated local graph.

In the environmental health study, for example, PFS was run for 15 nodes (two target nodes plus the 13 nodes in their combined neighborhoods), with each IPSS run taking about 60 seconds, yielding a total runtime of approximately 15 minutes. In the breast cancer study, PFS was run for 88 nodes, with each IPSS run taking approximately 35 seconds, yielding a total runtime of approximately 50 minutes. These runtimes are consistent with previous work, which showed that individual IPSS runs complete in under 20 seconds on datasets with 500 samples and 5000 features~\cite{ipss_nonparametric}. In contrast, two full graph estimation methods with global FDR control, desparsified nodewise scaled lasso~\cite{dsnwsl} (DSNWSL) and Gaussian graphical models with FDR control using scaled lasso~\cite{fdr_liu} (GFCSL), took approximately $4.5$ and $5.25$ hours on the same breast cancer dataset, respectively. The other three full graph methods in the R package \texttt{SILGGM}~\cite{silggm} were computationally infeasible in this setting, as were many of the methods from the R package \texttt{bnlearn}~\cite{bnlearn}.

Table~\ref{tab:runtimes} shows that PFS with lasso (PFS(L1)) and with gradient boosting (PFS(GB)) scale favorably with dimension relative to many other graph estimation methods. Both PFS variants maintain runtimes on the order of seconds to a few minutes even at $p=8000$, whereas many full-graph and Markov blanket-based algorithms exceed the imposed two-hour limit well before this regime.

All PFS runtimes reported in this work are based on single-threaded computation, but substantial speedups are easily achievable through parallelization. At each iteration, nodewise selections are performed independently for all nodes in the current layer, enabling straightforward parallel execution. Once neighborhoods are estimated, the pruning step---implemented via the function \texttt{prune\_graph} in the \texttt{localgraph} package---is computationally negligible and can be used to further sparsify the graph without additional feature selection. IPSS is also trivially parallelizable, as each resampling step is independent, allowing for further reductions in runtime~\cite{ipss_nonparametric}.

\setlength{\tabcolsep}{10pt}
\begin{table*}[ht]
\centering
\begin{tabular}{lccccccc}
\toprule
Method & 125 & 250 & 500 & 1000 & 2000 & 4000 & 8000 \\
\midrule
\textbf{PFS}(L1) & 0:00:06 & 0:00:15 & 0:00:13 & 0:00:02 & 0:00:16 & 0:00:24 & 0:00:13 \\
ARACNE & 0:00:00 & 0:00:00 & 0:00:00 & 0:00:00 & 0:00:01 & 0:00:06 & 0:00:28 \\
\textbf{PFS}(GB) & 0:01:10 & 0:00:52 & 0:01:19 & 0:00:38 & 0:03:27 & 0:03:13 & 0:02:24 \\
NLasso & 0:00:02 & 0:00:03 & 0:00:07 & 0:00:29 & 0:01:45 & 0:07:35 & 0:26:19 \\
HPC(L) & 0:00:00 & 0:00:01 & 0:00:54 & 0:00:23 & 0:03:50 & 0:13:42 & 0:35:37 \\
GFCSL & 0:00:00 & 0:00:00 & 0:00:03 & 0:00:14 & 0:01:10 & 0:11:21 & 1:00:46 \\
DSNWSL & 0:00:00 & 0:00:00 & 0:00:03 & 0:00:11 & 0:01:01 & 0:15:46 & $>$2hrs \\
GLasso & 0:00:02 & 0:00:03 & 0:00:08 & 0:00:34 & 0:04:03 & 0:31:55 & $>$2hrs \\
MMPC(L) & 0:00:00 & 0:00:00 & 0:00:02 & 0:00:21 & 0:03:00 & 0:48:32 & $>$2hrs \\
SIHPC(L) & 0:00:00 & 0:00:00 & 0:00:02 & 0:00:27 & 0:04:34 & 1:20:06 & $>$2hrs \\
MMPC & 0:00:03 & 0:00:12 & 0:00:59 & 0:06:41 & 0:41:10 & $>$2hrs & $>$2hrs \\
SIHPC & 0:00:01 & 0:00:05 & 0:00:28 & 0:05:51 & 1:26:06 & $>$2hrs & $>$2hrs \\
HPC & 0:00:09 & 0:01:37 & 0:05:24 & 0:28:55 & $>$2hrs & $>$2hrs & $>$2hrs \\
IAMB(L) & 0:00:00 & 0:00:17 & 0:16:00 & $>$2hrs & $>$2hrs & $>$2hrs & $>$2hrs \\
StablePC & 0:00:06 & 0:01:07 & 0:16:06 & $>$2hrs & $>$2hrs & $>$2hrs & $>$2hrs \\
FastIAMB(L) & 0:00:03 & 0:00:37 & 0:28:44 & $>$2hrs & $>$2hrs & $>$2hrs & $>$2hrs \\
FastIAMB & 0:00:08 & 0:01:12 & 1:25:45 & $>$2hrs & $>$2hrs & $>$2hrs & $>$2hrs \\
\bottomrule
\end{tabular}
\caption{\textit{Runtime versus dimension.} Wall-clock runtimes (H:MM:SS) for a single run of different graph estimation methods on simulated data as the number of variables increases from $p=125$ to $8000$ with fixed sample size $n = 500$. Data are generated from a Gaussian graphical model with a random sparse precision matrix. There is one target variable. All methods are run with default settings; PFS(L1) and PFS(GB) denote PFS using IPSS with lasso and gradient-boosting, respectively. PFS parameters are always set to $r_{\max} = 3$ and $q_{\max} = 0.25$. Entries marked $>$2\,hrs exceed the imposed time limit of 2 hours. All other methods are described in \cref{sup_sec:method_settings}. Methods are ordered by runtime at $p = 8000$.}
\label{tab:runtimes}
\end{table*}

\clearpage


\section{Other methods: Implementation details and settings}\label{sup_sec:other_method_details}


In this section, we describe the other graph estimation methods to which PFS is compared in this work. As discussed in \cref{sec:intro}, the following methods belong to one of two classes: regularization-based approaches or constraint-based and mutual-information based approaches. Throughout this work, all method parameters are set to their software-specific defaults unless otherwise specified.

\textit{Regularization-based approaches}. These methods assume data are Gaussian. The graphical lasso~\cite{glasso} and the nodewise lasso~\cite{mb_graph} are implemented using the R package \texttt{huge}~\cite{huge}. The bivariate nodewise scaled lasso~\cite{bnwsl} (BNWSL), desparsified graphical lasso~\cite{dsgl} (DSGL), desparsified nodewise scaled lasso~\cite{dsnwsl} (DSNWSL), and Gaussian graphical models with FDR control based on lasso or scaled lasso~\cite{fdr_liu} (GFCL, GFCSL) are implemented using the R package \texttt{SILGGM}~\cite{silggm}. The graphical and nodewise lasso do not provide any form of error control; the five \texttt{SILGGM} methods provide global FDR control. In several applications, custom nominal FDR values were used for the \texttt{SILGGM} methods to yield better results.

\textit{Constraint-based and mutual information-based methods}. Constraint-based and mutual information-based structure learning algorithms are implemented using the R package \texttt{bnlearn}~\cite{bnlearn}. These include ARACNE~\cite{bnlearn_aracne}, hybrid parents and children~\cite{bnlearn_hpc} (HPC), incremental association Markov blanket~\cite{bnlearn_iamb} (IAMB) and its fast variant FastIAMB, max-min parents and children~\cite{bnlearn_mmpc} (MMPC), Hiton parents and children~\cite{bnlearn_sihpc} (SIHPC), and the stable PC algorithm~\cite{bnlearn_stablepc} (StablePC).

In addition to their global implementations, several of these procedures operate nodewise and can be adapted to estimate neighborhoods via the \texttt{learn.nbr} or \texttt{learn.mb} functions in \texttt{bnlearn}; these neighborhoods can then be combined to estimate local graphs. We denote these local variants with a trailing ``(L)", for example HPC(L). Methods based on \texttt{learn.mb} estimate Markov blankets, whereas those using \texttt{learn.nbr} estimate parents-and-children (PC) sets in a directed graph. Markov blanket methods, namely IAMB and FastIAMB, are typically substantially slower than PC-based methods, and both classes scale poorly to high dimensions relative to PFS (\cref{tab:runtimes}). PC-based algorithms do not directly estimate conditional independence graphs; however, we include them because their skeletons can be interpreted as undirected graphs and they are widely used in practice. Unlike PFS, none of these approaches provide edgewise or pathwise false discovery control.


\section{Additional details: Simulation studies}\label{sup_sec:simulations}


We conduct five simulation studies to evaluate the performance of PFS relative to existing methods across regimes that vary in sparsity, linearity, and sample size. Specifically, we consider: sparse, linear graphs with $n=100$ (\cref{tab:linear_sparse_n100}); sparse, nonlinear graphs with $n=100$ (\cref{tab:nonlinear_sparse_n100}); dense, linear graphs with $n=100$ (\cref{tab:linear_dense_n100}); dense, linear graphs with $n=500$ (\cref{tab:linear_dense_n500}); and dense, nonlinear graphs with $n=500$ (\cref{tab:nonlinear_dense_n500}). These settings are chosen to assess (i) the behavior of local graph estimators under ideal Gaussian assumptions, (ii) robustness to nonlinear relationships, and (iii) performance when the true local structure around the target variable is dense.

In each study, there are $p=200$ variables and a single target variable. Each table reports the true positive rate (TPR) and false discovery rate (FDR) (\cref{eq:tpr,eq:fdp}) for local graph recovery at multiple radii, averaged over $100$ independent trials. In \cref{fig:varying_n}, we further examine the performance of PFS as the sample size $n$ varies. Full details of the data-generating mechanisms and parameter choices for each setting are provided below.


\subsection{Simulation designs}\label{sup_sec:simulation_designs}


All simulation studies are based on block-structured conditional independence graphs with a single target variable and $p=200$ total variables. Data are generated from Gaussian graphical models with precision matrix $\Theta$ (\cref{sup_sec:limitations}), optionally followed by a nonlinear transformation of the target variable.

\textit{Block-structured precision matrices}.
In all settings, the precision matrix $\Theta$ is constructed with three blocks: the first contains only the target variable, the second contains the target's immediate neighbors, and the third block contains all remaining variables. Nonzero entries within and between blocks are generated at random according to prespecified degrees, with values chosen uniformly from $\{-1,1\}$. The resulting matrix is symmetrized, a positive constant is added to the diagonal to ensure positive definiteness, and the matrix is rescaled so that its eigenvalues lie between $0.01$ and $10$, allowing for strong correlations between variables as is common in practice.

\textit{Sparse, linear, $n=100$}.
In the sparse linear setting (\cref{tab:linear_sparse_n100}), the three blocks of $\Theta$ are sizes $1\times 1$ (target variable), $4\times 4$ (immediate neighbors of the target), and $195\times 195$ (the remaining variables). There are no edges between the four neighbors of the target. Each neighbor is connected to the third block with average degree six, while variables in the third block have average degree two. In each of the $100$ trials, a new $\Theta$ is generated and $n=100$ samples are drawn independently from $\mathcal{N}(0,\Theta^{-1})$.

\textit{Sparse, nonlinear, $n=100$}.
In the sparse nonlinear setting (\cref{tab:nonlinear_sparse_n100}), the block structure is identical to the sparse linear case, except that the connectivity between the target's neighbors and the third block is reduced to an average degree of two. The only other difference between this setting and the sparse, linear, $n=100$ case is that after sampling $n=100$ observations from $\mathcal{N}(0,\Theta^{-1})$, the target variable $X_1$ is redefined as a nonlinear function of its true neighbors:
\begin{align}\label{eq:nonlinear_target}
X_1 &= \sum_{j \in S_1(X_1)} \exp(-X_j^2/2) + \varepsilon,
\end{align}
where $\varepsilon$ is mean-zero Gaussian noise scaled to achieve a signal-to-noise ratio---defined as the variance of the signal divided by the variance of the noise---of $4$. This form of nonlinearity is motivated by scenarios in medicine and epidemiology where both high and low levels of certain covariates can be harmful, making the response sensitive to deviations from a normal range.

\textit{Dense, linear, $n=100$}.
In the dense linear setting (\cref{tab:linear_dense_n100}), the three blocks of $\Theta$ are sizes $1\times 1$ (target variable), $20\times 20$ (immediate neighbors of the target), and $179\times 179$ (the remaining variables). The neighbor block has average internal degree five, and variables in the third block have average degree ten. Each of the 20 immediate neighbors of the target is connected to the third block with average degree five. As before, in each of the $100$ trials, a new $\Theta$ is generated and $n=100$ samples are drawn independently from $\mathcal{N}(0,\Theta^{-1})$.

\textit{Dense, linear, $n=500$}.
In the dense linear setting with increased sample size (\cref{tab:linear_dense_n500}), the block structure and connectivity are identical to the dense, linear, $n=100$ case. In each of the $100$ trials, a new $\Theta$ is generated and $n=500$ samples are drawn independently from $\mathcal{N}(0,\Theta^{-1})$.

\textit{Dense, nonlinear, $n=500$}.
In the dense nonlinear setting (\cref{tab:nonlinear_dense_n500}), the block structure and connectivity are identical to the dense, linear, $n=500$ case. The only other difference is that after sampling $n=500$ observations from $\mathcal{N}(0,\Theta^{-1})$, the target variable $X_1$ is redefined as a nonlinear function of its true neighbors according to \cref{eq:nonlinear_target}, where $\varepsilon$ is again mean-zero Gaussian noise, scaled to achieve a signal-to-noise ratio of $4$.


\subsection{Method settings}\label{sup_sec:method_settings}


All methods are run using default settings from their respective software packages unless otherwise specified. PFS is implemented using integrated path stability selection (IPSS) via the Python package \texttt{ipss}~\cite{ipss,ipss_nonparametric}. In linear settings, IPSS uses lasso~\cite{lasso}, while in nonlinear settings it uses variable importance scores from gradient boosting~\cite{xgboost}. The maximum pathwise $q$-value threshold is set to $q_{\mathrm{path}}^{*}=0.2$. Local neighborhood $q$-value thresholds $(q_1^*, q_2^*, \dots)$ are set to $(0.2,0.1,0.1,0.1)$ in linear settings and $(0.2,0.05,0.05,0.05)$ in nonlinear settings. For the \texttt{SILGGM} methods, the target FDR is set to $0.1$. For \texttt{bnlearn} methods, correlation-based tests are used in linear settings and mutual information-based tests are used in nonlinear settings. 


\subsection{Simulation results}\label{sup_sec:simulation_results}


The results of the five simulation studies are reported in \cref{tab:linear_sparse_n100,tab:nonlinear_sparse_n100,tab:linear_dense_n100,tab:linear_dense_n500,tab:nonlinear_dense_n500}. The false discovery rates for each method at each radius are the average local false discovery proportions (\cref{eq:fdp}) over the $100$ trials. For an estimated local graph $\widehat{G}_r(V_0) = (\widehat{B}_r(V_0), \widehat{E}_r(V_0))$, the local \textit{true positive rate} (TPR) is the proportion of true edges in the estimated local graph of radius $r$ among all edges in the true local graph $E_r(V_0)$:
\begin{align}\label{eq:tpr}
\mathrm{TPR}(\widehat{E}_r(V_0)) = \frac{|\widehat{E}_r(V_0) \cap E_r(V_0)|}{|E_r(V_0)|}.
\end{align}

\begin{table}[ht]
\centering
\begin{tabular}{lccccccccc}
\toprule
 & \multicolumn{4}{c}{TPR} & & \multicolumn{4}{c}{FDR} \\
\cmidrule(lr){2-5}\cmidrule(lr){7-10}
Method & $r=1$ & $r=2$ & $r=3$ & $r=4$ &  & $r=1$ & $r=2$ & $r=3$ & $r=4$ \\
\midrule
\textbf{PFS} & 0.85 & 0.59 & 0.31 & 0.19 &  & 0.16 & 0.07 & 0.07 & 0.07 \\
GLasso & 0.92 & 0.95 & 0.90 & 0.80 &  & 0.90 & 0.95 & 0.92 & 0.89 \\
NLasso & 0.90 & 0.87 & 0.79 & 0.71 &  & 0.43 & 0.43 & 0.47 & 0.39 \\
BNWSL & 0.34 & 0.15 & 0.09 & 0.06 &  & 0.01 & 0.02 & 0.02 & 0.02 \\
DSGL & 0.03 & 0.02 & 0.01 & 0.01 &  & 0.00 & 0.00 & 0.00 & 0.00 \\
DSNWSL & 0.60 & 0.40 & 0.28 & 0.22 &  & 0.01 & 0.00 & 0.01 & 0.01 \\
GFCL & 0.92 & 0.84 & 0.71 & 0.58 &  & 0.38 & 0.40 & 0.36 & 0.29 \\
GFCSL & 0.60 & 0.40 & 0.29 & 0.25 &  & 0.09 & 0.10 & 0.13 & 0.14 \\
ARACNE & 0.84 & 0.77 & 0.68 & 0.68 &  & 0.15 & 0.17 & 0.38 & 0.64 \\
HPC & 0.79 & 0.67 & 0.54 & 0.47 &  & 0.02 & 0.03 & 0.07 & 0.13 \\
MMPC & 0.80 & 0.65 & 0.57 & 0.55 &  & 0.08 & 0.08 & 0.21 & 0.36 \\
SIHPC & 0.79 & 0.63 & 0.55 & 0.52 &  & 0.06 & 0.06 & 0.18 & 0.34 \\
HPC(L) & 0.82 & 0.71 & 0.59 & 0.55 &  & 0.07 & 0.08 & 0.22 & 0.32 \\
MMPC(L) & 0.88 & 0.74 & 0.69 & 0.74 &  & 0.20 & 0.26 & 0.54 & 0.68 \\
SIHPC(L) & 0.87 & 0.71 & 0.66 & 0.70 &  & 0.17 & 0.20 & 0.46 & 0.62 \\
FastIAMB(L) & 0.09 & 0.04 & -- & -- &  & 0.20 & 0.22 & -- & -- \\
IAMB(L) & 0.86 & 0.78 & -- & -- &  & 0.93 & 0.99 & -- & -- \\
\bottomrule
\end{tabular}
\caption{\textit{(Linear, sparse, $n = 100$)}. True positive rate (TPR) and false discovery rate (FDR) for local graph recovery across methods (rows) and neighborhood radii (columns), averaged over 100 trials. Methods with a trailing (L) are local versions of their global counterparts; for example, HPC(L) is the local version of HPC. FastIAMB(L) and IAMB(L) were only run to radius 2 due to computational constraints.}

\label{tab:linear_sparse_n100}
\end{table}

\begin{table}[ht]
\centering
\begin{tabular}{lccccccccc}
\toprule
 & \multicolumn{4}{c}{TPR} & & \multicolumn{4}{c}{FDR} \\
\cmidrule(lr){2-5}\cmidrule(lr){7-10}
Method & $r=1$ & $r=2$ & $r=3$ & $r=4$ &  & $r=1$ & $r=2$ & $r=3$ & $r=4$ \\
\midrule
\textbf{PFS} & 0.83 & 0.47 & 0.30 & 0.20 &  & 0.19 & 0.22 & 0.23 & 0.22 \\
GLasso & 0.00 & 0.00 & 0.00 & 0.01 &  & 0.05 & 0.05 & 0.05 & 0.05 \\
NLasso & 0.00 & 0.00 & 0.00 & 0.00 &  & 0.05 & 0.05 & 0.04 & 0.05 \\
BNWSL & 0.00 & 0.00 & 0.00 & 0.00 &  & 0.04 & 0.04 & 0.04 & 0.04 \\
DSGL & 0.00 & 0.00 & 0.00 & 0.00 &  & 0.03 & 0.03 & 0.03 & 0.03 \\
DSNWSL & 0.00 & 0.00 & 0.00 & 0.00 &  & 0.07 & 0.07 & 0.06 & 0.06 \\
GFCL & 0.00 & 0.00 & 0.01 & 0.01 &  & 0.23 & 0.21 & 0.20 & 0.19 \\
GFCSL & 0.00 & 0.00 & 0.01 & 0.01 &  & 0.17 & 0.17 & 0.15 & 0.13 \\
ARACNE & 0.03 & 0.04 & 0.23 & 0.59 &  & 0.98 & 0.99 & 0.98 & 0.94 \\
HPC & 0.00 & 0.00 & 0.01 & 0.02 &  & 0.34 & 0.33 & 0.31 & 0.27 \\
MMPC & 0.01 & 0.02 & 0.06 & 0.18 &  & 0.91 & 0.90 & 0.87 & 0.81 \\
SIHPC & 0.01 & 0.02 & 0.06 & 0.16 &  & 0.90 & 0.90 & 0.87 & 0.81 \\
HPC(L) & 0.00 & 0.00 & 0.02 & 0.07 &  & 0.52 & 0.51 & 0.48 & 0.42 \\
MMPC(L) & 0.02 & 0.02 & 0.16 & 0.49 &  & 0.99 & 0.99 & 0.96 & 0.91 \\
SIHPC(L) & 0.02 & 0.02 & 0.14 & 0.45 &  & 0.98 & 0.98 & 0.96 & 0.91 \\
\bottomrule
\end{tabular}
\caption{\textit{(Nonlinear, sparse, $n = 100$)}. True positive rate (TPR) and false discovery rate (FDR) for local graph recovery across methods (rows) and neighborhood radii (columns), averaged over 100 trials. Methods with a trailing (L) are local versions of their global counterparts; for example, HPC(L) is the local version of HPC.}

\label{tab:nonlinear_sparse_n100}
\end{table}

\begin{table}[ht]
\centering
\begin{tabular}{lccccccc}
\toprule
 & \multicolumn{3}{c}{TPR} & & \multicolumn{3}{c}{FDR} \\
\cmidrule(lr){2-4}\cmidrule(lr){6-8}
Method & $r=1$ & $r=2$ & $r=3$ &  & $r=1$ & $r=2$ & $r=3$ \\
\midrule
\textbf{PFS} & 0.43 & 0.13 & 0.05 &  & 0.19 & 0.21 & 0.18 \\
GLasso & 0.57 & 0.55 & 0.52 &  & 0.62 & 0.87 & 0.80 \\
NLasso & 0.42 & 0.23 & 0.21 &  & 0.29 & 0.56 & 0.49 \\
BNWSL & 0.01 & 0.00 & 0.00 &  & 0.00 & 0.00 & 0.00 \\
DSGL & 0.00 & 0.00 & 0.00 &  & 0.00 & 0.00 & 0.00 \\
DSNWSL & 0.07 & 0.01 & 0.00 &  & 0.02 & 0.03 & 0.03 \\
GFCL & 0.17 & 0.05 & 0.02 &  & 0.07 & 0.14 & 0.12 \\
GFCSL & 0.16 & 0.05 & 0.02 &  & 0.13 & 0.22 & 0.20 \\
ARACNE & 0.26 & 0.08 & 0.05 &  & 0.20 & 0.38 & 0.37 \\
HPC & 0.17 & 0.05 & 0.02 &  & 0.06 & 0.15 & 0.16 \\
MMPC & 0.19 & 0.06 & 0.03 &  & 0.12 & 0.26 & 0.28 \\
SIHPC & 0.16 & 0.05 & 0.02 &  & 0.08 & 0.21 & 0.24 \\
HPC(L) & 0.22 & 0.07 & 0.05 &  & 0.12 & 0.28 & 0.27 \\
MMPC(L) & 0.24 & 0.10 & 0.09 &  & 0.20 & 0.46 & 0.43 \\
SIHPC(L) & 0.21 & 0.08 & 0.06 &  & 0.16 & 0.37 & 0.39 \\
\bottomrule
\end{tabular}
\caption{\textit{(Linear, dense, $n = 100$)}. True positive rate (TPR) and false discovery rate (FDR) for local graph recovery across methods (rows) and neighborhood radii (columns), averaged over 100 trials. Methods with a trailing (L) are local versions of their global counterparts; for example, HPC(L) is the local version of HPC.}

\label{tab:linear_dense_n100}
\end{table}

\begin{table}[ht]
\centering
\begin{tabular}{lccccccc}
\toprule
 & \multicolumn{3}{c}{TPR} & & \multicolumn{3}{c}{FDR} \\
\cmidrule(lr){2-4}\cmidrule(lr){6-8}
Method & $r=1$ & $r=2$ & $r=3$ &  & $r=1$ & $r=2$ & $r=3$ \\
\midrule
\textbf{PFS} & 0.87 & 0.59 & 0.51 &  & 0.12 & 0.16 & 0.09 \\
GLasso & 0.71 & 0.77 & 0.91 &  & 0.55 & 0.83 & 0.70 \\
NLasso & 0.63 & 0.51 & 0.60 &  & 0.22 & 0.46 & 0.34 \\
BNWSL & 0.79 & 0.76 & 0.86 &  & 0.06 & 0.15 & 0.10 \\
DSGL & 0.56 & 0.45 & 0.45 &  & 0.01 & 0.03 & 0.02 \\
DSNWSL & 0.75 & 0.68 & 0.72 &  & 0.01 & 0.05 & 0.04 \\
GFCL & 0.81 & 0.76 & 0.82 &  & 0.05 & 0.14 & 0.09 \\
GFCSL & 0.82 & 0.78 & 0.87 &  & 0.06 & 0.15 & 0.10 \\
ARACNE & 0.33 & 0.11 & 0.07 &  & 0.09 & 0.18 & 0.16 \\
HPC & 0.44 & 0.30 & 0.27 &  & 0.02 & 0.03 & 0.02 \\
MMPC & 0.40 & 0.25 & 0.21 &  & 0.03 & 0.07 & 0.07 \\
HPC(L) & 0.48 & 0.35 & 0.34 &  & 0.03 & 0.06 & 0.04 \\
MMPC(L) & 0.46 & 0.34 & 0.36 &  & 0.08 & 0.21 & 0.19 \\
SIHPC(L) & 0.42 & 0.29 & 0.29 &  & 0.05 & 0.15 & 0.15 \\
\bottomrule
\end{tabular}
\caption{\textit{(Linear, dense, $n = 500$)}. True positive rate (TPR) and false discovery rate (FDR) for local graph recovery across methods (rows) and neighborhood radii (columns), averaged over 100 trials. Methods with a trailing (L) are local versions of their global counterparts; for example, HPC(L) is the local version of HPC.}

\label{tab:linear_dense_n500}
\end{table}

\begin{table}[ht]
\centering
\begin{tabular}{lccccccc}
\toprule
 & \multicolumn{3}{c}{TPR} & & \multicolumn{3}{c}{FDR} \\
\cmidrule(lr){2-4}\cmidrule(lr){6-8}
Method & $r=1$ & $r=2$ & $r=3$ &  & $r=1$ & $r=2$ & $r=3$ \\
\midrule
\textbf{PFS} & 0.76 & 0.48 & 0.39 &  & 0.11 & 0.24 & 0.14 \\
GLasso & 0.00 & 0.00 & 0.02 &  & 0.04 & 0.04 & 0.03 \\
NLasso & 0.00 & 0.00 & 0.01 &  & 0.04 & 0.03 & 0.01 \\
BNWSL & 0.01 & 0.02 & 0.12 &  & 0.52 & 0.44 & 0.07 \\
DSGL & 0.00 & 0.00 & 0.01 &  & 0.13 & 0.11 & 0.01 \\
DSNWSL & 0.00 & 0.01 & 0.05 &  & 0.33 & 0.27 & 0.02 \\
GFCL & 0.01 & 0.02 & 0.10 &  & 0.47 & 0.40 & 0.06 \\
GFCSL & 0.01 & 0.02 & 0.13 &  & 0.54 & 0.46 & 0.07 \\
ARACNE & 0.01 & 0.01 & 0.02 &  & 0.91 & 0.82 & 0.22 \\
HPC & 0.00 & 0.00 & 0.01 &  & 0.19 & 0.17 & 0.01 \\
MMPC & 0.01 & 0.01 & 0.06 &  & 0.76 & 0.68 & 0.13 \\
HPC(L) & 0.01 & 0.01 & 0.07 &  & 0.76 & 0.66 & 0.06 \\
MMPC(L) & 0.02 & 0.02 & 0.15 &  & 0.91 & 0.83 & 0.21 \\
SIHPC(L) & 0.01 & 0.01 & 0.12 &  & 0.92 & 0.82 & 0.18 \\
\bottomrule
\end{tabular}
\caption{\textit{(Nonlinear, dense, $n = 500$)}. True positive rate (TPR) and false discovery rate (FDR) for local graph recovery across methods (rows) and neighborhood radii (columns), averaged over 100 trials. Methods with a trailing (L) are local versions of their global counterparts; for example, HPC(L) is the local version of HPC.}

\label{tab:nonlinear_dense_n500}
\end{table}

\ifthenelse{\boolean{showfigures}}{
\begin{figure*}[ht!]
\begin{adjustbox}{center}
\includegraphics[
	height=0.275\textheight, 
	width=\textwidth,
	trim={0pt 0pt 0pt 0pt}, clip]
	{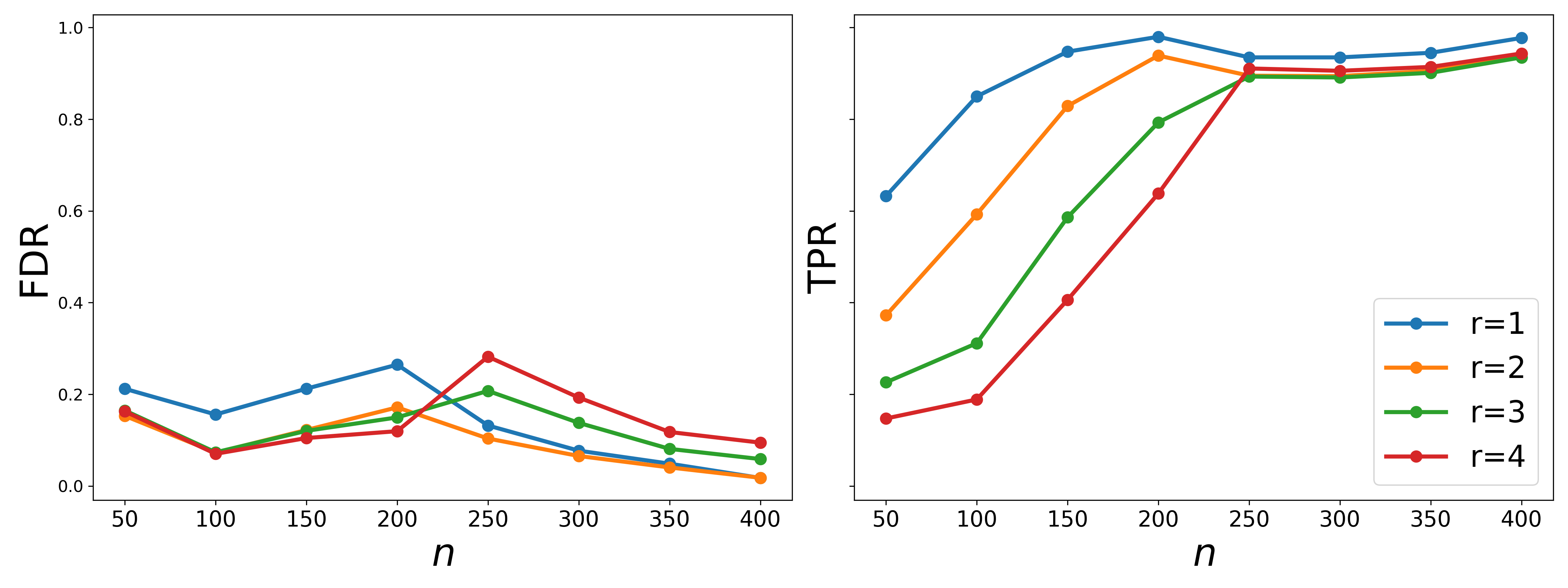}
\end{adjustbox}
\caption{\textit{PFS performance as a function of sample size}. FDR (left) and TPR (right) for local graphs estimated by PFS as the sample size $n$ ranges over $\{50,100,150,200,250,300,350,400\}$. All other simulation settings follow the sparse, linear design in \cref{sup_sec:simulation_designs}. Curves correspond to radii $r \in \{1,2,3,4\}$; reported FDR and TPR values are averaged over 100 simulation trials.}
\label{fig:varying_n}
\end{figure*}
}{}

\clearpage


\subsection{PFS with other selection methods}\label{sup_sec:qvalue_methods}


PFS can be combined with any FDR-controlling procedure, since such procedures can be used to construct $q$-values. Throughout this work, we use integrated path stability selection (IPSS) for this purpose due to its favorable properties (see Methods) and strong performance relative to other variable selection methods in previous studies~\cite{ipss,ipss_nonparametric}. In this section, we report PFS results using $q$-values from other approaches in the same five simulation settings discussed above: (i) linear, sparse, $n=100$; (ii) nonlinear, sparse, $n=100$; (iii) linear, dense, $n=100$; (iv) linear, dense, $n=500$; (v) nonlinear, dense, $n=500$. In addition to IPSS, we implemented PFS with the Benjamini--Hochberg (BH) method~\cite{bh}, the Benjamini--Yekutieli (BY) method~\cite{by}, and three versions of model-X knockoffs~\cite{modelX}: knockoffs with importance statistics from generalized linear models (KOGLM), the lasso (KOL1), and random forests (KORF). Unlike IPSS and knockoffs, BH and BY require $p$-values; in our experiments, we compute $p$-values using ordinary least squares (OLS) regression when $n>p$. When $n<p$, we compute $p$-values using marginal, univariate regressions, which are widely used in practice, though they reflect marginal rather than conditional associations. BH and BY are implemented with the Python package \texttt{statsmodels}~\cite{statsmodels}. The knockoff methods are implemented with the R package \texttt{knockoff}.

Tables \ref{tab:qval_linear_sparse_n100}--\ref{tab:qval_nonlinear_dense_n500} show that IPSS consistently outperforms other $q$-value constructions across simulation settings. In all experiments, the PFS framework is held fixed across methods: each method uses the same pathwise $q$-value threshold $q_{\mathrm{path}}^*$ and local thresholds $(q_1^*, q_2^*, \dots)$ specified in \cref{sup_sec:method_settings}. In the two nonlinear experiments, PFS with IPSS significantly outperforms PFS with other selection methods---including KORF, which uses random forests and is therefore also capable of capturing nonlinear relationships---maintaining the highest TPR across all radii while controlling the FDR in both sparse and dense cases. In the two linear, $n=100$ settings, PFS with BH and BY achieve high TPR but also exceedingly high FDR, while in the linear, dense, $n=500$ experiment, both methods control FDR but have a lower TPR than PFS with IPSS. PFS with the three knockoff approaches has lower TPR and higher FDR than IPSS in the linear, sparse, $n=100$ experiment; similar or lower TPR and higher FDR than IPSS in the linear, dense, $n=100$ experiment; and lower TPR and slightly lower FDR than IPSS in the linear, dense, $n=500$ experiment.

\begin{table}[ht]
\centering
\begin{tabular}{lccccccccc}
\toprule
 & \multicolumn{4}{c}{TPR} & & \multicolumn{4}{c}{FDR} \\
\cmidrule(lr){2-5}\cmidrule(lr){7-10}
Method & $r=1$ & $r=2$ & $r=3$ & $r=4$ &  & $r=1$ & $r=2$ & $r=3$ & $r=4$ \\
\midrule
\textbf{IPSS} & 0.85 & 0.59 & 0.31 & 0.19 &  & 0.16 & 0.07 & 0.07 & 0.07 \\
KOGLM & 0.47 & 0.11 & 0.05 & 0.03 &  & 0.29 & 0.26 & 0.26 & 0.26 \\
KOL1 & 0.43 & 0.11 & 0.05 & 0.03 &  & 0.25 & 0.23 & 0.23 & 0.23 \\
KORF & 0.33 & 0.11 & 0.06 & 0.03 &  & 0.46 & 0.46 & 0.45 & 0.45 \\
BH & 0.97 & 0.96 & 0.95 & 0.92 &  & 0.96 & 0.99 & 0.98 & 0.97 \\
BY & 0.96 & 0.95 & 0.91 & 0.86 &  & 0.94 & 0.98 & 0.98 & 0.96 \\
\bottomrule
\end{tabular}
\caption{\textit{PFS with different selection methods (linear, sparse, $n = 100$)}. True positive rate (TPR) and false discovery rate (FDR) for local graph recovery across methods (rows) and neighborhood radii (columns), averaged over 100 trials. PFS with integrated path stability selection~\cite{ipss,ipss_nonparametric} (IPSS) is the method used throughout this work and the default option in the \texttt{localgraph} software package. Other methods are: model-X knockoffs~\cite{modelX} with importance statistics based on (i) generalized linear models (KOGLM), (ii) the lasso (KOL1), and (iii) random forests (KORF); the Benjamini--Hochberg method~\cite{bh} (BH); the Benjamini--Yekutieli method~\cite{by} (BY).}

\label{tab:qval_linear_sparse_n100}
\end{table}

\begin{table}[ht]
\centering
\begin{tabular}{lccccccccc}
\toprule
 & \multicolumn{4}{c}{TPR} & & \multicolumn{4}{c}{FDR} \\
\cmidrule(lr){2-5}\cmidrule(lr){7-10}
Method & $r=1$ & $r=2$ & $r=3$ & $r=4$ &  & $r=1$ & $r=2$ & $r=3$ & $r=4$ \\
\midrule
\textbf{IPSS} & 0.83 & 0.47 & 0.30 & 0.20 &  & 0.19 & 0.22 & 0.23 & 0.22 \\
KOGLM & 0.00 & 0.00 & 0.00 & 0.00 &  & 0.01 & 0.01 & 0.01 & 0.01 \\
KOL1 & 0.00 & 0.00 & 0.00 & 0.00 &  & 0.00 & 0.00 & 0.00 & 0.00 \\
KORF & 0.13 & 0.03 & 0.02 & 0.01 &  & 0.11 & 0.11 & 0.11 & 0.11 \\
BH & 0.01 & 0.01 & 0.03 & 0.06 &  & 0.13 & 0.12 & 0.13 & 0.13 \\
BY & 0.00 & 0.00 & 0.01 & 0.01 &  & 0.03 & 0.03 & 0.03 & 0.03 \\
\bottomrule
\end{tabular}
\caption{\textit{PFS with different selection methods (nonlinear, sparse, $n = 100$)}. True positive rate (TPR) and false discovery rate (FDR) for local graph recovery across methods (rows) and neighborhood radii (columns), averaged over 100 trials. PFS with integrated path stability selection~\cite{ipss,ipss_nonparametric} (IPSS) is the method used throughout this work and the default option in the \texttt{localgraph} software package. Other methods are: model-X knockoffs~\cite{modelX} with importance statistics based on (i) generalized linear models (KOGLM), (ii) the lasso (KOL1), and (iii) random forests (KORF); the Benjamini--Hochberg method~\cite{bh} (BH); the Benjamini--Yekutieli method~\cite{by} (BY).}

\label{tab:qval_nonlinear_sparse_n100}
\end{table}

\begin{table}[ht]
\centering
\begin{tabular}{lccccccc}
\toprule
 & \multicolumn{3}{c}{TPR} & & \multicolumn{3}{c}{FDR} \\
\cmidrule(lr){2-4}\cmidrule(lr){6-8}
Method & $r=1$ & $r=2$ & $r=3$ &  & $r=1$ & $r=2$ & $r=3$ \\
\midrule
\textbf{IPSS} & 0.43 & 0.13 & 0.05 &  & 0.19 & 0.21 & 0.18 \\
KOGLM & 0.40 & 0.07 & 0.02 &  & 0.34 & 0.37 & 0.36 \\
KOL1 & 0.41 & 0.07 & 0.02 &  & 0.37 & 0.41 & 0.38 \\
KORF & 0.15 & 0.02 & 0.00 &  & 0.34 & 0.36 & 0.36 \\
BH & 0.93 & 0.81 & 0.79 &  & 0.88 & 0.99 & 0.95 \\
BY & 0.82 & 0.68 & 0.67 &  & 0.85 & 0.98 & 0.95 \\
\bottomrule
\end{tabular}
\caption{\textit{PFS with different selection methods (linear, dense, $n = 100$)}. True positive rate (TPR) and false discovery rate (FDR) for local graph recovery across methods (rows) and neighborhood radii (columns), averaged over 100 trials. PFS with integrated path stability selection~\cite{ipss,ipss_nonparametric} (IPSS) is the method used throughout this work and the default option in the \texttt{localgraph} software package. Other methods are: model-X knockoffs~\cite{modelX} with importance statistics based on (i) generalized linear models (KOGLM), (ii) the lasso (KOL1), and (iii) random forests (KORF); the Benjamini--Hochberg method~\cite{bh} (BH); the Benjamini--Yekutieli method~\cite{by} (BY).}

\label{tab:qval_linear_dense_n100}
\end{table}

\begin{table}[ht]
\centering
\begin{tabular}{lccccccc}
\toprule
 & \multicolumn{3}{c}{TPR} & & \multicolumn{3}{c}{FDR} \\
\cmidrule(lr){2-4}\cmidrule(lr){6-8}
Method & $r=1$ & $r=2$ & $r=3$ &  & $r=1$ & $r=2$ & $r=3$ \\
\midrule
\textbf{IPSS} & 0.87 & 0.59 & 0.51 &  & 0.12 & 0.16 & 0.09 \\
KOGLM & 0.35 & 0.06 & 0.01 &  & 0.06 & 0.07 & 0.06 \\
KOL1 & 0.31 & 0.05 & 0.01 &  & 0.06 & 0.06 & 0.06 \\
KORF & 0.02 & 0.00 & 0.00 &  & 0.02 & 0.02 & 0.02 \\
BH & 0.57 & 0.28 & 0.20 &  & 0.19 & 0.26 & 0.16 \\
BY & 0.26 & 0.07 & 0.03 &  & 0.03 & 0.04 & 0.03 \\
\bottomrule
\end{tabular}
\caption{\textit{PFS with different selection methods (linear, dense, $n = 500$)}. True positive rate (TPR) and false discovery rate (FDR) for local graph recovery across methods (rows) and neighborhood radii (columns), averaged over 100 trials. PFS with integrated path stability selection~\cite{ipss,ipss_nonparametric} (IPSS) is the method used throughout this work and the default option in the \texttt{localgraph} software package. Other methods are: model-X knockoffs~\cite{modelX} with importance statistics based on (i) generalized linear models (KOGLM), (ii) the lasso (KOL1), and (iii) random forests (KORF); the Benjamini--Hochberg method~\cite{bh} (BH); the Benjamini--Yekutieli method~\cite{by} (BY).}

\label{tab:qval_linear_dense_n500}
\end{table}

\begin{table}[ht]
\centering
\begin{tabular}{lccccccc}
\toprule
 & \multicolumn{3}{c}{TPR} & & \multicolumn{3}{c}{FDR} \\
\cmidrule(lr){2-4}\cmidrule(lr){6-8}
Method & $r=1$ & $r=2$ & $r=3$ &  & $r=1$ & $r=2$ & $r=3$ \\
\midrule
\textbf{IPSS} & 0.76 & 0.48 & 0.39 &  & 0.11 & 0.24 & 0.14 \\
KOGLM & 0.00 & 0.00 & 0.00 &  & 0.00 & 0.00 & 0.00 \\
KOL1 & 0.00 & 0.00 & 0.00 &  & 0.00 & 0.00 & 0.00 \\
KORF & 0.03 & 0.00 & 0.00 &  & 0.02 & 0.02 & 0.02 \\
BH & 0.00 & 0.00 & 0.00 &  & 0.17 & 0.14 & 0.04 \\
BY & 0.00 & 0.00 & 0.00 &  & 0.03 & 0.03 & 0.03 \\
\bottomrule
\end{tabular}
\caption{\textit{PFS with different selection methods (nonlinear, dense, $n = 500$)}. True positive rate (TPR) and false discovery rate (FDR) for local graph recovery across methods (rows) and neighborhood radii (columns), averaged over 100 trials. PFS with integrated path stability selection~\cite{ipss,ipss_nonparametric} (IPSS) is the method used throughout this work and the default option in the \texttt{localgraph} software package. Other methods are: model-X knockoffs~\cite{modelX} with importance statistics based on (i) generalized linear models (KOGLM), (ii) the lasso (KOL1), and (iii) random forests (KORF); the Benjamini--Hochberg method~\cite{bh} (BH); the Benjamini--Yekutieli method~\cite{by} (BY).}

\label{tab:qval_nonlinear_dense_n500}
\end{table}

\clearpage


\section{Additional details: Environmental and social drivers of cancer}\label{sup_sec:eqi}



\subsection{Data sources}\label{sup_sec:eqi_data}


Cancer incidence and mortality rates, along with smoking prevalence and cancer screening data, were obtained from the State Cancer Profiles platform, a joint initiative of the National Cancer Institute (NCI) and Centers for Disease Control and Prevention (CDC). This platform aggregates geographically resolved cancer statistics from national systems for public health planning. Cancer incidence and mortality were aggregated over 2017--2022, while screening and smoking data were drawn from 2017--2019. County-level demographic data---including race/ethnicity, sex, and age---were sourced from the U.S. Census Bureau using the 2000 decennial census.

The majority of covariates were drawn from the Environmental Protection Agency's (EPA) Environmental Quality Index (EQI) dataset, specifically the 2000--2005 release\cite{exposome_eqi}. The EQI compiles county-level indicators across five domains---air, water, land, built environment, and sociodemographic factors---from diverse federal and commercial sources, including the EPA, U.S. Geological Survey, Department of Agriculture, and Federal Bureau of Investigation (FBI). To improve comparability across counties, variables were often log-transformed to reduce skew, kriged to fill spatial gaps, or rescaled for consistency in valence. Additional details about the EQI dataset are provided in ``\textit{Construction of an environmental quality index for public health research}," by Messer et al.\cite{exposome_eqi}. The period 2000--2005 was chosen for most covariates to reflect the latency of cancer development~\cite{exposome_eqi,latency_breast,latency_lung1,latency_ovarian}. 


\subsection{Data cleaning}\label{sup_sec:eqi_cleaning}


The raw dataset consisted of $n = 3141$ counties and $p = 240$ variables aggregated from the sources described above. Indiana and Kansas were excluded because they did not report incidence data in the State Cancer Profiles dataset. We also removed Union County, Florida due to extreme outlier values in both cancer incidence and mortality. Some covariates were sparsely reported across large regions: variables missing from five or more states, as well as variables missing in at least $1500$ counties were removed. This resulted in our final dataset consisting of $n = 2857$ counties and $p = 165$ variables.


\subsection{PFS implementation}\label{sup_sec:eqi_pfs}


We applied PFS using random forest variable importance scores, specifically mean decrease in impurity (MDI), as implemented in the \texttt{ipss} Python package with its default settings\cite{ipss,ipss_nonparametric}. Connections between environmental exposures tended to be slightly stronger than those between socioeconomic variables. Thus, to improve interpretability in specific neighborhoods and in the estimated local graph as a whole, we customized $q$-value thresholds for a few variables based on prior relevance or notable patterns observed during analysis. The neighborhood $q$-value threshold for incidence was set to $0.04$. We also increased thresholds for Hispanic ($0.01$), Poverty ($0.011$), and Smoking ($0.01$). In the case of Education, where multiple variables had the same lowest $q$-value, we broke ties using expected false positive (efp) scores\cite{ipss_nonparametric}. For all remaining variables, the default neighborhood $q$-value threshold was $0.008$. In the case of graphical lasso, all variables were standardized to have mean $0$ and variance $1$. Both the PFS and graphical lasso graphs were estimated with a fixed random seed to ensure reproducibility.


\subsection{Supporting analyses}\label{sup_sec:eqi_additional}


Below we present additional tables and figures that complement the main analysis in \cref{sec:eqi}. These include carcinogenicity classifications of key environmental exposures identified in the local graph around cancer incidence (\cref{tab:exposures}), partial correlation results supporting the potential mediating role of poverty in racial disparities in mortality (\cref{tab:correlations}), and a nonlinear visualization of associations among environmental factors (\cref{fig:eqi_nonlinear}).

\begin{table*}[ht]
\centering
\small
\begin{tabular}{lllll}
\toprule
Exposure & Formula & IARC & IRIS & Common sources \\
\midrule
PM$_{2.5}$ & -- & 1 & -- & Combustion (gasoline, diesel, oil, wood) \\
Acrylonitrile & C$_3$H$_3$N & 1 & 2 & Plastic and rubber manufacturing \\
TCE & C$_2$HCl$_3$ & 1 & 1 & Industrial degreasing \\
Diesel & -- & 1 & 2 & Diesel engine exhaust \\
EtO & C$_2$H$_4$O & 1 & 1 & Chemical manufacturing \\
Hydrazine & N$_2$H$_4$ & 2A & 2 & Rocket fuels, polymer foam production \\
Quinoline & C$_9$H$_7$N & 2B & 2 & Chemical manufacturing (dyes, pesticides, solvents) \\
DEHP & C$_24$H$_38$O$_4$ & 2B & 2 & Plastic manufacturing \\
DBCP & C$_3$H$_5$Br$_2$Cl & 2B & -- & Pesticides (banned in contiguous U.S. in 1979) \\
Hg & Hg & 3 & -- & Industrial manufacturing, pesticides \\
SO$_2$ & SO$_2$ & 3 & -- & Combustion (coal and oil) \\
\bottomrule
\end{tabular}
\caption{\textit{Carcinogenicity classifications of environmental exposures identified by PFS in the local graph around cancer incidence (\cref{fig:exposome_pfs})}. Many of the exposures connected to cancer incidence are classified as known or likely human carcinogens. IARC categories: Group 1 = carcinogenic to humans; 2A = probably carcinogenic; 2B = possibly carcinogenic; 3 = unclassifiable. IRIS categories: 1 = carcinogenic to humans; 2 = likely to be carcinogenic. Dashes indicate no classification is available.}
\label{tab:exposures}
\end{table*}

\begin{table*}[ht]
\centering
\small
\begin{tabular}{lllcc}
\toprule
Target & Covariate & Adjusted for & $\rho$ ($p$-value) & $\rho_{\text{adj}}$ ($p$-value) \\
\midrule
Incidence & Hispanic & Screening & $-0.340$ ($<10^{-10}$) & $-0.324$ ($<10^{-10}$) \\
Incidence & Hispanic & Mercury & $-0.340$ ($<10^{-10}$) & $-0.174$ ($<10^{-10}$) \\
Mortality & Black & Poverty & $0.228$ ($<10^{-10}$) & $0.053$ ($0.005$) \\
Mortality & Poverty & Black & $0.432$ ($<10^{-10}$) & $0.380$ ($<10^{-10}$) \\
Mortality & White & Poverty & $-0.059$ ($0.002$) & $0.249$ ($<10^{-10}$) \\
Mortality & Poverty & White & $0.432$ ($<10^{-10}$) & $0.484$ ($<10^{-10}$) \\
\bottomrule
\end{tabular}
\caption{\textit{Partial correlations between variables in the county-level cancer study}. For each target-covariate pair, we report the Pearson correlation $\rho$ and the partial correlation $\rho_{\text{adj}}$ after adjusting for a third variable. These results support the potential mediating role of poverty in racial disparities in mortality, and suggest that reduced mercury exposure may partially explain lower cancer incidence in counties with large Hispanic populations, whereas screening differences do not.}
\label{tab:correlations}
\end{table*}

\ifthenelse{\boolean{showfigures}}{
\begin{figure*}[ht!]
\begin{adjustbox}{center}
\includegraphics[
	height=0.35\textheight, 
	width=\textwidth,
	trim={0pt 0pt 0pt 0pt}, clip]
	{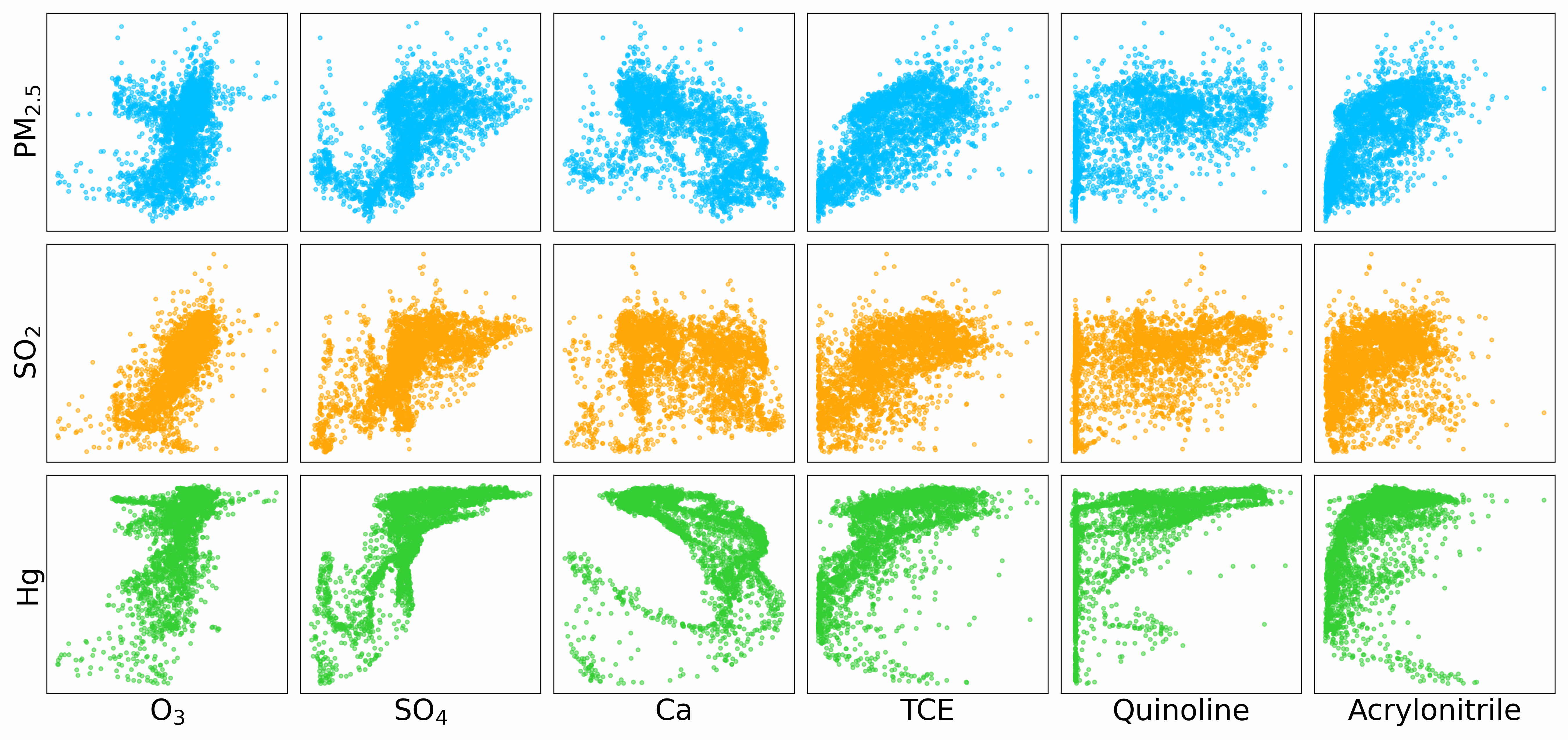}
\end{adjustbox}
\caption{\textit{Nonlinear associations among environmental exposures}. Scatter plots show marginal pairwise relationships between PM$_{2.5}$ (row 1), SO$_2$ (row 2), and Hg (row 3), and six other exposures---O$_3$, SO$_4$, Ca, TCE, Quinoline, and Acrylonitrile---that are present in the local graph in \cref{fig:exposome_pfs}. Several of the relationships exhibit nonlinear patterns, including strong curvature. These findings highlight the limitations of linear modeling assumptions and support the use of nonparametric methods for analyzing high-dimensional environmental data.}
\label{fig:eqi_nonlinear}
\end{figure*}
}{}

\clearpage


\section{Additional details: Cross-modal pathways in breast cancer}\label{sup_sec:bc}



\subsection{Data sources}\label{sup_sec:bc_data}


All data were obtained from The Cancer Genome Atlas breast cancer cohort (TCGA-BRCA) and downloaded from the publicly available LinkedOmics platform~\cite{linkedomics}. The dataset includes gene expression measured by bulk RNA sequencing (RNAseq), microRNA (miRNA) expression, protein abundance profiled using reverse phase protein arrays (RPPA), and clinical information. Bulk RNAseq data are provided at the gene level, normalized as log$_2$(RSEM $+ 1$) expression values. The raw RNAseq matrix contains $1093$ tumor samples across $20{,}155$ genes. miRNA data, derived from the Illumina HiSeq platform and provided as log$_2$(RPM $+ 1$) normalized values, include $755$ samples across $823$ miRNAs. RPPA data contain log-normalized protein measurements across $887$ samples and $175$ proteins. The three clinical target variables are (i) pathologic stage (stage I, II, III, or IV), (ii) histological type, filtered to the two most common categories, namely invasive ductal carcinoma (IDC) and invasive lobular carcinoma (ILC), and (iii) survival status at last follow-up (in the cleaned dataset, $76$ patients were deceased and $471$ were alive).


\subsection{Data cleaning}\label{sup_sec:bc_cleaning}


We retained genes whose mean expression or variance exceeded the $75$th percentile of the gene-wise distributions, thereby removing genes with relatively low expression and variability. To integrate the modalities, we first identified samples with complete measurements across the RNAseq, miRNA, and RPPA data, as well as the three clinical variables. We then removed features with absolute pairwise correlation exceeding $0.999$, retaining only one representative column from each highly correlated pair. Columns containing any missing values were removed, and samples with missing values in the response variables were also excluded. After these steps, the final analysis set consists of $n=547$ tumor samples with $p=10{,}744$ features, spanning $9785$ genes, $819$ miRNAs, $137$ proteins, and the three clinical targets.


\subsection{PFS implementation}\label{sup_sec:bc_pfs}


PFS was implemented using IPSS with variable importance scores computed from gradient boosting using XGBoost~\cite{xgboost}, which is the default choice in the \texttt{localgraph} package. Due to the high dimensionality of the data, we increased the number of subsamples in the IPSS algorithm to $B = 200$ (default $B = 100$) to improve robustness. The IPSS preselection expansion factor was set to $2.5$ (default $1.5$) to retain a broader pool of candidate features before final selection. All other IPSS hyperparameters were left at their default settings.

The maximum radius in the PFS algorithm was set to $3$. The pathwise $q$-value threshold $q_{\text{path}}^{*}$ was set to $1$, effectively delegating all thresholding to the neighborhood level. For the three clinical target variables, neighborhood $q$-value thresholds were set to $0.25$. For non-target features, an intermodal $q$-value threshold of $0.035$ was applied; that is, an edge was included between nodes from different modalities if its $q$-value was below this cutoff. For intramodal edges---that is, for edges between nodes of the same type---all nodes in $S_{1}(V_{0})$ had an intramodal neighborhood $q$-value threshold of $0.025$, while all nodes in $S_{2}(V_{0})$ had an intramodal neighborhood $q$-value threshold of $0.0175$.


\subsection{Supporting analyses}\label{sup_sec:bc_additional}


In addition to reviewing the literature, we performed two post hoc validation analyses to assess the biological coherence of structures in the local graph estimated by PFS (\cref{fig:breast_cancer}). These analyses were not used to construct the graph, tune parameters, or select features, but rather to provide external statistical validation of the estimated graph.

\textit{Gene set ORA}. For selected gene clusters extracted from the estimated local graph, we conducted over-representation analysis (ORA) using curated gene sets from the Molecular Signatures Database (MSigDB) C2 collection (v2025.1), which includes gene sets derived from pathway databases, genetic perturbation experiments, and the biomedical literature~\cite{msigdb}. ORA was performed using the \texttt{gseapy}~\cite{gseapy} implementation of Fisher's exact test, with the background defined as all $9785$ genes in the dataset to which PFS was applied. For each cluster, gene sets were ranked by FDR-adjusted $p$-value. Additional details and results, including lists of genes in each cluster, are in \cref{sup_sec:bc_additional}.

\textit{Protein-protein interaction support}. To assess whether protein-protein edges inferred by PFS are supported by existing biological knowledge, we compared inferred protein-protein edges against the STRING v12.0 protein interaction database~\cite{string}. We extracted all protein-protein edges from the estimated local graph and recorded the number of edges with a STRING interaction score of at least 400, which is the database's default threshold for moderate to high-confidence interactions. To evaluate enrichment relative to chance, we constructed a null distribution by repeatedly sampling the same number of protein pairs uniformly at random from the set of all measured proteins and computing the number of sampled pairs exceeding the STRING score threshold. This procedure was repeated $100{,}000$ times to obtain an empirical null distribution. Enrichment was quantified by the fold increase in supported edges relative to the null expectation, and significance was assessed using an empirical $p$-value. The same procedure was applied to graphs inferred by alternative methods for comparison. Detailed results are reported in \cref{sup_sec:bc_additional}.

We report additional results from over-representation analyses (ORAs) conducted on gene modules extracted from the estimated local graph in the TCGA breast cancer study (\cref{sec:breast_cancer}). These analyses provide set-level statistical validation of the molecular neighborhoods surrounding the clinical target variables, complementing the literature-based interpretation of individual genes and edges presented in the main text.

For each module, ORA was performed using curated gene sets from MSigDB (C2 collection). Reported results correspond to the top-ranked enriched gene sets for each module, ordered by FDR-adjusted $p$-value. Only gene sets with overlap at least three were considered. Here, \emph{overlap} denotes the number of genes shared between a given module and a gene set, relative to the total size of that gene set. For modules exhibiting multiple related enrichments (e.g., HOOK3), we report representative results capturing distinct biological themes.

We report results for the following clusters of genes in the local graph estimated by PFS (\cref{fig:breast_cancer}); each cluster is named by a variable in the cluster that links directly to a target in the estimated graph: ANKLE2 cluster (ANKLE2, GOLGA3, EP400, POLE, SETD1B, MLL2, C12orf51, ULK1, SFRS8); HOOK3 cluster (HOOK3, QKI, SGK196, MYST3, MAP4K5, FNTA, HGSNAT, AGPAT6, GOLGA7, IKBKB, VDAC3, WHSC1L1, SLC20A2); miR-210 cluster (P4HA1, ISCU, RIC8A, CA9, EPS8L2, BTNL9, NDRG1, CD151, PDDC1, PKP3, PHRF1, BET1L, NAP1L4); CDH1 cluster (CDH1, CES2, ABHD1, CES8, KIAA0174, AP1G1, USP10, DYNC1LI2, CTCF, ATXN1L, ATMIN, AIP, TERF2IP, GABARAPL2); and miR-133a-1 cluster (FBXO2, ACTC1, DES, MYH11, LOC728264).

\begin{table*}[ht]
\centering
\small
\begin{tabular}{lll}
\toprule
Cluster & Gene set & Adjusted $p$-value \\
\midrule
ANKLE2 & NIKOLSKY\_BREAST\_CANCER\_12Q24\_AMPLICON & $8.6\times10^{-13}$ \\

HOOK3 & NIKOLSKY\_BREAST\_CANCER\_8P12\_P11\_AMPLICON & $3.3\times10^{-13}$ \\
HOOK3 & NIKOLSKY\_MUTATED\_AND\_AMPLIFIED\_IN\_BREAST\_CANCER & $9.8\times10^{-3}$ \\
HOOK3 & CREIGHTON\_ENDOCRINE\_THERAPY\_RESISTANCE\_5 & $4.8\times10^{-2}$ \\

CDH1 & PROVENZANI\_METASTASIS\_UP & $1.5\times10^{-4}$ \\

miR-210 & BUFFA\_HYPOXIA\_METAGENE & $5.4\times10^{-3}$ \\

miR-133a-1 & REACTOME\_MUSCLE\_CONTRACTION & $3.5\times10^{-3}$ \\
\bottomrule
\end{tabular}
\caption{\textit{Over-representation analysis (ORA) results for gene modules extracted from the estimated local graph in TCGA breast cancer data}. For each module, we report the top enriched gene sets ranked by FDR-adjusted $p$-value. Only gene sets with minimum overlap 3 were considered, where \textit{overlap} is defined as the number of genes shared between the module and the corresponding gene set.}
\label{tab:bc_ora}
\end{table*}

\begin{table*}[ht]
\centering
\begin{tabular}{lccccc}
\toprule
 & Total & Supported & Expected & Fold enrichment & $p$-value \\
\midrule
\textbf{PFS} & 34 & 22 & 6.82 & 3.228 & $10^{-5}$ \\
ARACNE & 25 & 5 & 5.01 & 0.998 & 0.580 \\
GFCSL & 112 & 15 & 22.43 & 0.669 & 0.974 \\
GFCSL(NPN) & 388 & 66 & 77.76 & 0.849 & 0.942 \\
DSNWSL & 0 & 0 & 0 & 0 & -- \\
DSNWSL(NPN) & 0 & 0 & 0 & 0 & -- \\
NLasso & 0 & 0 & 0 & 0 & -- \\
NLasso(NPN) & 0 & 0 & 0 & 0 & -- \\
\bottomrule
\end{tabular}
\caption{\textit{Support of inferred protein-protein interactions}. For each method, we report the total number of inferred protein-protein edges, the number of inferred edges supported by a STRING~\cite{string} score of at least 400 (the default STRING threshold for moderate-to-high confidence), the expected number of supported edges under a null model based on randomly sampled protein pairs, the corresponding fold enrichment, and an empirical $p$-value computed from $100{,}000$ random draws. Methods with zero inferred protein-protein edges are shown for completeness.}
\label{tab:bc_string}
\end{table*}

\clearpage


\section{Additional details: Brain networks and cognition}\label{sup_sec:hcp}



\subsection{Data sources}\label{sup_sec:hcp_data}


Subject-level data together with the official Human Connectome Project~\cite{hcp} (HCP) data dictionary were obtained from the HCP Young Adult (S1200) release and downloaded from the public repository \url{https://balsa.wustl.edu}.


\subsection{Data cleaning}\label{sup_sec:hcp_cleaning}

We began with the subject-level HCP dataset and official data dictionary. All cognition variables were removed except the target (Fluid Cognition Composite). Bookkeeping and metadata variables (e.g., subject identifiers, acquisition details, scanner and session fields, completeness indicators) were excluded. To ensure consistency, only age-adjusted phenotype versions were retained. We restricted behavioral variables to summary-level measures, including working memory, emotion, language, motor, sensory, and personality summary scores. Structural variables were limited to FreeSurfer~\cite{hcp_freesurfer} morphometry features, identified by the prefix \texttt{FS\_}. Only numeric variables were retained. Subjects missing the target were removed. Features with zero variance were dropped. Remaining missing values were mean-imputed columnwise. After cleaning, the final dataset consists of $n = 1188$ individuals and $p = 213$ variables.


\subsection{PFS implementation}\label{sup_sec:hcp_pfs}


The maximum radius was set to 3 and local $q$-value thresholds at radii 1, 2, and 3 were set to $0.15$, $0.03$, and $0.02$, respectively. All other PFS parameters were set to their default values in the \texttt{localgraph} package.


\subsection{Supporting analyses}\label{sup_sec:hcp_additional}


To interpret cortical structure recovered in the local graphs estimated by PFS and other methods, we performed over-representation analysis (ORA) of Yeo-7 functional systems~\cite{hcp_yeo}. Results are reported in \cref{tab:hcp}. For each method we extracted the set of cortical Desikan-Killiany (DK) regions~\cite{hcp_dk} appearing in the local graph component anchored at the node directly connected to the target. Only cortical FreeSurfer morphometry features (thickness and surface area) were included~\cite{hcp_freesurfer}; subcortical and non-morphometric variables were excluded from enrichment testing.

Each DK region was mapped to a Yeo-7 network using a vertex-count lookup table assigning each parcel to its dominant functional system~\cite{hcp_alexander}. The background consists of all cortical DK regions present in the cleaned dataset. Enrichment $p$-values are computed using an upper-tail hypergeometric test. Thus, ORA tests whether a functional system is represented in the local graph more often than expected under random sampling of cortical regions, linking recovered structural organization to established large-scale functional systems implicated in cognition. Enrichment outputs and gene lists are provided with the accompanying code release.

\ifthenelse{\boolean{showfigures}}{
\begin{figure*}[ht!]
\begin{adjustbox}{center}
\includegraphics[
	height=0.65\textheight, 
	width=\textwidth,
	trim={0pt 0pt 0pt 0pt}, clip]
	{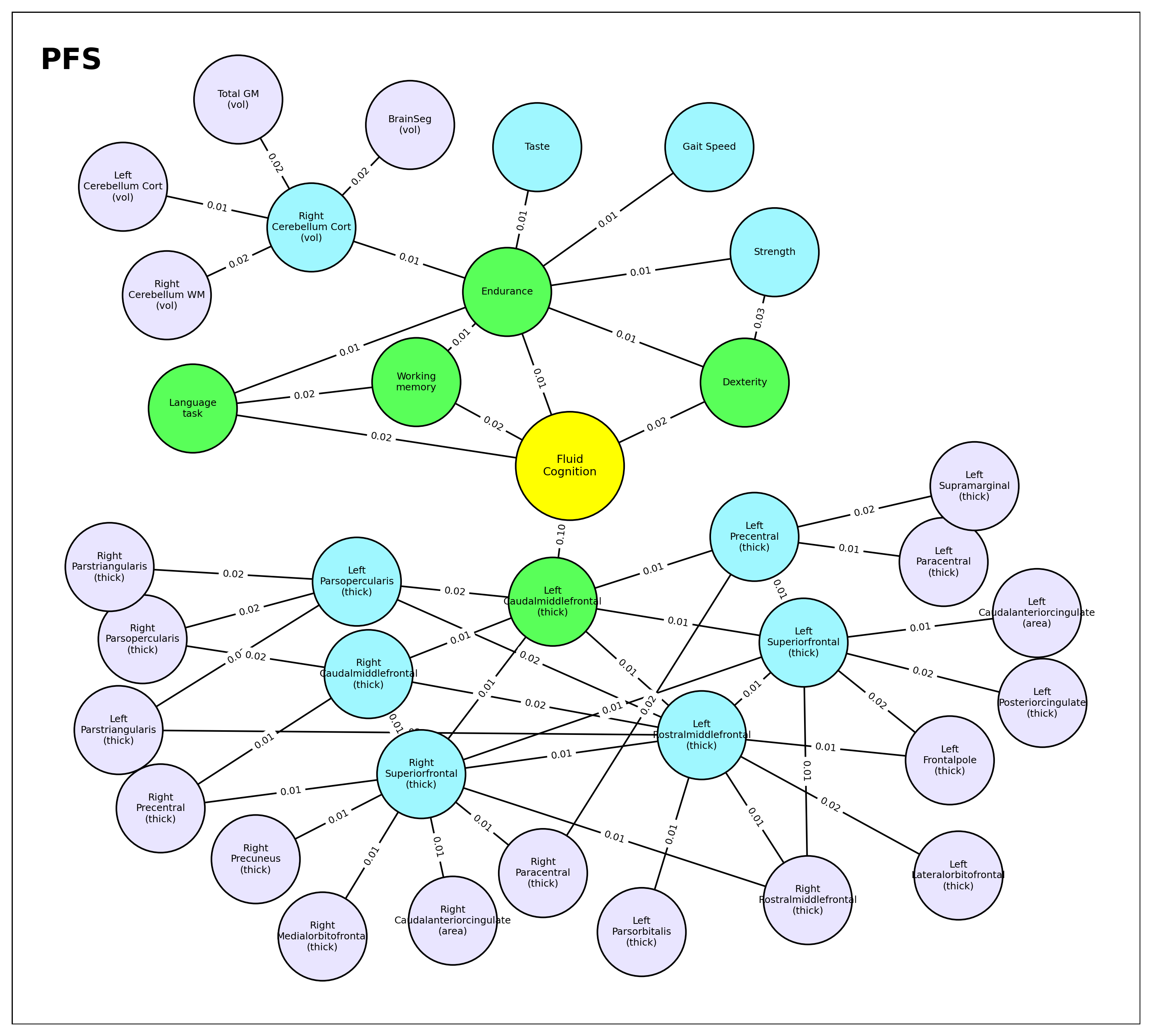}
\end{adjustbox}
\caption{\textit{Radius 3 graph around fluid cognition for the Human Connectome Project data, estimated by PFS}. The target variable, Fluid Cognition score, is shown in yellow. Nodes directly connected to the target are shown in green, nodes at distance two are shown in blue, and nodes at distance three are shown in purple. Edges are annotated with $q$-values, with smaller values indicating stronger evidence of conditional dependence. Thickness, area, and volume measurements of specific brain regions are labeled with``(thick)", ``(area)", and ``(volume)" respectively.}
\label{fig:hcp_pfs}
\end{figure*}
}{}

\clearpage


\section{Additional details: Cell-type-specific gene networks in Alzheimer's disease}\label{sup_sec:ad}



\subsection{Data sources}\label{sup_sec:ad_data}


Single-nucleus RNA-sequencing data from the study by Grubman et al.~\cite{ad_grubman} were downloaded from \url{https://adsn.ddnetbio.com}. The dataset includes gene expression profiles for nuclei from six Alzheimer's disease (AD) and six control individuals, together with cell-type annotations and metadata. Analyses were conducted separately within astrocytes ($n=2171$ cells), microglia ($n=449$), and oligodendrocyte progenitor cells (OPCs; $n=1078$). The target variable in each analysis was a binary indicator of AD versus control status.


\subsection{Data cleaning}\label{sup_sec:ad_cleaning}


Raw count matrices (genes $\times$ nuclei), gene identifiers, and cell barcodes were aligned with the provided metadata. Genes detected in fewer than $50$ nuclei were removed. Library-size normalization was performed by scaling each nucleus to $10{,}000$ total counts, followed by $\log(1+x)$ transformation. After filtering, the analysis matrices contained $p=10{,}788$ genes (astrocytes), $p=10{,}513$ genes (microglia), and $p=10{,}772$ genes (OPCs).


\subsection{PFS implementation}\label{sup_sec:ad_pfs}


PFS was applied separately within each cell type, with maximum radius 3 in each case. For astrocytes, the local $q$-value thresholds were set to $0.025$ at all radii. For microglia and OPCs, local $q$-value thresholds were set to $0.05$, $0.1$, and $0.1$ at radii $1$, $2$, and $3$, respectively. All other parameters were set to their defaults in the \texttt{localgraph} package.

\ifthenelse{\boolean{showfigures}}{
\begin{figure*}[ht!]
\begin{adjustbox}{center}
\includegraphics[
	height=0.4\textheight, 
	width=\textwidth,
	trim={0pt 0pt 0pt 0pt}, clip]
	{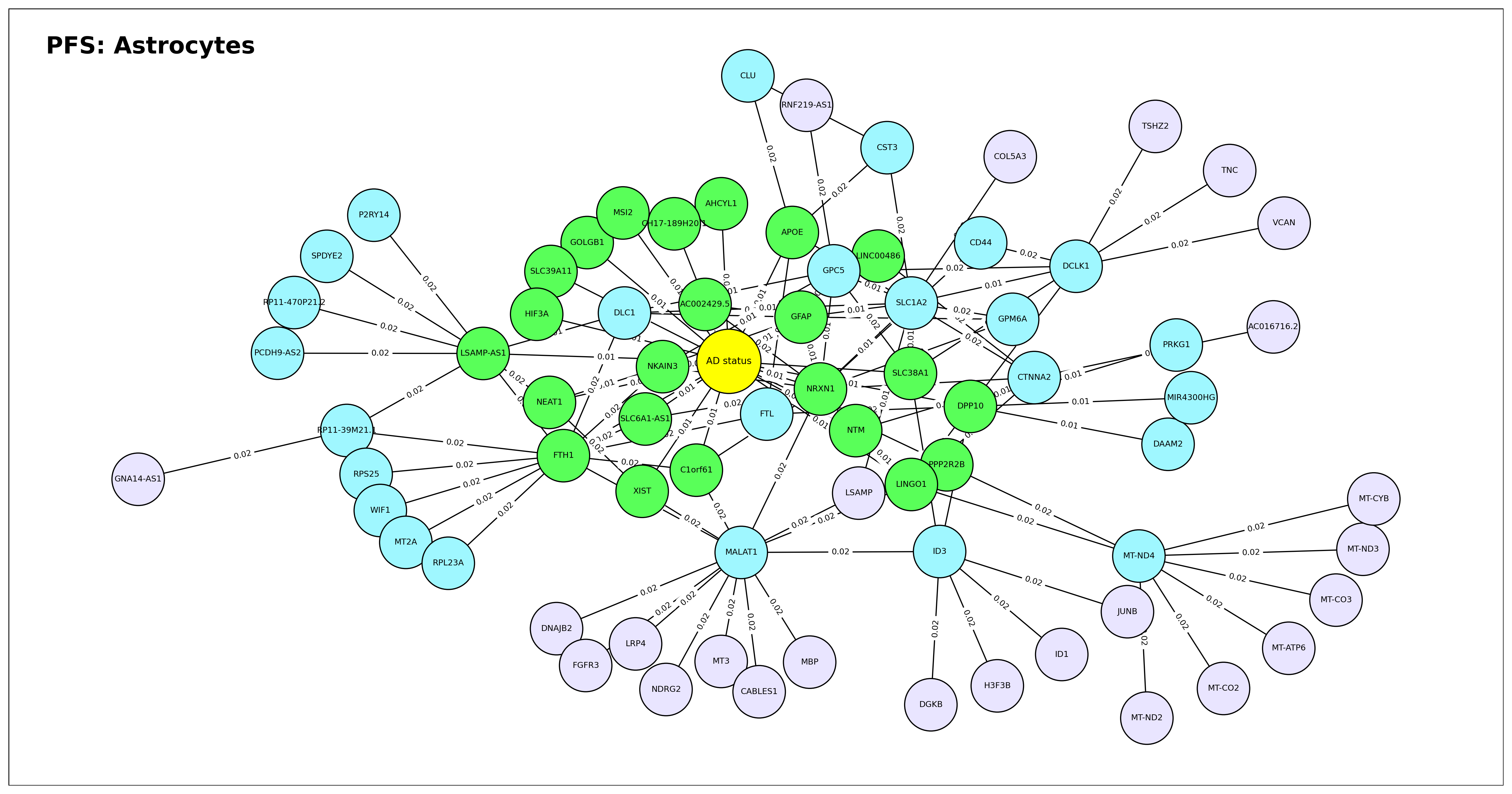}
\end{adjustbox}
\caption{\textit{PFS graph of radius 3 around AD status in astrocytes}. The target (AD status) is shown in yellow. Genes directly connected to the target are shown in green, distance-2 genes are blue, and distance-3 genes are purple. Edges are annotated with $q$-values, with smaller values indicating stronger evidence of conditional dependence.}
\label{fig:ad_pfs_astro}
\end{figure*}
}{}

\ifthenelse{\boolean{showfigures}}{
\begin{figure*}[ht!]
\begin{adjustbox}{center}
\includegraphics[
	height=0.4\textheight, 
	width=\textwidth,
	trim={0pt 0pt 0pt 0pt}, clip]
	{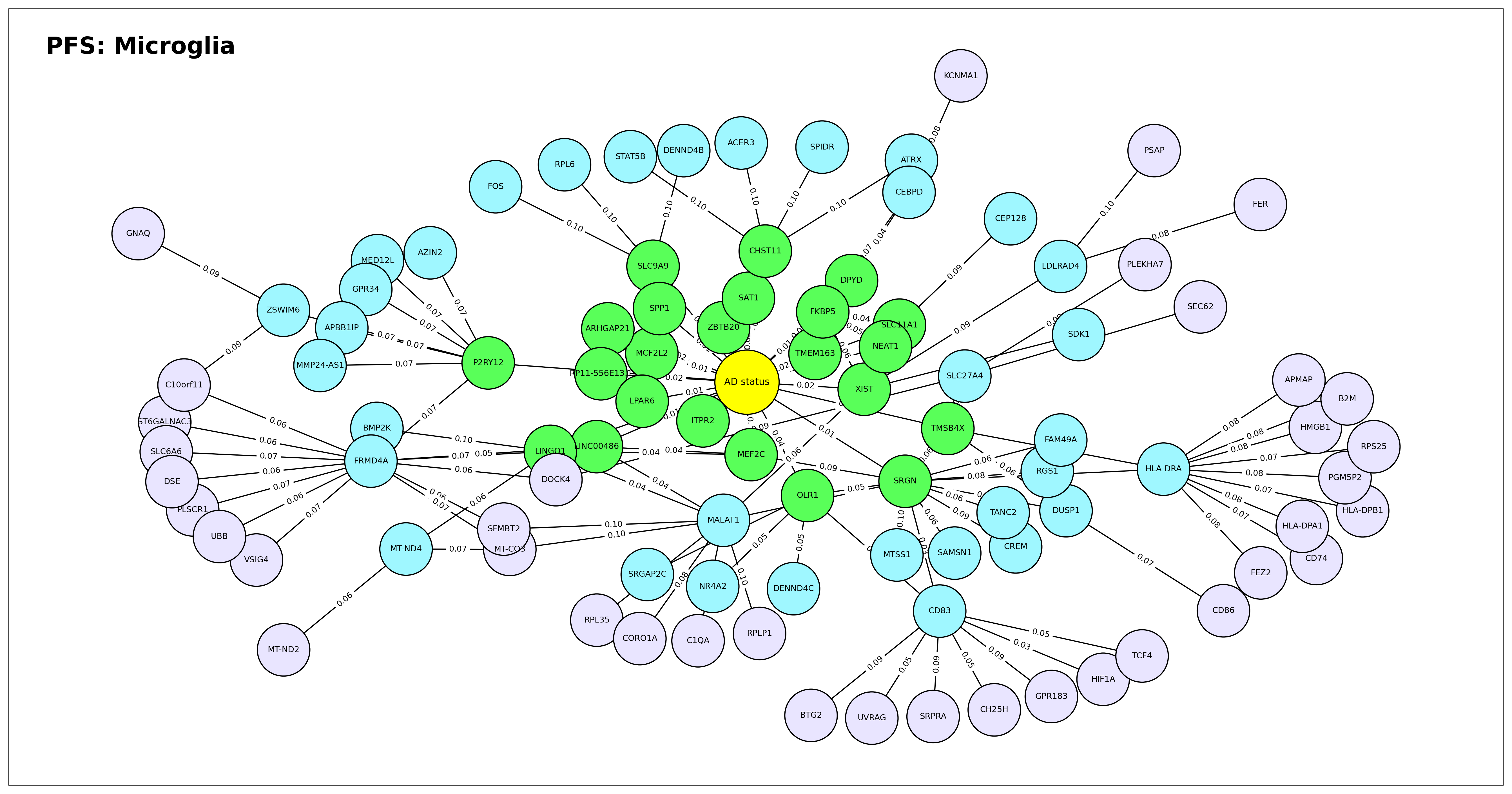}
\end{adjustbox}
\caption{\textit{PFS graph of radius 3 around AD status in microglia}. Visual encoding and statistical interpretation are as in \cref{fig:ad_pfs_astro}.}
\label{fig:ad_pfs_mg}
\end{figure*}
}{}

\ifthenelse{\boolean{showfigures}}{
\begin{figure*}[ht!]
\begin{adjustbox}{center}
\includegraphics[
	height=0.4\textheight, 
	width=\textwidth,
	trim={0pt 0pt 0pt 0pt}, clip]
	{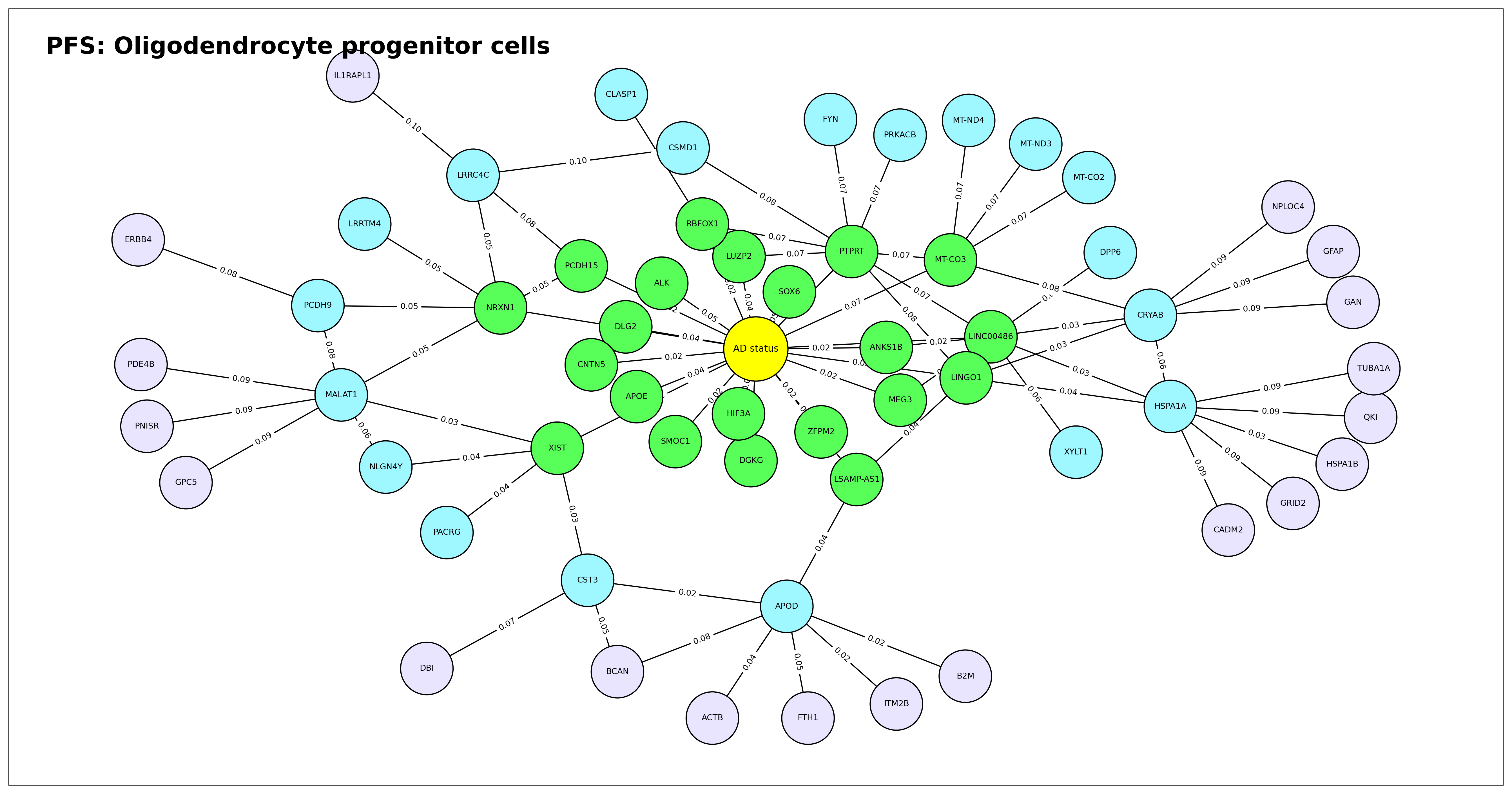}
\end{adjustbox}
\caption{\textit{PFS graph of radius 3 around AD status in oligodendrocyte progenitor cells (OPCs)}. Visual encoding and statistical interpretation are as in \cref{fig:ad_pfs_astro}.}
\label{fig:ad_pfs_opc}
\end{figure*}
}{}


\subsection{Supporting analyses}\label{sup_sec:ad_additional}


We performed gene set over-representation analyses (ORAs) to assess the biological coherence of the PFS local graphs around AD status for each cell type. Results are reported in \cref{tab:ad}. As with the other ORAs presented in this work, these analyses serve as external statistical validation of the estimated local structures.

Gene clusters consist of all genes in the local graph of radius 3 around AD status for each cell type. ORA was performed using curated gene sets from the MSigDB C2 Reactome and C5 Gene Ontology Biological Process collections (v2025.1) via the \texttt{gseapy} implementation of Fisher’s exact test~\cite{msigdb,gseapy}. The background set is comprised of all genes retained after cell-type-specific preprocessing (\cref{sup_sec:ad_cleaning}). Gene sets were filtered to require minimum overlap of 5 genes. FDR-adjusted $p$-values are reported. Enrichment outputs and gene lists are provided with the accompanying code release.

\clearpage


\section{Local graphs estimated by other methods}\label{sup_sec:other_methods}


In \cref{sup_sec:limitations} we provided theoretical explanations for why certain global graph estimation methods can struggle to recover local structure, and in \cref{sup_sec:simulations} we presented simulation results illustrating these limitations when the true graph is known. In this section, we examine how several global graph estimation methods perform in practice when applied to the same county-level cancer and breast cancer datasets that we analyzed with PFS in \cref{sec:eqi,sec:breast_cancer}. We apply methods that assume Gaussian data to both the original data and to a nonparanormal transformation of the data, which is designed to relax the Gaussian assumption~\cite{nonparanormal}. Nonparanormal transformations were implemented using the R package \texttt{huge}.


\subsection{County-level cancer study}\label{sup_sec:other_methods_eqi}


We report results from ten methods applied to the same county-level cancer dataset that PFS was applied to in \cref{sec:eqi}. The graphical lasso and nodewise lasso are applied with default settings using \texttt{huge}. The five methods with global FDR control---BNWSL, DSGL, DSNWSL, GFCL, and GFCSL---were applied using the R package \texttt{SILGGM}. For these methods, the target FDR was set to $0.001$, substantially more stringent than the \texttt{SILGGM} defaults of $0.05$ or $0.1$. ARACNE, MMPC, and the stable PC algorithm (PC) were applied using the R package \texttt{bnlearn}. For these three methods, we used mutual information rather than correlation to test for conditional independence due to observed nonlinearities in the data (\cref{fig:eqi_nonlinear}). We attempted to run several other methods from \texttt{bnlearn}, namely Fast IAMB, HPC, SIHPC, but these failed to return results within the allotted time limit of 24 hours. The following figures show estimated local graphs of radius $2$ around cancer incidence and mortality for each respective method.

\ifthenelse{\boolean{showfigures}}{
\begin{figure*}[ht!]
\begin{adjustbox}{center}
\includegraphics[
	height=0.225\textheight, 
	width=\textwidth,
	trim={0pt 0pt 0pt 0pt}, clip]
	{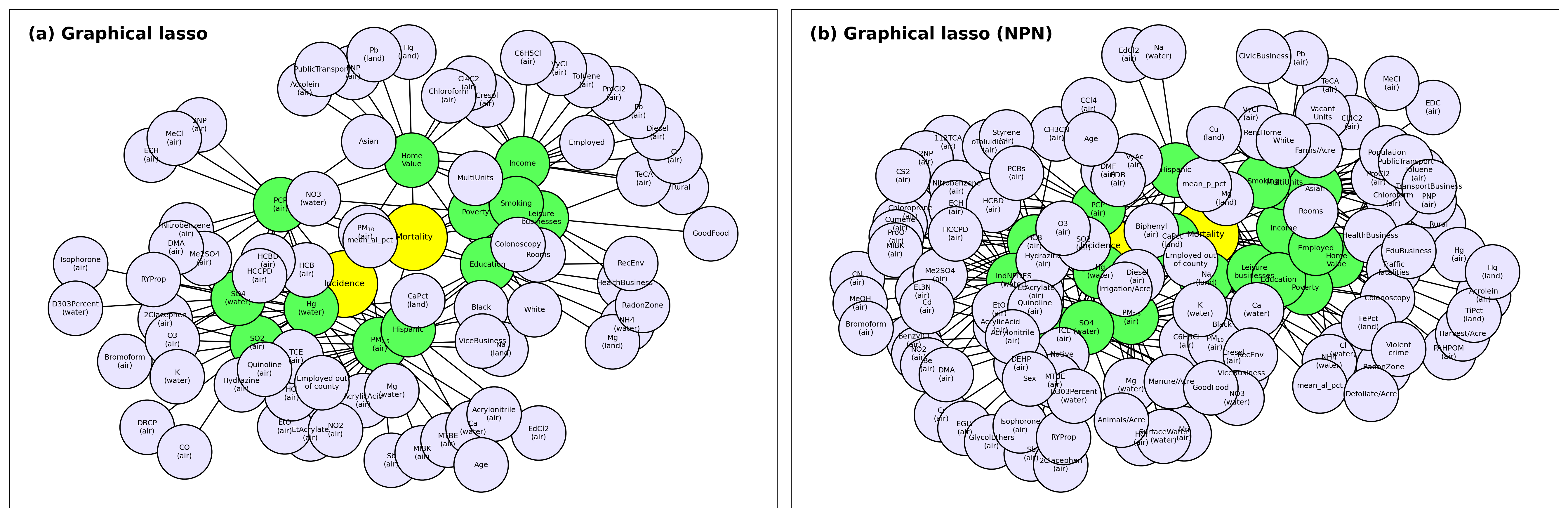}
\end{adjustbox}
\caption{\textit{County-level cancer results using graphical lasso}. Estimated local graph of radius 2 using \textbf{(a)} original data and \textbf{(b)} nonparanormal data.}
\label{fig:eqi_glasso}
\end{figure*}
}{}

\ifthenelse{\boolean{showfigures}}{
\begin{figure*}[ht!]
\begin{adjustbox}{center}
\includegraphics[
	height=0.225\textheight, 
	width=\textwidth,
	trim={0pt 0pt 0pt 0pt}, clip]
	{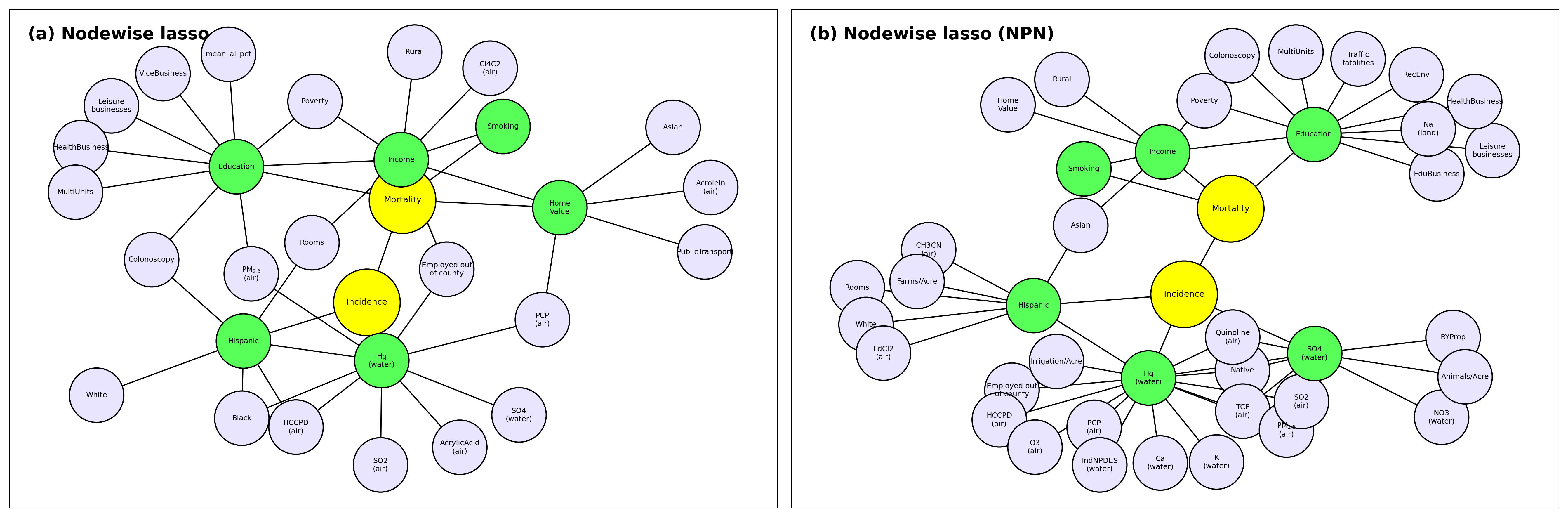}
\end{adjustbox}
\caption{\textit{County-level cancer results using nodewise lasso}. Estimated local graph of radius 2 using \textbf{(a)} original data and \textbf{(b)} nonparanormal data.}
\label{fig:eqi_mb}
\end{figure*}
}{}

\ifthenelse{\boolean{showfigures}}{
\begin{figure*}[ht!]
\begin{adjustbox}{center}
\includegraphics[
	height=0.225\textheight, 
	width=\textwidth,
	trim={0pt 0pt 0pt 0pt}, clip]
	{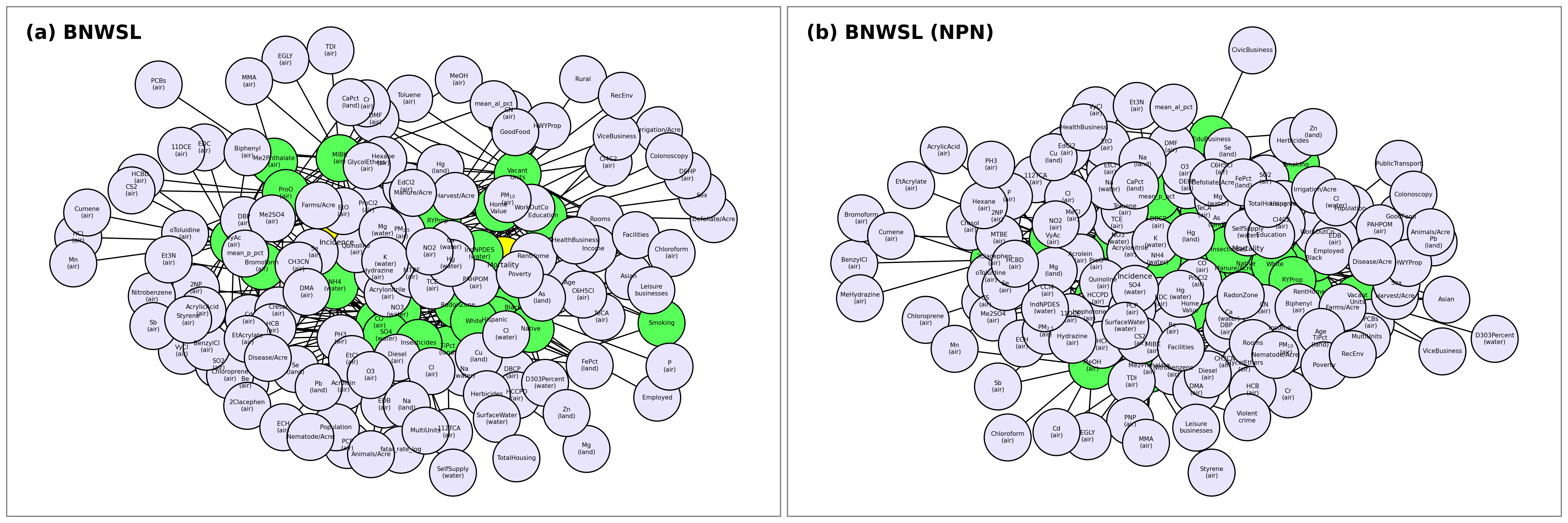}
\end{adjustbox}
\caption{\textit{County-level cancer results using BNWSL}. Estimated local graph of radius 2 using \textbf{(a)} original data and \textbf{(b)} nonparanormal data.}
\label{fig:eqi_bnwsl}
\end{figure*}
}{}

\ifthenelse{\boolean{showfigures}}{
\begin{figure*}[ht!]
\begin{adjustbox}{center}
\includegraphics[
	height=0.225\textheight, 
	width=\textwidth,
	trim={0pt 0pt 0pt 0pt}, clip]
	{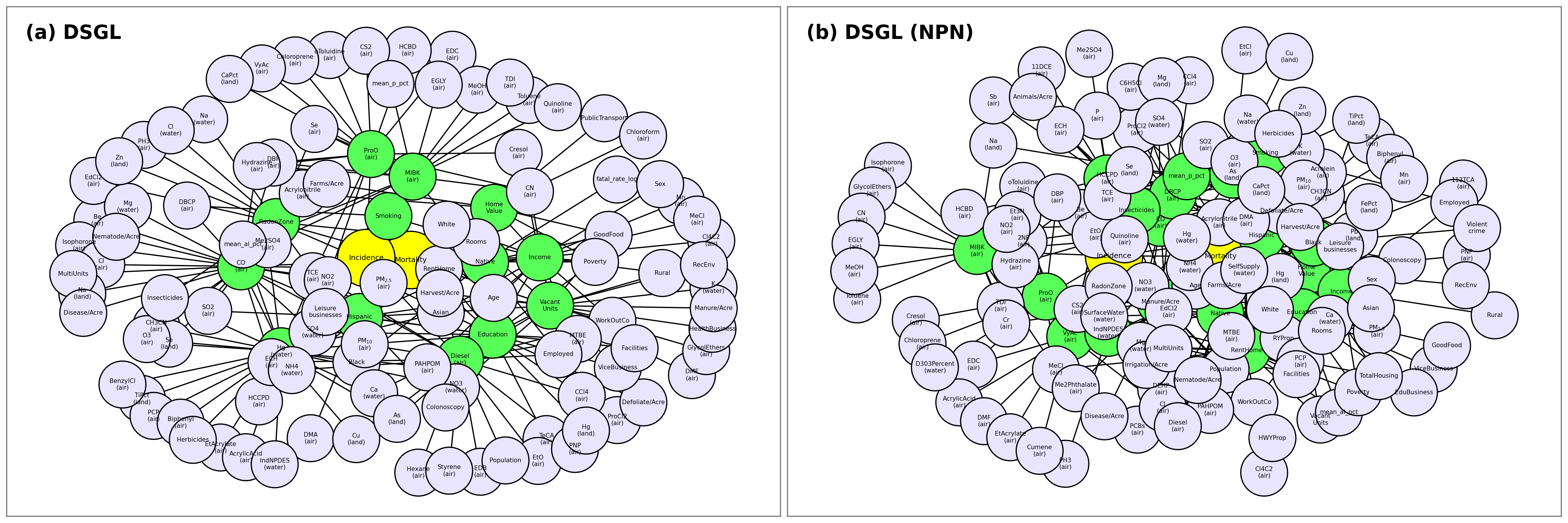}
\end{adjustbox}
\caption{\textit{County-level cancer results using DSGL}. Estimated local graph of radius 2 using \textbf{(a)} original data and \textbf{(b)} nonparanormal data.}
\label{fig:eqi_dsgl}
\end{figure*}
}{}

\ifthenelse{\boolean{showfigures}}{
\begin{figure*}[ht!]
\begin{adjustbox}{center}
\includegraphics[
	height=0.225\textheight, 
	width=\textwidth,
	trim={0pt 0pt 0pt 0pt}, clip]
	{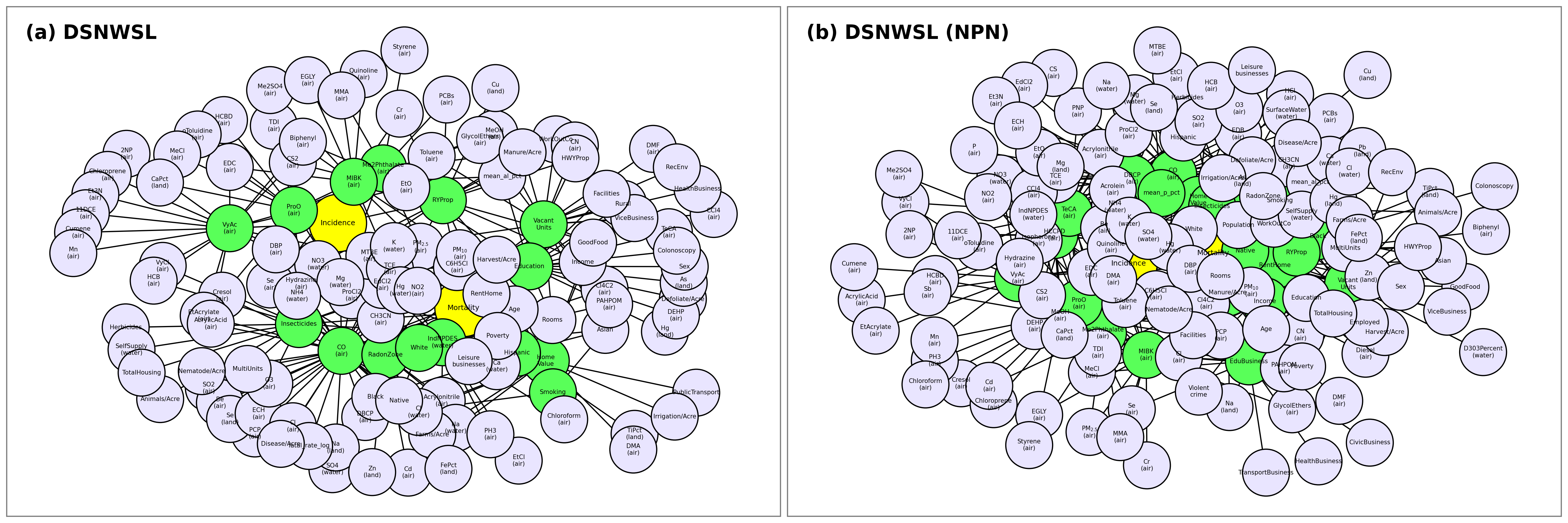}
\end{adjustbox}
\caption{\textit{County-level cancer results using DSNWSL}. Estimated local graph of radius 2 using \textbf{(a)} original data and \textbf{(b)} nonparanormal data.}
\label{fig:eqi_dsnwsl}
\end{figure*}
}{}

\ifthenelse{\boolean{showfigures}}{
\begin{figure*}[ht!]
\begin{adjustbox}{center}
\includegraphics[
	height=0.225\textheight, 
	width=\textwidth,
	trim={0pt 0pt 0pt 0pt}, clip]
	{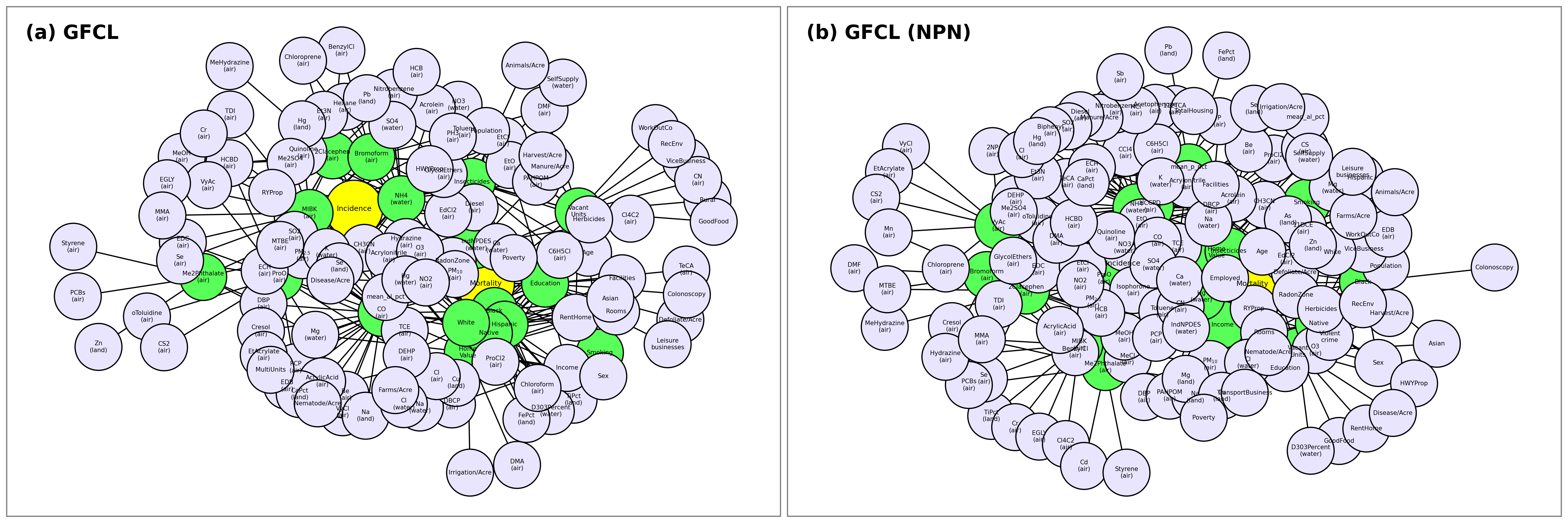}
\end{adjustbox}
\caption{\textit{County-level cancer results using GFCL}. Estimated local graph of radius 2 using \textbf{(a)} original data and \textbf{(b)} nonparanormal data.}
\label{fig:eqi_gfcl}
\end{figure*}
}{}

\ifthenelse{\boolean{showfigures}}{
\begin{figure*}[ht!]
\begin{adjustbox}{center}
\includegraphics[
	height=0.225\textheight, 
	width=\textwidth,
	trim={0pt 0pt 0pt 0pt}, clip]
	{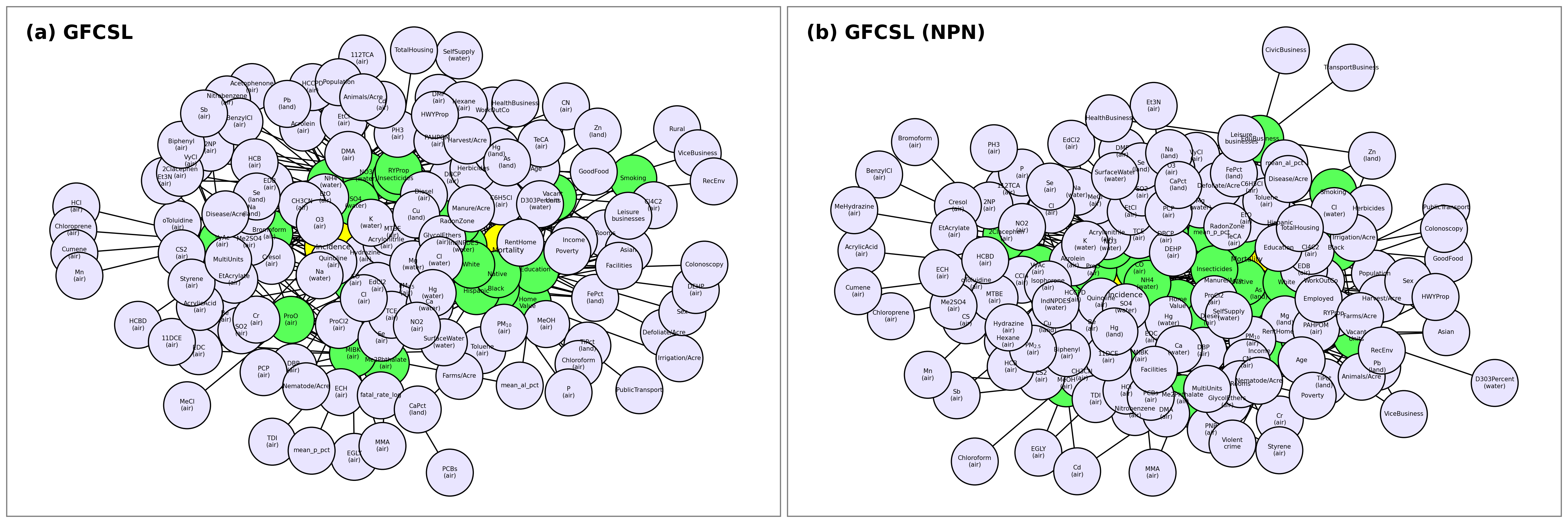}
\end{adjustbox}
\caption{\textit{County-level cancer results using GFCSL}. Estimated local graph of radius 2 using \textbf{(a)} original data and \textbf{(b)} nonparanormal data.}
\label{fig:eqi_gfcsl}
\end{figure*}
}{}

\ifthenelse{\boolean{showfigures}}{
\begin{figure*}[ht!]
\begin{adjustbox}{center}
\includegraphics[
	height=0.35\textheight, 
	width=0.85\textwidth,
	trim={0pt 0pt 0pt 0pt}, clip]
	{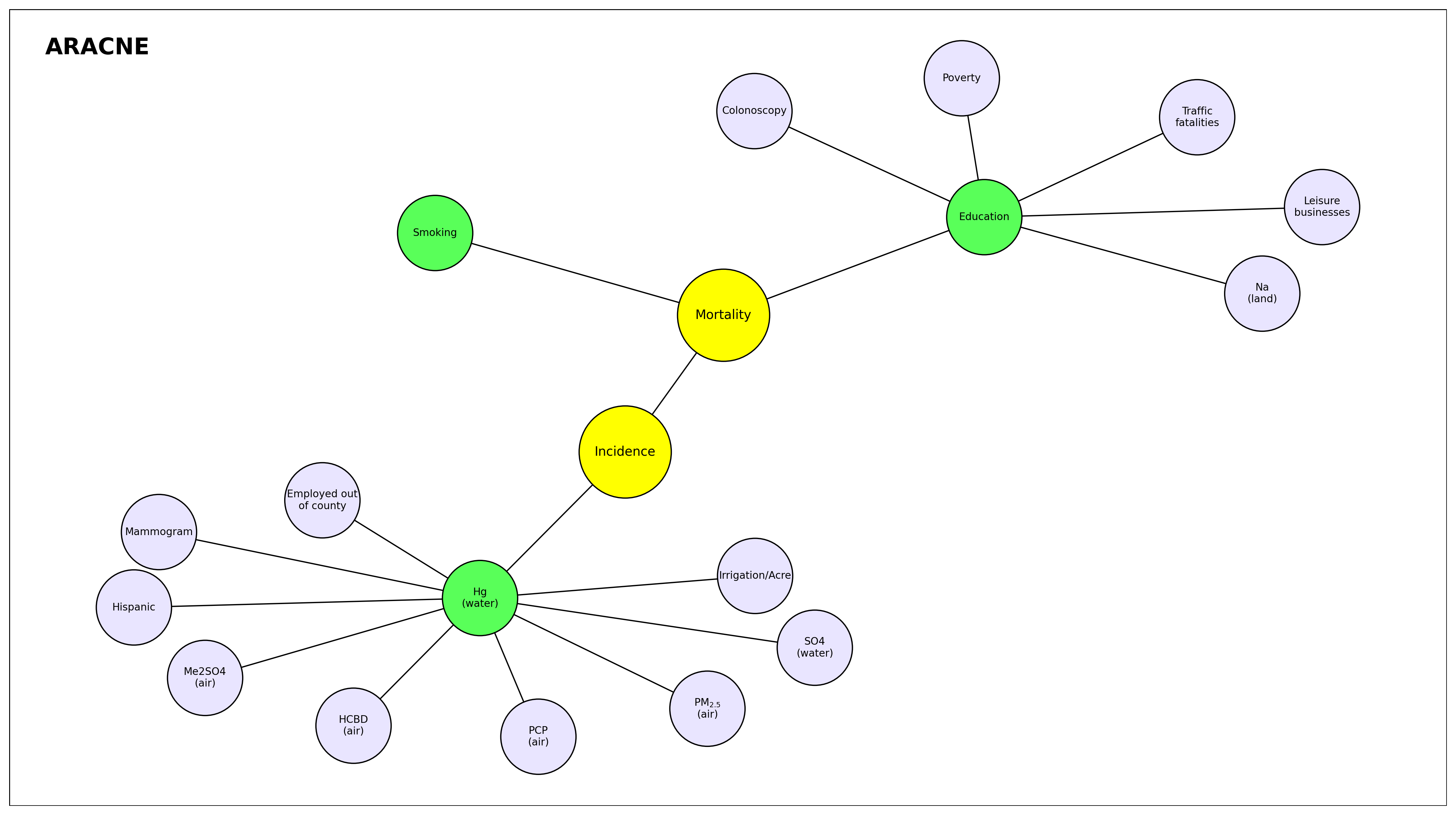}
\end{adjustbox}
\caption{\textit{County-level cancer results using ARACNE}. Estimated local graph of radius 2 using original data (applying the nonparanormal did not change results).}
\label{fig:eqi_aracne}
\end{figure*}
}{}

\ifthenelse{\boolean{showfigures}}{
\begin{figure*}[ht!]
\begin{adjustbox}{center}
\includegraphics[
	height=0.35\textheight, 
	width=0.85\textwidth,
	trim={0pt 0pt 0pt 0pt}, clip]
	{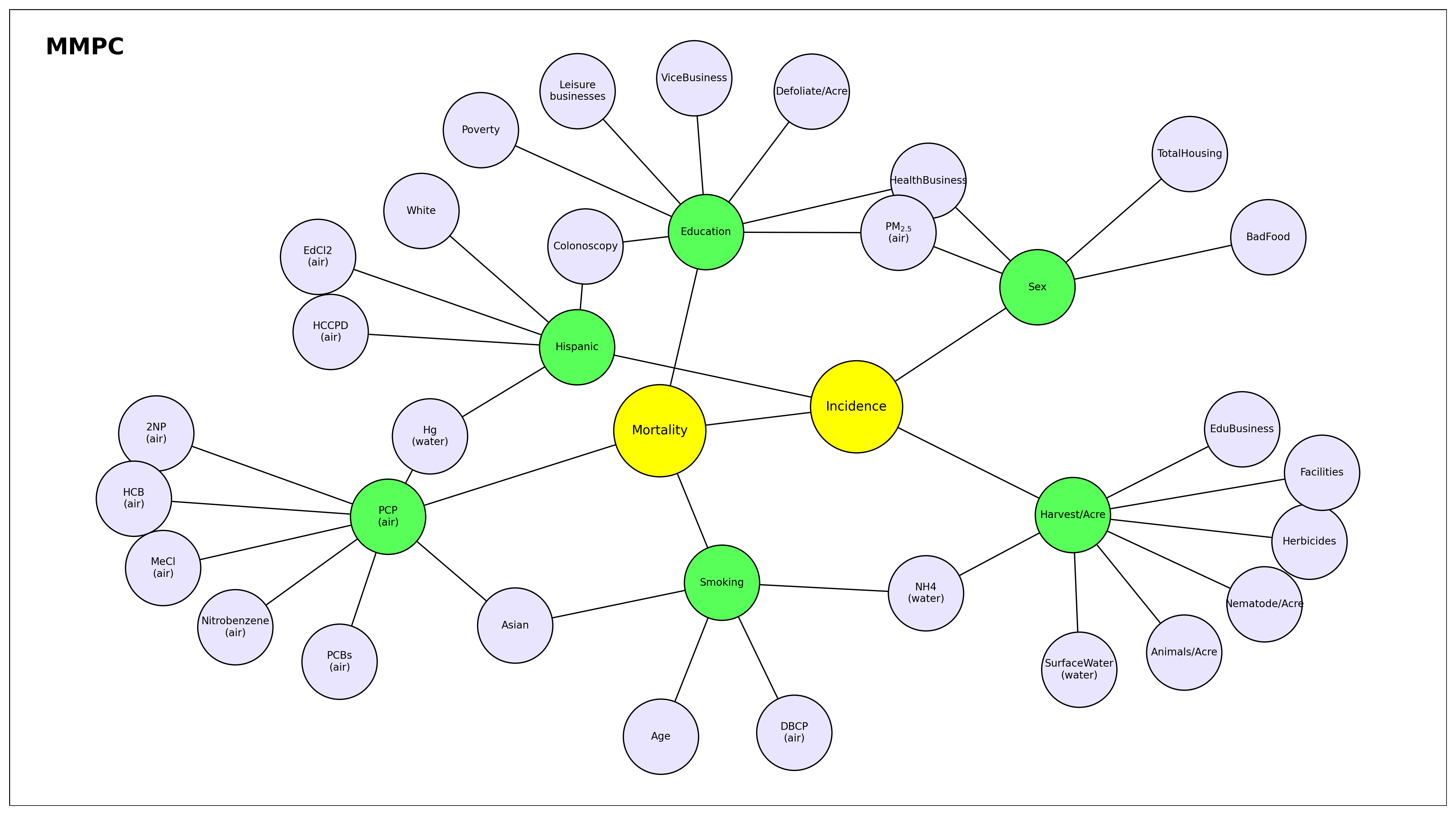}
\end{adjustbox}
\caption{\textit{County-level cancer results using MMPC}. Estimated local graph of radius 2 using original data (applying the nonparanormal did not change results).}
\label{fig:eqi_mmpc}
\end{figure*}
}{}

\ifthenelse{\boolean{showfigures}}{
\begin{figure*}[ht!]
\begin{adjustbox}{center}
\includegraphics[
	height=0.35\textheight, 
	width=0.85\textwidth,
	trim={0pt 0pt 0pt 0pt}, clip]
	{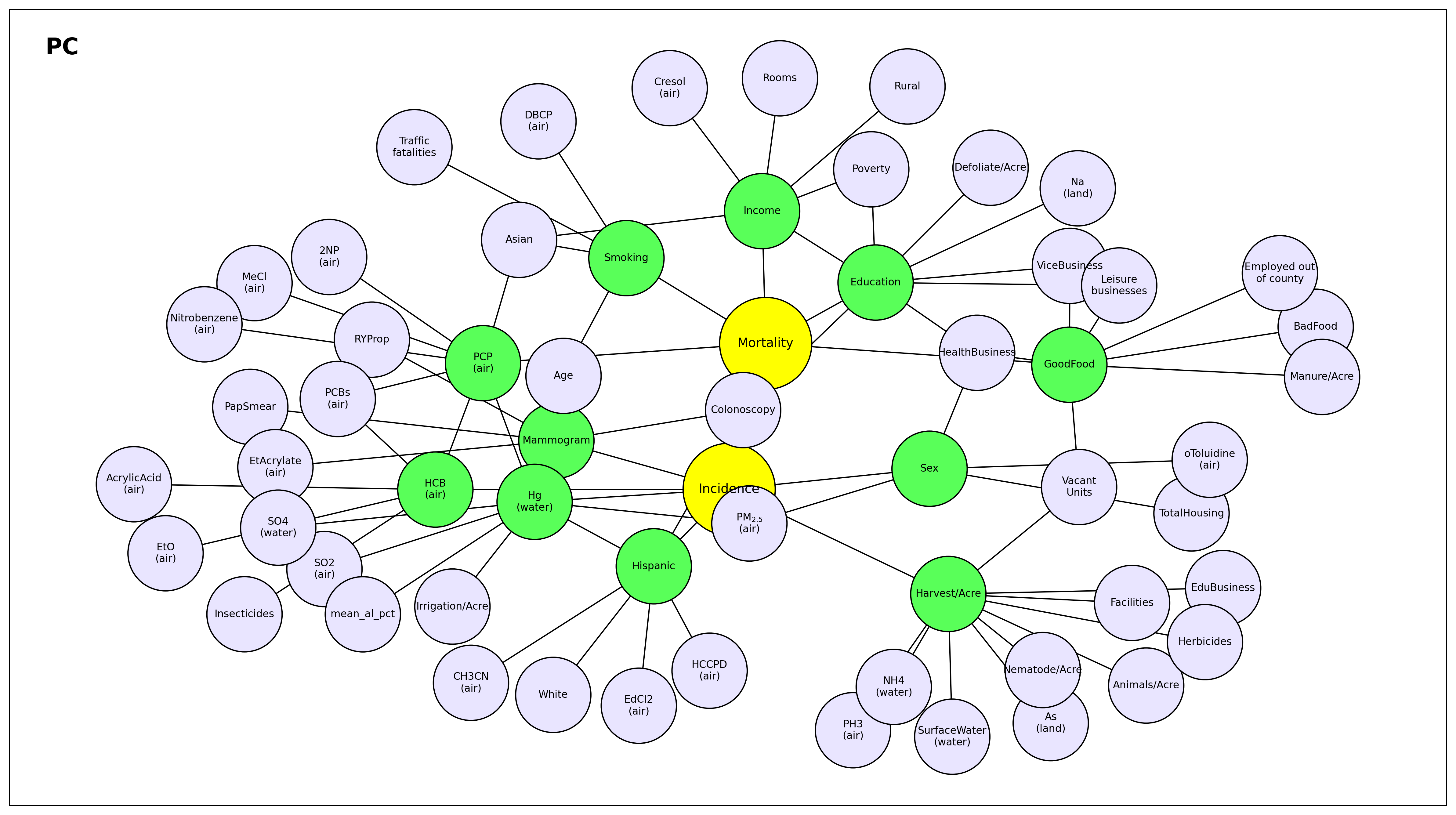}
\end{adjustbox}
\caption{\textit{County-level cancer results using StablePC}. Estimated local graph of radius 2 using original data (applying the nonparanormal did not change results).}
\label{fig:eqi_pc_stable}
\end{figure*}
}{}

\clearpage


\subsection{Breast cancer}\label{sup_sec:other_methods_bc}


For the breast cancer data ($n=547$, $p=10{,}744$), graphical lasso and nodewise lasso returned graphs in which all three of the target variables were isolated (no neighbors). Among the methods with global FDR control, only DSNWSL and GFCSL were computationally feasible (these ran in $4.5$ and $5.25$ hours, respectively; by contrast, PFS ran in $50$ minutes on the same data). We ran these two methods at multiple FDR thresholds ranging from $0.01$ to $0.5$ (the \texttt{SILGGM} defaults are $0.05$ and $0.1$). Since the connection between pathologic stage and status has strong clinical support, results are reported for the smallest target FDR at which an edge is present between pathologic stage and status. GFCSL first detects the edge between pathologic stage and status at a target FDR of $0.1$ for the original data, and $0.15$ for the nonparanormal data. Figures \ref{fig:bc_gfcsl}a and \ref{fig:bc_gfcsl}b show that these target FDR values yield interpretable subgraphs of radius $1$ around the target variables. However, the subgraphs of radius 2 (Figures \ref{fig:bc_gfcsl}c and \ref{fig:bc_gfcsl}d) are extremely dense and uninterpretable. DSNWSL does not detect the edge between pathologic stage and status for any target FDR at or below $0.5$ (\cref{fig:bc_dsnwsl}a) on the original data, and detects it only at a target FDR of $0.5$ for the nonparanormal data (\cref{fig:bc_dsnwsl}b). Among methods from the \texttt{bnlearn} package, ARACNE was the only approach that was computationally feasible for this dataset. 

\ifthenelse{\boolean{showfigures}}{
\begin{figure*}[ht!]
\begin{adjustbox}{center}
\includegraphics[
	height=0.425\textheight, 
	width=\textwidth,
	trim={0pt 0pt 0pt 0pt}, clip]
	{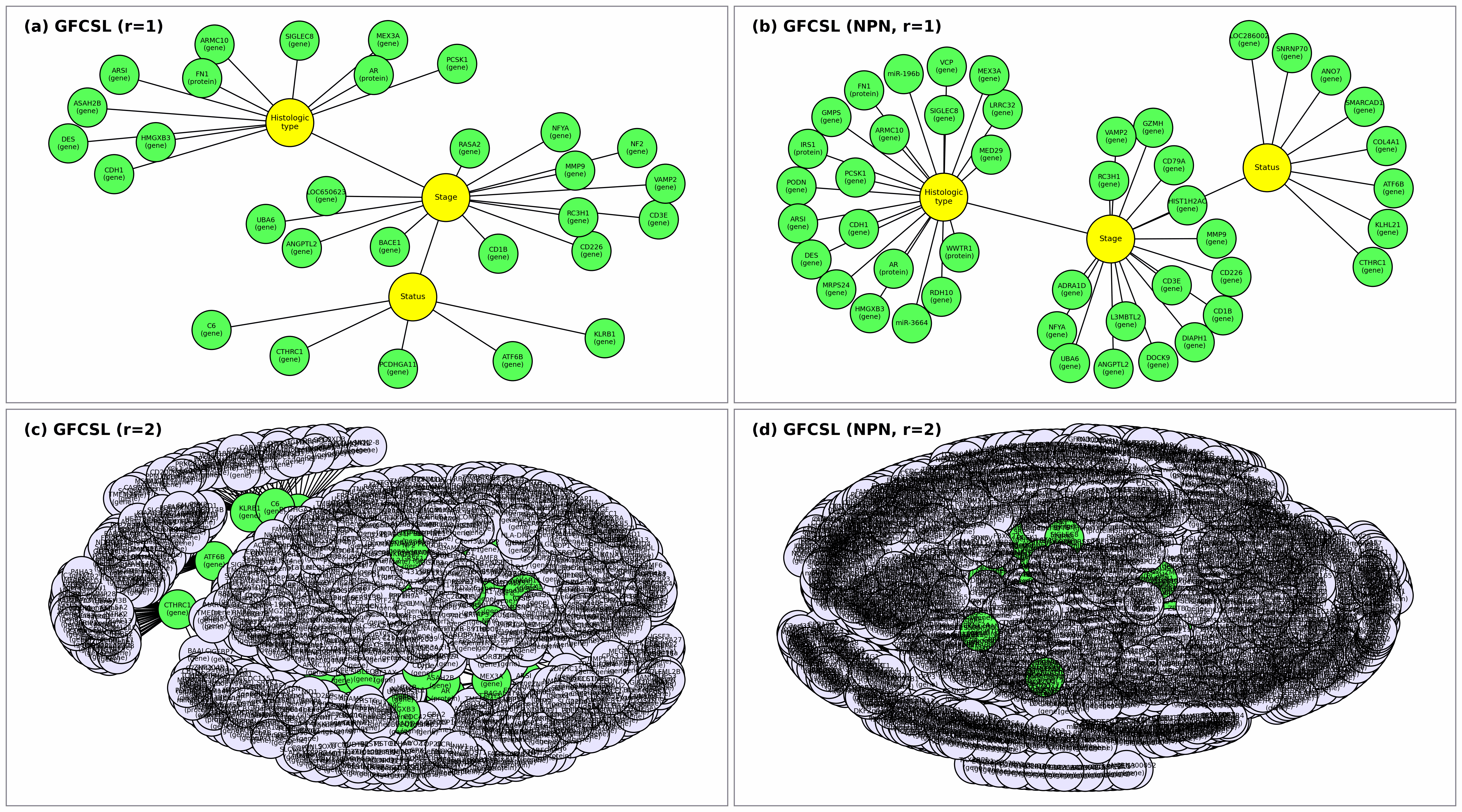}
\end{adjustbox}
\caption{\textit{Breast cancer results using GFCSL}. Estimated local graph of radius 1 using \textbf{(a)} original data and \textbf{(b)} nonparanormal data. Estimated local graph of radius 2 using \textbf{(c)} original data and \textbf{(d)} nonparanormal data.}
\label{fig:bc_gfcsl}
\end{figure*}
}{}

\ifthenelse{\boolean{showfigures}}{
\begin{figure*}[ht!]
\begin{adjustbox}{center}
\includegraphics[
	height=0.235\textheight, 
	width=\textwidth,
	trim={0pt 0pt 0pt 0pt}, clip]
	{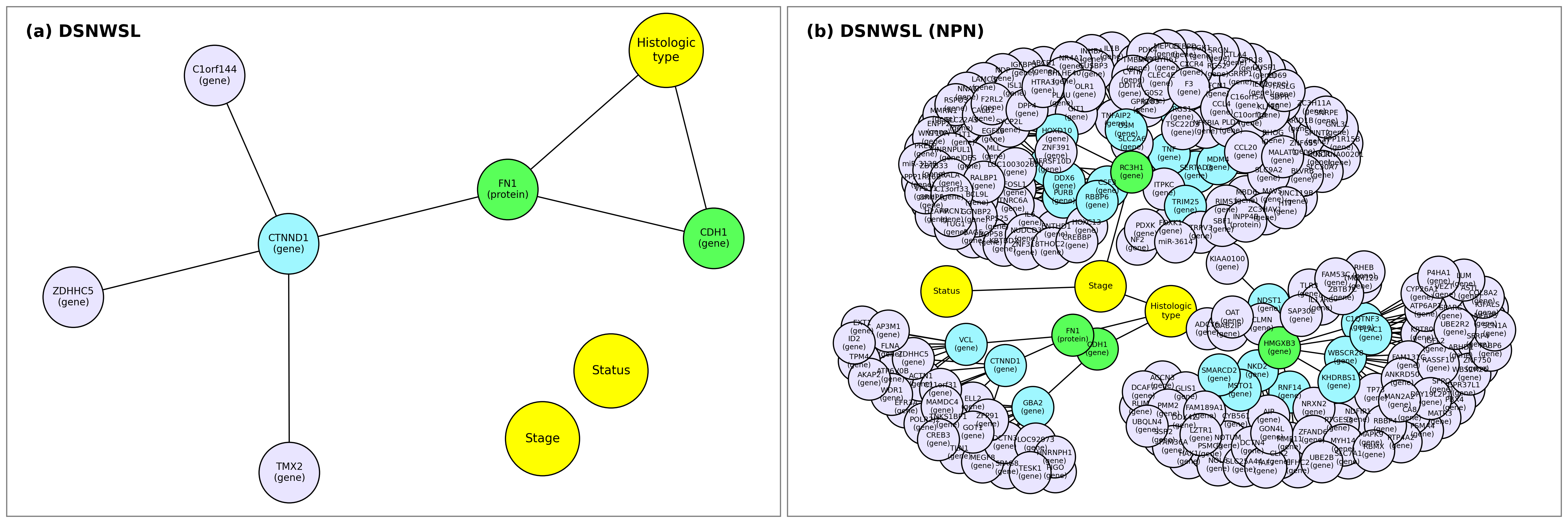}
\end{adjustbox}
\caption{\textit{Breast cancer results using DSNWSL}. Estimated local graph of radius 3 using \textbf{(a)} original data and \textbf{(b)} nonparanormal data.}
\label{fig:bc_dsnwsl}
\end{figure*}
}{}

\ifthenelse{\boolean{showfigures}}{
\begin{figure*}[ht!]
\begin{adjustbox}{center}
\includegraphics[
	height=0.35\textheight, 
	width=0.825\textwidth,
	trim={0pt 0pt 0pt 0pt}, clip]
	{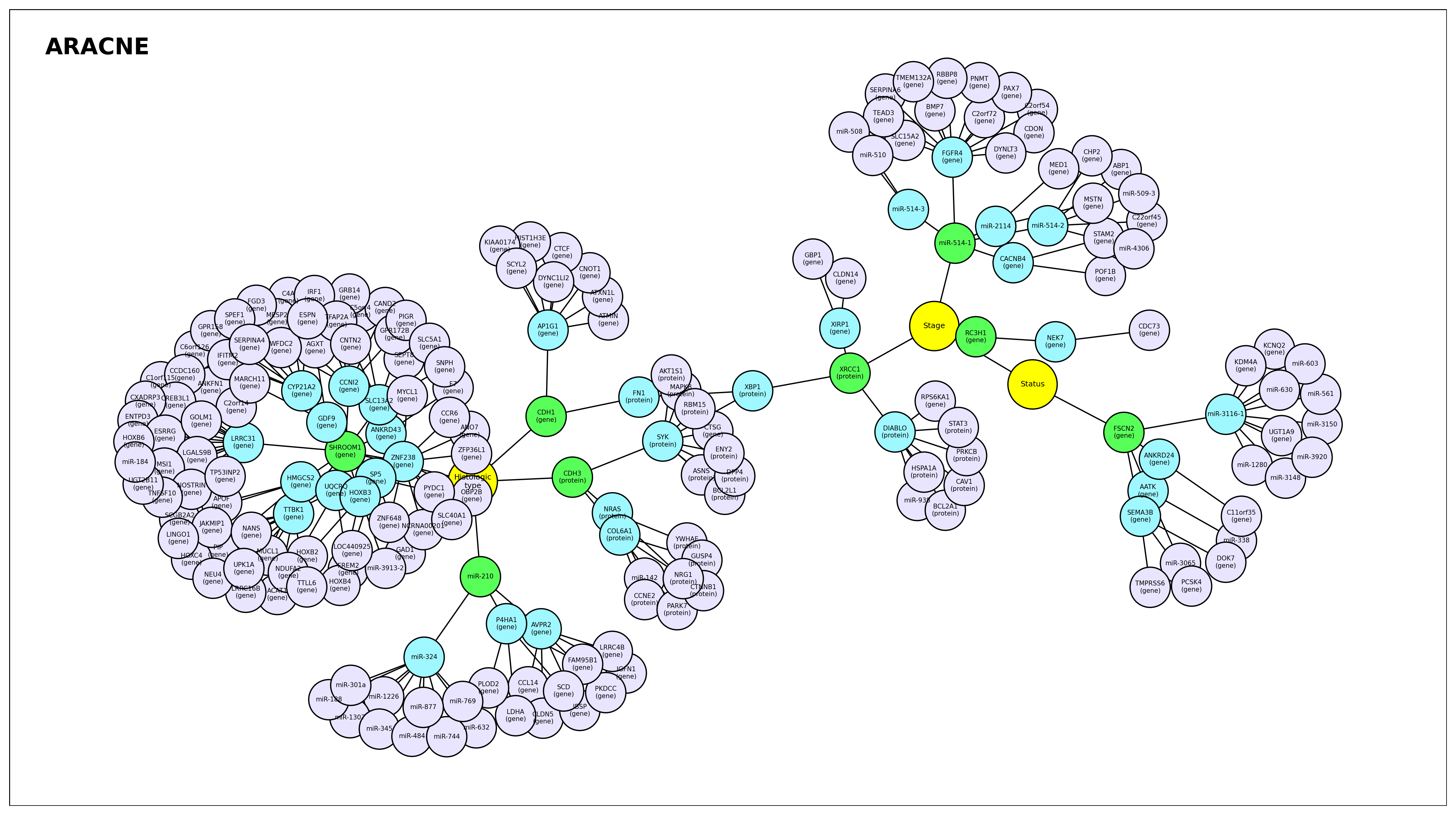}
\end{adjustbox}
\caption{\textit{Breast cancer results using ARACNE}. Estimated local graph of radius 3 using original data (applying the nonparanormal did not change results).}
\label{fig:bc_aracne}
\end{figure*}
}{}

\clearpage


\subsection{Human connectome and cognition}\label{sup_sec:other_methods_hcp}


In addition to PFS, we applied several other graph estimation methods to the HCP dataset to estimate local graphs of radius $3$ around fluid cognition. The two regularization-based approaches, graphical lasso and GFCSL, were run with default settings. The stable PC algorithm was run with significance level $\alpha = 0.05$ using mutual information. For local neighborhood estimation, we applied nodewise versions of several constraint-based methods: HPC with $\alpha = 0.05$, IAMB with $\alpha = 0.005$, and MMPC with $\alpha = 0.025$. In each case, neighborhoods were estimated for the fluid cognition variable and expanded to radius $3$. All conditional independence tests in these methods used mutual information.

\ifthenelse{\boolean{showfigures}}{
\begin{figure*}[ht!]
\begin{adjustbox}{center}
\includegraphics[
	height=0.425\textheight, 
	width=1.1\textwidth,
	trim={0pt 0pt 0pt 0pt}, clip]
	{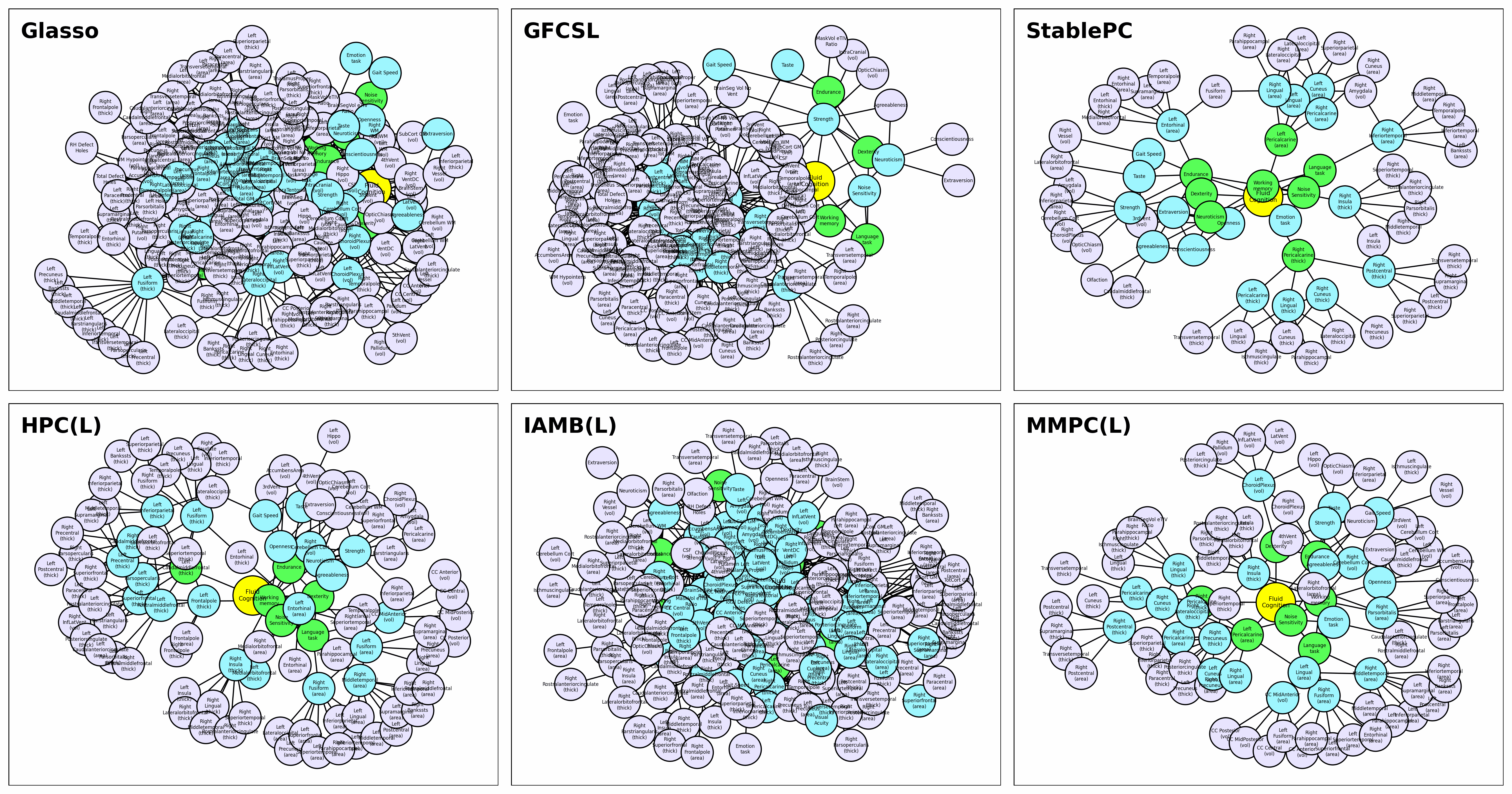}
\end{adjustbox}
\caption{\textit{Human Connectome Project (HCP) results using other methods}. Estimated local graphs of radius 3 around fluid cognition (yellow) using various graph estimation methods. Only methods that identified at least one edge between Fluid Cognition and a FreeSurfer-derived structural brain morphometry feature are shown.}
\label{fig:hcp_other_methods}
\end{figure*}
}{}

\clearpage


\renewcommand{\refname}{Supplementary references}
\putbib

\end{bibunit}


\end{document}